\PassOptionsToPackage{usenames,dvipsnames}{xcolor}

\newif\ifdraft\draftfalse
\makeatletter \@input{localoptions} \makeatother

\documentclass{llncs}

\usepackage[utf8x]{inputenc}
\usepackage{subcaption}
\usepackage{multicol}
\usepackage{mathpartir} % not supported
\usepackage{mathtools}  % not supported
\usepackage{stmaryrd}
\usepackage{amssymb}
\usepackage{color}
\usepackage[inline]{enumitem}
\usepackage{xspace}
\usepackage{thm-restate}
\usepackage{multirow}
\usepackage{listings}
\usepackage{mathtools}
\usepackage{wasysym}
\usepackage{hyperref}
\usepackage{xcolor}
\usepackage{bm}
\usepackage{xifthen}% provides \isempty test

\makeatletter
\@input{cmds} % all commands should go in this file
\makeatother

\begin{document}

\title{Operationally-based Program Equivalence Proofs using LCTRSs}

\author{Ștefan Ciobâcă\inst{1} \and
Dorel Lucanu\inst{1} \and
Andrei Sebastian Buruiană\inst{2}}

\institute{Alexandru Ioan Cuza University \\
\email{\{stefan.ciobaca,dlucanu\}@info.uaic.ro} \and
Bitdefender\\
\email{sburuiana@bitdefender.com}}

\maketitle

\begin{abstract}
  We propose an operationally-based deductive proof method for program
  equivalence. It is based on encoding the language semantics as
  logically constrained term rewriting systems (LCTRSs) and the two
  programs as terms.
  The main feature of our method is its flexibility. We illustrate
  this flexibility in two applications, which are novel.

  For the first application, we show how to encode low-level details
  such as stack size in the language semantics and how to prove
  equivalence between two programs operating at different levels of
  abstraction. For our running example, we show how our method can
  prove equivalence between a recursive function operating with an
  unbounded stack and its tail-recursive optimized version operating
  with a bounded stack. This type of equivalence checking can be used
  to ensure that new, undesirable behavior is not introduced by a more
  concrete level of abstraction.

  For the second application, we show how to formalize read-sets and
  write-sets of symbolic expressions and statements by extending the
  operational semantics in a conservative way. This enables the
  relational verification of program schemas, which we exploit to
  prove correctness of compiler optimizations, some of which cannot be
  proven by existing tools.

  Our method requires an extension of standard LCTRSs with axiomatized
  symbols. We also present a prototype implementation that proves the
  feasibility of both applications that we propose.
\end{abstract}
\section{Introduction}
\label{sec:intro}
A typical transformation in optimizing recursive functions is to add
an additional parameter called an accumulator, which holds the current
result of the computation. The transformed function is usually
tail-recursive, enabling the compiler to emit efficient code.
Typically, the optimized version of a function (the tail-recursive
version) is simply \emph{assumed} to be functionally equivalent to the
original function. However, as we show here, this is not necessarily
the case. Consider the two functions for computing the sum of the
first \(n\) positive naturals, presented in a C-like language, in
Figure~\ref{fig:fF}.

\begin{figure}[t]
  \hrule
  \begin{subfigure}{0.45\textwidth}
\begin{verbatim}
int f(int n) {
  if (n == 0) {
    return 0;
  } else {
    return n + f(n - 1);
  }
}
\end{verbatim}
  \end{subfigure}
  \begin{subfigure}{0.45\textwidth}
\begin{verbatim}
int F(int n, int i, int a) {
  if (i > n) {
    return a;
  } else {
    return F(n, i + 1, a + i);
  }
}
\end{verbatim}
  \end{subfigure}
  \caption{\label{fig:fF} Two programs computing the sum of the first
    \(n\) positive naturals. The program on the right uses an
    accumulator and is tail-recursive.}
  \hrule
\end{figure}

The programs \texttt{f(n)} and \texttt{F(n, 0, 0)} are functionally
equivalent in an idealized setting. However, depending on the exact
definition of functional equivalence, the equivalence may not hold. We
illustrate two settings where the equivalence does not hold:
\begin{enumerate}[wide, labelwidth=!, labelindent=0pt]
\item[\emph{Setting 1}.] Consider that the two programs have a bounded
  stack. The program \(f\) (the left-hand side, \emph{lhs}) uses \(O(n)\)
  stack cells, while the program \(F\) (the right-hand side, \emph{rhs})
  can use constant stack size, since it is tail-recursive.

  In this model of computation, with a bounded stack, which is more
  realistic, the results of the function calls would be different for
  a sufficiently large input \(n\): the program on the left-hand side
  would crash (running out of stack size), while the program on the
  right-hand side would work as expected.

  We have confirmed this difference between the two programs on a real
  system. The first program (compiled on a typical Windows laptop with
  a recent {\tt g++} compiler, without optimizations) has a stack
  overflow when \(n \geq 43340\). The second program (compiled with
  tail-call optimizations) exhibits no stack overflow (even for larger
  values of \(n\)) on the same system.
\item[\emph{Setting 2}.] Also consider a variation of the two programs
  presented in Figure~\ref{fig:fFvariation}.  Even with an unbounded
  stack, if the language has introspection capabilities that allow
  programs to query the current stack size (the function {\tt
    stack\_size}), then the \emph{rhs} and the \emph{lhs} behave
  differently, since the stack size will be large only in the
  \emph{lhs}. Therefore \texttt{f(n)} and \texttt{F(n, 0, 0)} are not
  equivalent in this setting.
  We have also confirmed this difference between the two programs on a
  real system as well. We did this by developing a non-portable
  implementation of {\tt stack\_size()} (for the X86 architecture, by
  querying the {\tt RSP} register). The first program produces an
  error for sufficiently large \(n\), while the second program does not,
  for any \(n\).
\end{enumerate}
Therefore, telling whether the optimized tail-recursive version of a
function is equivalent to the original function is worth investigating
in a more principled manner.

\emph{Our solution.} We propose a method for proving program
equivalence based on modeling the operational semantics of the
language as a logically constrained term rewriting system
(LCTRS). This method allows us to compare two programs for equivalence
in various settings, by varying the underlying semantics defined as an
LCTRS.
\begin{figure}[t]
  \hrule
  \begin{subfigure}{0.45\textwidth}
\begin{verbatim}
int f(int n) {
  if (n == 0) {
    if (stack_size() > 10) {
      error;
    }
    return 0;
  } else {
    return n + f(n - 1);
  }
}
\end{verbatim}
  \end{subfigure}
  \begin{subfigure}{0.45\textwidth}
\begin{verbatim}
int F(int n, int i, int a) {
  if (i > n) {
    if (stack_size() > 10) {
      error;
    }
    return a;
  } else {
    return F(n, i + 1, a + i);
  }
}
\end{verbatim}
  \end{subfigure}
  \caption{\label{fig:fFvariation} A variation of the programs in
    Figure~\ref{fig:fF}. The only difference is in the base case.}
  \hrule
\end{figure}
We study the two programs in the running example above using as
operational semantics an imperative language featuring integer
variables, boolean conditions, if-then-else and while statements, and
function calls. We call the language \IMP{} and we introduce it
formally in the subsequent sections. We propose two versions of \IMP{}
with the same syntax but with different semantics:
\begin{enumerate*}%[wide, labelwidth=!, labelindent=0pt]
  
\item[\IMP{1}] has an idealized semantics, with an unbounded stack;

\item[\IMP{2}] has a more realistic semantics, with a bounded stack.
  
\end{enumerate*}

Our method proves that the two programs in Figure~\ref{fig:fF} are
equivalent in \IMP1, but the equivalence proof correctly fails in
\IMP2. Our method also shows that they are equivalent when the first
program is interpreted in \IMP1 and the second program in \IMP2. When
the two programs query the stack size as in
Figure~\ref{fig:fFvariation}, the equivalence proof correctly fails in
both \IMP1 and \IMP2. We write \IMP{} in the cases where the exact
version, \IMP1 or \IMP2, does not matter.

We encode the operational semantics of the language as a logically
constrained term rewriting system and the two programs as terms. An
LCTRS consists of rewrite rules of the form \(\rrule{l}{r}{\phi},\)
where \(l, r\) are terms and \(\phi\) is a first-order logical
constraint. In \IMP{}, \(l\) and \(r\) are terms of sort \(\Cfg\), 
representing program configurations. \IMP{} configurations are tuples
\(\cfgimp{[\msh{e_1}, \msh{\ldots}, \msh{e_n}]}{\menv}{\mfuncs}\) of:
\begin{enumerate}
\item a cons-list \(\fsh{[ \msh{e_1, \ldots, e_n} ]}\) of expressions
  and statements to be evaluated in order, representing the evaluation
  stack,
\item an environment \(\menv\) mapping identifiers to their value,
\item and an environment \(\mfuncs\) mapping function identifiers to the
  function bodies.
\end{enumerate}
We use a Haskell-like notation for cons-lists: \(\fsh{[]}\) is the empty
list, \(\lisymb\) is the (right-associative) list constructor, and
\(\fsh{[\msh{e_1, e_2, \ldots, e_n}]}\) is a shorthand for
\(e_1 \lisymb e_2 \lisymb \ldots \lisymb e_n \lisymb \fsh{[]}\). The
full order-sorted algebra defining the syntax of \IMP{} is given in
BNF-like notation in Figure~\ref{fig:impsyntax}. The operational
semantics of \IMP{} is given by logically constrained rewrite rules
such as the following, which define assignments and summations:
\begin{enumerate}

\item \(\cfgimp{\laassign{\vx}{\vi} \lisymb \vcs}{\venv}{\vfuncs}
  \rewrite
  \cfgimp{\vcs}{\bupdate(\venv, \vx, \vi)}{\vfuncs}\);
\item \(\cfgimp{\laassign{\vx}{\ve} \lisymb \vcs}{\venv}{\vfuncs}
  \rewrite
  \cfgimp{\ve \lisymb \laassign{\vx}{\square{}} \lisymb \vcs}{\venv}{\vfuncs}
  \myif \lnot \bval(\ve)\);
\item \(\cfgimp{\vx \lisymb \vcs}{\venv}{\vfuncs}
  \rewrite
  \cfgimp{\blookup(\venv, \vx) \lisymb \vcs}{\venv}{\vfuncs}\);
\item \(\cfgimp{\plus{\vesub1}{\vesub2} \lisymb \vcs}{\venv}{\vfuncs}
  \rewrite
  \cfgimp{\vesub1 \lisymb \plus{\square{}}{\vesub2} \lisymb \vcs}{\venv}{\vfuncs}
  \myif \lnot \bval(\vesub1)\);
\item \(\cfgimp{\visub1 \lisymb \plus{\square{}}{\vesub2} \lisymb \vcs}{\venv}{\vfuncs}
  \rewrite
  \cfgimp{\plus{\visub1}{\vesub2} \lisymb \vcs}{\venv}{\vfuncs}\);
\item \(\cfgimp{\plus{\visub1}{\vesub2} \lisymb \vcs}{\venv}{\vfuncs}
  \rewrite
  \cfgimp{\vesub2 \lisymb \plus{\visub1}{\square{}} \lisymb \vcs}{\venv}{\vfuncs}
  \myif \lnot \bval(\vesub2)\);
\item
  \(\cfgimp{\visub2 \lisymb \plus{\visub1}{\square{}} \lisymb \vcs}{\venv}{\vfuncs}
  \rewrite
  \cfgimp{\plus{\visub1}{\visub2} \lisymb \vcs}{\venv}{\vfuncs}\);
\item \( \cfgimp{\plus{\visub1}{\visub2} \lisymb \vcs}{\venv}{\vfuncs}
  \rewrite
  \cfgimp{\visub1 + \visub2 \lisymb \vcs}{\venv}{\vfuncs} \);
\item
  \(\cfgimp{\vi \lisymb \laassign{\vx}{\square{}} \lisymb \vcs}{\venv}{\vfuncs}
  \rewrite
  \cfgimp{\laassign{\vx}{\vi} \lisymb \vcs}{\venv}{\vfuncs}\).
\end{enumerate}
We use the following typographic conventions:
\begin{enumerate}
\item standard math font is used for meta-variables (e.g., \(l, r\)
  standing for terms and \(\phi\) standing for constraints),
\item \vsh{sans-serif,\ red\ font} for object-level variables (e.g.,
  the variable symbol \( \vx \), standing for program identifiers) and
\item \fsh{bold,\ blue\ font} for object-level non-variable symbols
  (e.g., the function symbols \( \bupdate, \bval \)).
\end{enumerate} 
The first rule defines the semantics of assigning an integer \(\vi\)
to the program identifier \(\vx\) (the variables \(\vi\) and \(\vx\)
are of sorts \(\sint\)eger and \(\sid\)entifier, respectively). The
variable \( \vcs \) matches the tail of the computation stack (we use
the Haskell convention of using the suffix \(\vsh{-s}\) to denote a
list). In this case, the environment is updated by using the
\(\bupdate\) function (in the usual theory of arrays). The second rule
defines assignment in the case where an expression \(\ve\), which is
not a value (i.e., not an integer) is assigned to the program
identifier \(\vx\). In this case, the expression \(\ve\) is
\emph{scheduled} for evaluation, by placing it in front of the
computation stack. The special constant \( \fsh{\square{}} \) is used
as a placeholder to recall which part of the statement has been
promoted for evaluation. The third rule defines the evaluation rules
for program identifiers (program variables), by using the \(\blookup\)
function in the theory of arrays. The following five rules define how
additions are evaluated (left to right). Note that
\(\ve, \vesub1, \vesub2\) are variables of sort \(\sexp\), while
\(\vi, \visub1, \visub2\) are variables of sort \(\sint\) (as
\(\sint < \sexp\), any term of sort \(\sint\) is also a term of sort
\(\sexp\), but not vice-versa). Once an expression promoted by the
second rule is completely evaluated and \emph{becomes} a term of sort
\(\sint\), the last rule is allowed to fire, which places the value of
the expression back into the assignment statement, which may then
proceed to execute using rule \(1\). We use \( \fsh{\plus{}{}} \)
(using teletype font) for the summation operator in the language
(\( \fsh{\plus{}{}} \) is a constructor for program expressions), and
\( \fsh{+} \) (usual mathmode font) for integer summation (\(\fsh{+}\)
is an interpreted function symbol: summation in the usual theory of
integers).

Here is how \(\fsh{\laassign{\cx}{\plus{\cx}{2}}}\) is executed in an
initial environment \(\menv = \fsh{\cx \mapsto 12}\) and with an
arbitrary map \(\mfuncs\):
\(\cfgimp{[\laassign{\cx}{\plus{\cx}{2}}]}{\menv}{\mfuncs}\)
\(\rewrite\)
\(\cfgimp{[\plus{\cx}{2}, \laassign{\cx}{\square{}}]}{\allowbreak\menv}{\mfuncs}\) \(\rewrite\)
\(\cfgimp{[\cx, \plus{\square{}}{2},
  \laassign{\cx}{\square{}}]}{\allowbreak\menv\allowbreak}{\allowbreak\mfuncs}\)
\(\rewrite\)
\(\cfgimp{[12, \plus{\square{}}{2}, \allowbreak \laassign{\cx}{\square{}}]}{\allowbreak\menv}{\mfuncs}\) \(\rewrite\)
\(\fsh{\langle[\plus{12}{2}, \laassign{\cx}{\square{}}],} \allowbreak\menv\fsh{, \mfuncs \rangle}\)
\(\rewrite\)
\(\cfgimp{[14, \laassign{\cx}{\square{}}]}{\allowbreak\menv}{\mfuncs}\) \(\rewrite\)
\(\cfgimp{[\laassign{\cx}{14}]}{\menv}{\mfuncs}\)
\(\rewrite\)
\(\cfgimp{[]}{\cx \mapsto 14}{\mfuncs}\) \(\notrewrite\).
This style of giving an operational semantics to a language is called
\emph{frame stack style} according to~\cite{PittsAS2000}, was
popularized by the K framework~\cite{StefanescuOOPSLA2016} and it
avoids the necessity of refocusing~\cite{Danvy2004} typical of
operational semantics based on evaluation contexts.

Recall the two recursive programs \(f\) and \(F\) introduced earlier
(Figure~\ref{fig:fF}). Formally,
\(f = \cfgimp{\fsh{f(\vn)}}{\menv}{\mfuncs}\) and
\(F = \cfgimp{\fsh{F(\vn, 0, 0)}}{\menv}{\mfuncs},\) where \menv is
some environment and \(\mfuncs = \fsh{\{\ \pvar{f} \mapsto \lambda \pvar{n}.
  \laite{\pvar{n} = 0}{0}{\plus{\pvar{n}}{\lacall{\pvar{f}(\pvar{n} -
      1)}}}}, \ \fsh{\pvar{F} \mapsto \lambda \pvar{n}.}\allowbreak\fsh{\lambda
  \pvar{i} . \lambda \pvar{a} .}\allowbreak \fsh{\laite{\pvar{i} \leq
    \pvar{n}}{\lacall{F(n,\plus{i}{1}, \plus{a}{i})}}{\pvar{a}}}\}\)
 is a map from identifiers to function
bodies defining \fsh{f} and \fsh{F}. The fact that the two programs
share the same input is represented by having the same free variable
\(\vn\) in both symbolic configurations. We now show the main
difficulty in an operational semantics-based proof of equivalence
between them. To illustrate the difficulty, we use the operational
behaviors of the two programs on the input \(\fsh{3}\), shown in
Figure~\ref{fig:phases}.
\begin{figure}[t]
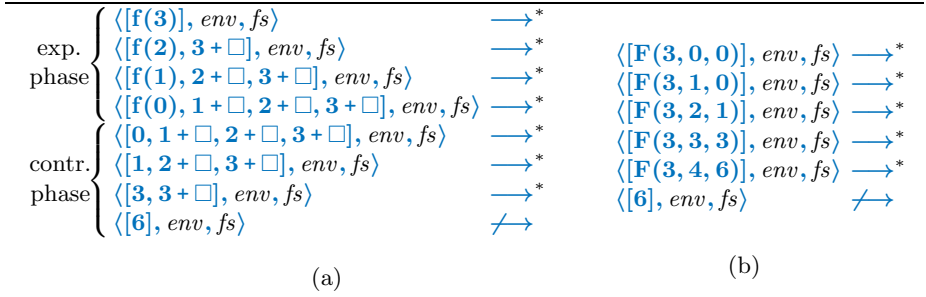

  \hrule
  \begin{subfigure}{0.70\textwidth}
    \[{\hspace{-1cm}}\begin{array}{ll}
        \mbox{\(\begin{array}{c}\mbox{exp.}\\ \mbox{phase}\end{array}\)}\hspace{-0.2cm} & \left\{
                                                                              \begin{array}{ll}
                                                                                \cfgimp{[ f(3) ]}{\menv}{\mfuncs} & \rewrite^* \\
                                                                                \cfgimp{[ f(2), \plus{3}{\square{}} ]}{\menv}{\mfuncs} &  \rewrite^* \\
                                                                                \cfgimp{[ f(1), \plus{2}{\square{}}, \plus{3}{\square{}} ]}{\menv}{\mfuncs} &  \rewrite^* \\
                                                                                \cfgimp{[ f(0), \plus{1}{\square{}}, \plus{2}{\square{}}, \plus{3}{\square{}}]}{\menv}{\mfuncs} &  \rewrite^*
                                                                              \end{array} \right. \\
        \mbox{\(\begin{array}{c}\mbox{contr.}\\ \mbox{phase}\end{array}\)}\hspace{-0.2cm} & \left\{
                                                                                \begin{array}{ll}
                                                                                  \mathrlap{\cfgimp{[ 0, \plus{1}{\square{}}, \plus{2}{\square{}}, \plus{3}{\square{}} ]}{\menv}{\mfuncs}}
                                                                                  \phantom{\cfgimp{[ f(0), \plus{1}{\square{}}, \plus{2}{\square{}}, \plus{3}{\square{}}]}{\menv}{\mfuncs}}
                                                                                  & \rewrite^*
                                                                                  \\
                                                                                  \cfgimp{[ 1, \plus{2}{\square{}}, \plus{3}{\square{}} ]}{\menv}{\mfuncs} &  \rewrite^* \\
                                                                                  \cfgimp{[ 3, \plus{3}{\square{}} ]}{\menv}{\mfuncs} &  \rewrite^* \\
                                                                                  \cfgimp{[ 6 ]}{\menv}{\mfuncs} & \notrewrite^{\phantom{*}}
                                                                                \end{array} \right. \\
      \end{array}\]
    \caption{\label{subfig:phasesa}}
  \end{subfigure}\hspace{-0.6cm}
  \begin{subfigure}{0.29\textwidth}
    \[\begin{array}{ll}
        \cfgimp{[ F(3, 0, 0) ]}{\menv}{\mfuncs} & \rewrite^* \\
        \cfgimp{[ F(3, 1, 0) ]}{\menv}{\mfuncs} & \rewrite^* \\
        \cfgimp{[ F(3, 2, 1) ]}{\menv}{\mfuncs} & \rewrite^* \\
        \cfgimp{[ F(3, 3, 3) ]}{\menv}{\mfuncs} & \rewrite^* \\
        \cfgimp{[ F(3, 4, 6) ]}{\menv}{\mfuncs} & \rewrite^* \\
        \cfgimp{[ 6 ]}{\menv}{\mfuncs} & \notrewrite^{\phantom{*}} \\
      \end{array}\]
    \caption{\label{subfig:phasesb}}
  \end{subfigure}
  \caption{\label{fig:phases}Operational steps of \(f(3)\) (in
    Subfigure~\ref{subfig:phasesa}) and of \(F(3, 0, 0)\) (in
    Subfigure~\ref{subfig:phasesb}). Note that \(F\) has a single phase,
    while \(f\) has two distinct phases.}
  \hrule
\end{figure}

Note that \(f\) has two phases in the execution:
\begin{enumerate}
\item the stack
expansion phase, before the recursive function reaches the base case
and
\item the stack compression phase, where the actual
  additions take place.
\end{enumerate}
Unlike this non-tail-recursive version, the tail-recursive function
\(F\) has a single phase, where the second argument (the index
\({\tt i}\)) increases in each step and the accumulator (the third
argument) holds in turn the integers
\(\fsh{0}, \fsh{0 + 1}, \fsh{0 + 1 + 2}, \) and
\(\fsh{0 + 1 + 2 + 3}\).

The method that we propose for equivalence proofs is based on two-way
simulation. To prove that \(f\) and \(F\) are equivalent, we show that
the configuration \(\cfgimp{\fsh{f(\vn)}}{\venv}{\mfuncs}\) simulates
\(\cfgimp{\fsh{F(\vn, 0, 0)}}{\venv}{\mfuncs}\) and vice-versa (under
the constraint \(\fsh{\vn \geq 0}\)). To show this, we set up a
relation \(R\) that relates configurations as in the following
diagram:
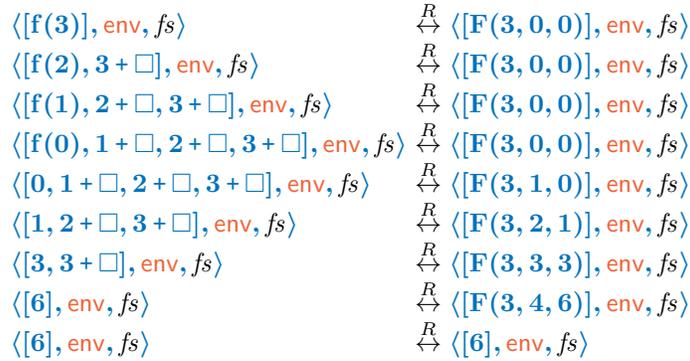
\[\begin{array}{lcl}
    \cfgimp{[ f(3) ]}{\venv}{\mfuncs}
    & \stackrel{R}{\leftrightarrow} &
                                      \cfgimp{[ F(3, 0, 0) ]}{\venv}{\mfuncs} \\

    \cfgimp{[ f(2), 3 \mathop{\fsh{\bf \texttt{+}}} \square{} ]}{\venv}{\mfuncs}
    & \stackrel{R}{\leftrightarrow} &
                                      \cfgimp{[ F(3, 0, 0) ]}{\venv}{\mfuncs} \\

    \cfgimp{[ f(1), 2 \mathop{\fsh{\bf \texttt{+}}} \square{}, 3 \mathop{\fsh{\bf \texttt{+}}} \square{} ]}{\venv}{\mfuncs}
    & \stackrel{R}{\leftrightarrow} &
                                      \cfgimp{[ F(3, 0, 0) ]}{\venv}{\mfuncs} \\
    
    \cfgimp{[ f(0), 1 \mathop{\fsh{\bf \texttt{+}}} \square{}, 2 \mathop{\fsh{\bf \texttt{+}}} \square{}, 3 \mathop{\fsh{\bf \texttt{+}}} \square{}]}{\venv}{\mfuncs}
    & \stackrel{R}{\leftrightarrow} &
                                      \cfgimp{[ F(3, 0, 0) ]}{\venv}{\mfuncs} \\

    \cfgimp{[ 0, 1 \mathop{\fsh{\bf \texttt{+}}} \square{}, 2 \mathop{\fsh{\bf \texttt{+}}} \square{}, 3 \mathop{\fsh{\bf \texttt{+}}} \square{} ]}{\venv}{\mfuncs}
    & \stackrel{R}{\leftrightarrow} &
                                      \cfgimp{[ F(3, 1, 0) ]}{\venv}{\mfuncs} \\
                                          
    \cfgimp{[ 1, 2 \mathop{\fsh{\bf \texttt{+}}} \square{}, 3 \mathop{\fsh{\bf \texttt{+}}} \square{} ]}{\venv}{\mfuncs}
    & \stackrel{R}{\leftrightarrow} &
                                      \cfgimp{[ F(3, 2, 1) ]}{\venv}{\mfuncs} \\

    \cfgimp{[ 3, 3 \mathop{\fsh{\bf \texttt{+}}} \square{} ]}{\venv}{\mfuncs}
    & \stackrel{R}{\leftrightarrow} &
                                      \cfgimp{[ F(3, 3, 3) ]}{\venv}{\mfuncs} \\
    
    \cfgimp{[ 6 ]}{\venv}{\mfuncs}
    & \stackrel{R}{\leftrightarrow} &
                                      \cfgimp{[ F(3, 4, 6) ]}{\venv}{\mfuncs} \\

    \cfgimp{[ 6 ]}{\venv}{\mfuncs}
    & \stackrel{R}{\leftrightarrow} &
                                      \cfgimp{[ 6 ]}{\venv}{\mfuncs}
  \end{array}\]
That is, the relation \(R\) relates:
\begin{enumerate}
\item the configurations in the expansion phase of
  \(\cfgimp{\fsh{f(\vn)}}{\venv}{\mfuncs}\) with the initial
  configuration \(\cfgimp{\fsh{F(\vn, 0, 0)}}{\venv}{\mfuncs}\) and
\item the configurations in contraction phase of
  \(\cfgimp{\fsh{f(\vn)}}{\venv}{\mfuncs}\) with the configurations in
  the single phase of \(\cfgimp{\fsh{F(\vn, 0, 0)}}{\venv}{\mfuncs}\).
\end{enumerate}
More formally, we would like \(R\) to relate configurations of the
form
\[\cfgimp{[ f(\vi), \plus{\mbrack{\vi + 1}}{\square{}}, \plus{\mbrack{\vi +
2}}{\square{}}, \ldots, \plus{\vn}{\square{}} ]}{\venv}{\mfuncs}\] to
\(\cfgimp{F(\vn, 0, 0)}{\venv}{\mfuncs}\) and configurations of the form
\[\cfgimp{[ \vs, \plus{\mbrack{\vi + 1}}{\square{}}, \plus{\mbrack{\vi + 2}}{\square{}},
  \ldots, \plus{\vn}{\square{}} ]}{\venv}{\mfuncs}\] to
\(\cfgimp{F(\vn, \vi + 1, \vs)}{\venv}{\mfuncs}\). However, in order to
even express this relation \(R\), we require a new technical
development in the context of constrained term rewriting systems that
we call \emph{axiomatized symbols}. Axiomatized symbols can be used to
mimic typical bigops in mathematics such as \(\Sigma\), \(\Pi\),
\(\forall\), etc. For our running example, we require a symbol called
\(\areduce\) axiomatized by the following constrained oriented
equations:
\begin{enumerate}
\item \(\fsh{\areduce(\vi, \vn) \mathrel{\rewrite} []\myif \vi > \vn}\);
\item
  \(\fsh{\areduce(\vi, \vn) \mathrel{\rewrite} \mbrack{\plus{\vi}{\square}}
    \lisymb \areduce(\vi + 1, \vn)\myif \vi \leq \vn}\).
\end{enumerate}
That is, \(\fsh{\areduce(\vi, \vn)}\) stands for the informally
presented cons-list
\[\fsh{[ \plus{\vi}{\square{}},} \allowbreak\fsh{\plus{\mbrack{\vi + 1}}{\square{}},}
  \ldots\fsh{, \plus{\vn}{\square{}} ]}.\] By using the axiomatized symbol
\(\areduce\), the relation \(R\) is formally defined as:
\begin{enumerate}
\item
  \((\cfgimp{[f(\vn)]}{\venv}{\mfuncs}, \cfgimp{[F(\vn, 0,
    0)]}{\venv}{\mfuncs}) \in R\mbox{ if }\fsh{0 \leq \vn}\);
\item
  \((\cfgimp{f(\vi) \lisymb \areduce(\vi + 1, \vn)}{\venv}{\mfuncs},
  \cfgimp{F(\vn, 0, 0)}{\venv}{\mfuncs}) \in R \\ \qquad \mbox{if }\fsh{0
  \leq \vi \leq \vn - 1}\);
\item
  \((\cfgimp{\vs \lisymb \areduce(\vi, \vn)}{\venv}{\mfuncs},
  \cfgimp{F(\vn, \vi, \vs)}{\venv}{\mfuncs}) \in R \mbox{ if }\fsh{1
  \leq \vi \leq \vn}\),
\end{enumerate}
\noindent where \(\mfuncs\) is the map defined earlier.
Our algorithm checks whether \(R\) is indeed a (weak) simulation in
the transition system generated by the LCTRS defining the operational
semantics of \IMP{}. To check equivalence of two symbolic program
configurations \(P\) and \(Q\) (e.g.,
\(\cfgimp{f(\vn)}{\msh\ldots}{\msh\ldots}\) and
\(\cfgimp{F(\vn, 0, 0)}{\msh\ldots}{\msh\ldots}\)), we check that there
exists simulations of \(P\) by \(Q\) and vice-versa. In practice, it
is often the case that \(R^{-1}\) works for the reverse direction. Our
proof method allows to conclude that \(f\) is \emph{fully simulated}
by \(F\), but we can only show that \(F\) is \emph{partially
  simulated} by \(f\). This is because it cannot establish that the
termination of \(F\) implies the termination of \emph{both} phases of
\(f\).

As we have already illustrated above, an axiomatized symbol is any
function symbol axiomatized by a set of oriented constrained
equations. Axiomatized symbols are necessary in defining powerful
relations, as shown above, but they can also be used to enable more
powerful specification in LCTRSs.

For example, the symbol \(\bval\) used in the rewrite rules above is
also an axiomatized symbol, and this symbol helps simplify the
presentation of the operational semantics (otherwise, we would have
had to enumerate all cases where expressions are values and
non-values, respectively).

Axiomatized symbols can also simulate a form of higher-order
rewriting. In order to define \IMP{} as an LCTRS, we mix \(\bullet\) an
environment based semantics for the global store (the environment maps
program identifiers to their integer value) and \(\bullet\) a
substitution-based semantics for the function calls. Substitution is
implemented by the symbol \( \asubst \) axiomatized as:
\begin{enumerate}
\item \(\asubst(\vx, \ve, \vx) \mathrel{\rewrite} \ve\),
\item \(\asubst(\vx, \ve, \vy) \mathrel{\rewrite} \vy \myif \fsh{\vx \mathrel{\not=} \vy}\),
\item
  \(\fsh{\asubst(\vx, \ve, \plus{\vesub1}{\vesub2}) \mathrel{\rewrite}
  \plus{\asubst(\vx, \ve, \vesub1)}{\asubst(\vx, \ve, \vesub2)}}\),
\item
  \(\fsh{\asubst(\vx, \ve, \laassign{\vy}{\vesub1}) \mathrel{\rewrite} \laassign{\vy}{\asubst(\vx, \ve, \vesub1)}}\),
  etc.
\end{enumerate}
\noindent Substitution-based function calls are formalized by the
following rule (the function with a parameter \(\vx\) and a body
\(\vfunbody\) is called on the integer argument \(\vi\)):

\noindent \(\bullet\)
\(\cfgimp{\lacall{\lambda \vx.\vfunbody(\vi)} \lisymb
  \vcs}{\venv}{\vfuncs}\)
\(\rewrite\)
\(\cfgimp{\asubst(\vx, \vi, \vfunbody) \lisymb \vcs
}{\venv}{\vfuncs}\).

We propose an algorithm, presented as a sound proof system, for
proving simulation between symbolic program configurations. Two-way
simulation is used to show equivalence. The simulation-checking
algorithm relies on an oracle for the problem of \emph{unification
  modulo axiomatized symbols}, which we formally define in this
paper. We have implemented the equivalence-checking algorithm as a
prototype in the RMT tool. Early work described here was presented,
without being formally published, at the Dagstuhl seminar 18151 and
the PERR 2019 workshop.

\emph{Contributions.}
\begin{enumerate}[wide, labelwidth=!, labelindent=0pt]
\item Our method is the first to allow equivalence checking in the
  case of bounded resources; being operationally-based, it is easy to
  check equivalence in various other settings (bounded versus
  unbounded stack, bounded versus unbounded integers, etc.) by simply
  varying the underlying LCTRS;
\item Unlike other relational logics, our method can easily handle
  structurally unrelated programs;
\item We extend our previous work on LCTRSs~\cite{CiobacaIJCAR2018} by
  adding \emph{axiomatized symbols}, which are critical for expressing
  powerful relations and we identify a new problem in rewriting, that
  of \emph{unification modulo axiomatized symbols}, which is a particular
  type of higher-order unification, and which appears naturally in the
  context of program equivalence;
\item We show that our method can be used to formalize read-sets and
  write-sets of expressions and statements; this enables the
  verification of program schemas, which we take advantage of to prove
  compiler optimizations correct;
\item We implement the proof method in the prototype RMT tool; it can
  prove correctness of optimizations that are out of the reach of
  other verifiers.
\end{enumerate}

\emph{Organization.}
In Section~\ref{sec:lctrss}, we introduce the technical notations and
background on LCTRSs, as well as the newly proposed notion of
axiomatized symbols. In Section~\ref{sec:semantics}, we give the
formal syntax of \IMP{} and its formal semantics as an LCTRS (for both
variations: \IMP1 and \IMP2). Section~\ref{sec:equivalence} contains
the formalization for the definitions of full/partial simulation and
equivalence and Section~\ref{sec:proofsystem} the algorithms for
checking simulations and equivalences. In
Section~\ref{sec:programschemas} we discuss how our method can be used
to formalize read/write-sets and prove compiler optimizations. In
Section~\ref{sec:related} we discuss related work before concluding in
Section~\ref{sec:conclusion}. Appendix~\ref{app:examplesemantics}
presents a complete example of an execution trace in \IMP{}.
Appendix~\ref{app:proofs} contains the
proofs. Appendix~\ref{app:optimization} describes in detail the
optimizations that we prove correct. Appendix~\ref{app:examples}
contains more details on the functional equivalence examples that
prove.
\section{LCTRSs}
\label{sec:lctrss}
We consider a presentation of LCTRSs that we have introduced in our
previous work~\cite{CiobacaIJCAR2018}, which we describe in this
section and we extend with \emph{axiomatized symbols}. We interpret
LCTRSs in a model combining order-sorted terms with builtins such as
integers, booleans, etc. Logical constraints are first-order formulae
interpreted over the fixed model.

We assume an order-sorted signature \(\Sigma = (S, \le, F)\) with the
following properties:

\begin{enumerate}
  
\item the set of sorts, \( S = S^b \cupdot S^c \), is partitioned into a
  set of \emph{builtin sorts} \( S^b \) and a set of \emph{``constructor''
    sorts} \( S^c \);
    % \dl{"constructed sorts"? Sorturile nu construiesc. Ma tenta sa spun "defined sorts", dar am renuntat ... :).
    % O definitie corecta ar fi "inductive sorts", deoarece multimile desemnate sunt definite inductiv.}

\item in the subsorting relation,
  \( {\leq} \subseteq (S^b \cup S^c) \times S^c \), builtin sorts do not
  have subsorts;

\item the set of function symbols,
  \( F = F^b \cupdot F^c \cupdot F^a \), is partitioned into a set of
  \emph{builtin symbols} \( F^b \), a set of \emph{constructor
    symbols} \( F^c \), and a set of \emph{axiomatized symbols}
  \( F^a \). 
  
\end{enumerate}

If a function symbol \( f \in F \) has arity
\( s_1 \times \ldots \times s_n \rightarrow s \), with
\( s_1, \ldots, s_n, s \in S \), we sometimes write
\( f \in F_{s_1 \ldots s_n, s} \). In particular, if \( n = 0 \) then
\( f \in \Sigma_{\varepsilon, s} \) is a constant of sort \( s \).

We assume that no constructor symbol \( f \in F^c \) is of arity
\( s_1 \times \ldots \times s_n \rightarrow s \), where
\( s \in S^b \) (no constructor symbol returns a builtin) and that any
builtin symbol \( f \in F^b \) has arity
\( f : s_1 \times \ldots \times s_n \rightarrow s \), with
\( s_1, \ldots, s_n, s \in S^b \).

We say that \( \Sigma^b = (S^b, F^b) \) is the many-sorted \emph{builtin
  subsignature} of \( \Sigma \). We assume that the set of builtin
sorts includes at least the sort \( \Bool \in S^b \) and that the
builtin signature \( \Sigma^b \) has a model \( M^b \) such that
\( M_{\Bool}^b = \{ \True, \False \} \), where the interpretation of
the boolean connectives such as \emph{and}
(\( \fsh{\land} : \Bool \times \Bool \to \Bool\)), \emph{or}, etc. are
standard and where the carrier set \( M^b_s \) of any builtin sort
\( s \in S^b \) is exactly the set of builtin constant symbols
\( F^b_{\epsilon,s} = M^b_{s} \) of the appropriate sort. By
\( F^b_0 \) we denote \( \cup_{s \in S^b} F^b_{\epsilon,s} \). As
\( M^b = F^b_0 \), the set of builtin function symbols \( F^b \) might
be infinite. We will assume that first-order \( \Sigma^b \) formulae
can be decided by an oracle that is implemented in practice by a
best-effort SMT solver. We let \(\X\) be an \(S\)-sorted set of
variables. The set of \( \Sigma \)-terms with variables in \( \X \) is
denoted by \( T_\Sigma(\X) \).

\begin{example}

  Let \( S^b = \{ \sbool, \sint, \sid \}\). Let
  \( F^b = \{ \fsh{0}, \fsh{1}, \fsh{2}, \ldots : \to
  \sint,\allowbreak \fsh{+} : \sint \times \sint \to \sint,\allowbreak
  \fsh{=} : \sint \times \sint \to \sbool, \fsh{true}, \fsh{false} :
  \to \sbool,\allowbreak \fsh{\land} : \sbool \times \sbool \to
  \sbool, \ldots \} \). We assume that first-order constraints over
  \( \Sigma^b \) can be solved by an SMT solver implementing integer
  arithmetic. 
  
  Let \( S^c = \{ \sexp, \sstack, \scfg, \senv, \sfuncs \} \) and
  \( F^c = \{ \fsh{\plus{}{}}, \fsh{;} : \sexp \times \sexp \to
  \sexp,\allowbreak \fsh{ite} : \sexp \times \sexp \times \sexp \to
  \sexp,\allowbreak \fsh{[]} : \to \sstack,\allowbreak \fsh{\lisymb} :
  \sexp \times \sstack \to \sstack,\allowbreak
  \cfgimp{\cdot}{\cdot}{\cdot} : \sstack \times \senv \times \sfuncs
  \to \scfg, \ldots \} \). Let
  \( {\leq} = \{ \sbool \leq \sexp, \sint \leq \sexp, \sid \leq \sexp,
  \ldots \} \). The constructors sorts and function symbols model the
  syntax of an imperative language with boolean and arithmetic
  expressions, a global environment and a function map.

  Let
  \( F^a = \{ \fsh{val} : \sexp \to \sbool, \fsh{reduce} : \sint
  \times \sint \to \sbool, \ldots \} \). The axiomatized function symbols
  correspond to those discussed in Section~\ref{sec:intro}.
  
\end{example}

The set \( \CF \) of \emph{constraint formulae} is the set of
first-order formulae with equality over the signature \( \Sigma
\). The set \( \CF^b \) of \emph{builtin constraint formulae} is the
set of first-order formulae with equality over the signature
\( \Sigma^b \).
  
\begin{definition}[LCTRS]\label{def:lctrs}
  A \emph{logically constrained rewrite rule} is a tuple
  \(( l, r, \phi )\), often written as \(\rrule{l}{r}{\phi}\), where
  \( l, r \) are terms in \(T_{\Sigma}(\X)\) of the same sort, and
  \( \phi \in \CF \) is a first-order formula. A \emph{logically
    constrained term rewriting system} \( \R \) is a set of logically
  constrained rewrite rules.
\end{definition}

\begin{definition}[Reduction Relation Induced by an LCTRS]\label{def:transition}
  Given a \( \Sigma \)-model \( M \), an LCTRS \( \R \) induces a
  reduction relation on the sorted carrier set of \( M \) defined by:
  \[ {\rewrite_{\R}^M} = \left\{ \Big(\rho(C[l]), \rho(C[r])\Big) \;\;\; \mid
    \begin{array}{l}
      \;\;\; \rrule{l}{r}{\phi} \in \R \\
      \;\;\; \mbox{\(\rho : \X \to M\) is a valuation s.t. \(\rho(\phi) = \True\)}
      \\
      \;\;\;
      \mbox{\(C\) is an arbitrary context}
    \end{array} \right\}. \]
\end{definition}

We consider the model \( M^a \) of \( \Sigma \) whose sorted carrier
set is defined inductively by the following equations:

\begin{enumerate}
\item for any builtin sort \( s \in S^b \):
  \( M_s^a = M_s^b \cup \{ \fsh{\msh{f}(\msh{m_1}, \msh{\ldots},
    \msh{m_n})} \mid f \in F^a_{s_1 \ldots s_n, s}, m_1 \in M^a_{s_1},
  \ldots, m_n \in M^a_{s_n} \} \cup \{ \fsh{\msh{f}(\msh{m_1},
    \msh{\ldots}, \msh{m_n})} \mid f \in F^b_{s_1 \ldots s_n, s}, m_1 \in
  M^a_{s_1}, \ldots, m_n \in M^a_{s_n}\mbox{ and there exists
    \(1 \leq i \leq n\) s.t. \( m_i \not\in M^b_{s_i}\)} \} \);

\item for any constructor sort \( s \in S^c \):
  \( M_s^a = \{ \fsh{\msh{f}(\msh{m_1}, \msh{\ldots}, \msh{m_n})} \mid
  f \in F^a_{s_1 \ldots s_n, s} \cup F^c_{s_1 \ldots s_n, s}, m_1 \in
  M^a_{s_1}, \ldots,\allowbreak m_n \in M^a_{s_n} \} \).
\end{enumerate}

The interpretation of the function symbols in \( M^a \)
is defined as: \begin{enumerate}
\item builtin symbols applied to elements of \( M^b \): the same
  interpretation as in \( M^b \);
\item builtin symbols applied to elements of \( M^a \setminus M^b \):
  interpreted as free symbols;
\item constructor symbols and axiomatized symbols:
  interpreted as free symbols.
\end{enumerate}
Note that, since \( M^b_s = F^b_{\epsilon, s} \) for any builtin sort
\( s \in S^b \), ground terms over \( F^b_0 \cup F^c \) are elements
of \( M^a \).

We say that a LCTRS \( \R^a \) axiomatizes the symbols in \( F^a \) if:
\begin{enumerate}
\item the reduction relation \( \rewrite_{R^a}^{M^a} \) induced by
  \( R^a \) on the model \( M^a \) defined above is convergent and
\item the normal form \( m\mathop{\downarrow} \) of any element
  \( m \in M^a_s \) w.r.t. to \( \rewrite_{R^a}^{M^a} \) is a ground
  term \( m \in T_{F^b_0 \cup F^c} \) built from nullary builtins and
  constructor symbols.
\end{enumerate}

In what follows, we assume that \( R^a \) is an LCTRS that axiomatizes
\( F^a \).

\begin{example}
  Continuing the previous example, we consider
  \(R^a = \{ \allowbreak \fsh{\areduce(\vi, \vn)}\allowbreak
  \mathrel{\rewrite} \fsh{[]\myif \vi > \vn}, \fsh{\areduce(\vi, \vn)
    \mathrel{\rewrite} \mbrack{\plus{\vi}{\square}} \lisymb
    \areduce(\vi + 1, \vn)\myif \vi \leq \vn}, \ldots \} \).
\end{example}

We now fix a model \( M \) (depending on \( R^a \)) defined as follows:

\begin{enumerate}
\item for any builtin sort \( s \in S^b \): \( M_s^a = M_s^b \);

\item for any constructor sort \( s \in S^c \):
  \( M_s^a = \{ \fsh{\msh{f}(\msh{m_1}, \msh{\ldots}, \msh{m_n})} \mid f
  \in F_{s_1 \ldots s_n, s}, m_1 \in M^a_{s_1}, \ldots,\allowbreak m_n
  \in M^a_{s_n} \} \).
\end{enumerate}

That is, the carrier set of \( M \) is the set of ground terms
\( T_{F^b_0 \cup F^c} \) built from nullary builtins and constructor
symbols.  In \( M \), the builtin symbols are interpreted as in
\( M^b \), the constructor symbols are interpreted as in \( M^c \) and
axiomatized symbols \( f \in F^a \) are interpreted by:
\( M_{f(m_1, \ldots, m_n)} = f(m_1, \ldots, m_n)\mathop{\downarrow} \)
(the normal form w.r.t. the reduction relation induced by \( \R^a \)).

We call solving equations over terms \( t_1, t_2 \in T_\Sigma(\X) \) in
the model \( M \) \emph{unification modulo axiomatized symbols}
(UMAS). Unification modulo axiomatized symbols is a generalization of
\emph{unification modulo builtins} (UMB) that we have introduced in
our previous work~\cite{CiobacaWOLLIC2018}. Unlike usual unification
problems, where a unifier is simply a substitution, in UMB (and
therefore in UMAS as well) a unifier is a pair \( (\phi, \sigma) \),
where \( \phi \in \CF^b \) is a builtin logical constraint and
\( \sigma \) is a substitution.

A complete set of unifiers modulo axiomatized symbols of
\( t_1, t_2 \) is a set \( \textit{umas}(t_1, t_2) \) of pairs of
builtin constraints and substitutions such that:
\begin{enumerate}
\item (soundness) for any
  \( (\phi, \sigma) \in \textit{umas}(t_1, t_2) \), we have:
  \( \rho(\sigma(t_1)) = \rho(\sigma(t_2)) \) for any valuation
  \( \rho \) such that \( \rho(\phi) = \True \).
\item (completeness) for any valuation \( \rho : \X \to M \) s.t.
  \( \rho(t_1) = \rho(t_2) \), there exists an unifier
  \( (\phi, \sigma) \in \textit{umas}(t_1, t_2) \) and a valuation
  \( \rho' \) such that \( \rho = \rho' \circ \sigma \) and \(
  \rho'(\phi) = \True \).
\end{enumerate}

\begin{example} Consider \( t_1 = \cfgimp{[]}{\venv}{\vfuncs} \) and
  \( t_2 = \cfgimp{\areduce(\vi, \vn)}{\venv'}{\vfuncs'} \). We have
  (for example) that
  \( \textit{umas}(t_1, t_2) = \{ (\vi \fsh{>} \vn, \{ \venv' \mapsto
  \venv, \vfuncs' \mapsto \vfuncs \}) \} \). Recall that
  \(R^a \supseteq \{ \fsh{\areduce(\vi, \vn)} \mathrel{\rewrite}
  \fsh{[]\myif \vi > \vn}\} \).
\end{example}

\begin{definition}[Top-most LCTRSs]
\label{def:tm-lctrs}
  An LCTRS \( \R \) is top-most on \( M \) if
  \[ {\rewrite_{R}^{M}} = \left\{ \Big(\rho(l), \rho(r)\Big) \qquad
    \mid
    \begin{array}{l}
      \qquad \rrule{l}{r}{\phi} \in \R \\
      \qquad \mbox{\(\rho\) is a valuation s.t. \(\rho(\phi) = \True\)}
    \end{array} \right\},\] that is, all rewritings take place at the
root.
\end{definition}
The fact that an LCTRS \( \R \) is top-most can be ensured by
requiring that all rewrite rules are of some sort \( s \in S \) with
the property that no function symbol takes elements of sort \( s \) as
arguments. Therefore, terms of sort \( s \) can only be rewritten at
the root. There exist techniques~\cite{SerbanutaENTCS2007} for
transforming an LCTRS into a top-most one.

\begin{example}
  We now consider the top-most LCTRS
  \( \R = \{ \cfgimp{\laassign{\vx}{\vi} \lisymb \vcs}{\venv}{\vfuncs}
  \rewrite \\\allowbreak \cfgimp{\vcs}{\bupdate(\venv, \vx,
    \vi)}{\vfuncs}, \cfgimp{\laassign{\vx}{\ve} \lisymb
    \vcs}{\venv}{\vfuncs} \rewrite \cfgimp{\ve \lisymb
    \laassign{\vx}{\square{}} \lisymb \vcs}{\venv}{\vfuncs} \myif
  \\\allowbreak \lnot \bval(\ve), \ldots \} \). Only the function
  symbol \( \cfgimp{\cdot}{\cdot}{\cdot} \) returns an element of sort
  \( \scfg \) and no function symbol takes \( \scfg \) as an argument;
  therefore \( \R \) is top-most.
\end{example}
\begin{definition}[Constrained Terms]
  A \emph{constrained term} \(\varphi\) of sort \(s\in S\) is a pair
  \(\fsh{(\msh{t}, \msh{\phi})}\) (written
  \(\ct{\msh{t}}{\msh{\phi}}\)), where \(t \in T_{\Sigma,s}(\X)\) and
  \(\phi \in \CF\).
\end{definition}

We consistently use \(\varphi\) for constrained terms and \(\phi\) for
constraint formulae.

\begin{definition}[Valuation Semantics of Constraints]
\label{def:valuationsemantics}
The \emph{valuation semantics} of a constraint \(\phi\) is the set
\(\sem{\phi} \eqbydef \{\alpha : X \to M^\Sigma \mid\allowbreak M^\Sigma, \alpha
\models \phi\}\).
\end{definition}
\begin{definition}[State Predicate Semantics of Constrained
  Terms]
\label{def:statepredicatesemantics}
The \emph{state predicate semantics} of a constrained term
\(\ct{t}{\phi}\) is the set
\\
\centerline{\(\tsem{\ct{t}{\phi}}\eqbydef\{\alpha(t)\mid\alpha\in\sem{\phi}\}.\)}
\end{definition}

\begin{definition}[Derivatives of Constrained Terms]
  The \emph{set of derivatives} of a constrained term \(\ct{t}{\phi}\)
  w.r.t. a rule \(\rrule{l}{r}{\phi_{\it lr}}\) is
\begin{align*}
  \Delta_{l,r,\phi_{\it lr}}\big(\ct{t}{\phi}\big) = \{\ct{\sigma'(r)}{\phi \land \phi'\land \phi_{\it lr}}\mid{} & (\phi', \sigma') \in \textit{umas}(t, l), \\
                                           & (\phi \land \phi'\land \phi_{\it lr})\textrm{~satisfiable}\}
\end{align*}
\end{definition}
A constrained term \(\varphi\) is \emph{\(\R\)-derivable} if
\( \Delta_{\R}(\varphi) \not= \emptyset \).

We assume as usual that the rewrite rules in \(\R^a\) and the rewrite
rules in \(\rrule{l}{r}{\phi_{\it lr}} \in R\) are
coherent in the generalized sense~\cite{MeseguerWRLA2018}.

\begin{restatable}{theorem}{deltacommutes2}
  \label{th:der2}
  If \(\ct{t}{\phi}\) is a constrained term, then
  \[\tsem{\Delta_{\R}\big(\ct{t}{\phi}\big)} = \{\gamma' \mid \gamma
 \rewrite_{\R}^{M} \gamma' \textrm{~for~some~}\gamma \in \tsem{\ct{t}{\phi}}\}.\]
\end{restatable}
The theorem above ensures that the symbolic successors (derivatives)
of a constrained term are semantically correct. We write
\( \rewrite \) instead of \( \rewrite_{\R}^{M} \) when \( \R \) and
\( M \) can be inferred from the context.
% \begin{definition}
% \(\varphi\gtran{\R}^*\varphi'\) iff \(\varphi'\in\bar\Delta^k_{\R/\R^d}(\varphi)\) for some \(k\ge 0\).
% \\
% \(\varphi\gtran{\R}^+\varphi'\) iff \(\varphi'\in\bar\Delta^k_{\R/\R^d}(\varphi)\) for some \(k\ge 1\). 
% \end{definition}
%
\section{Language Semantics as LCTRSs}
\label{sec:semantics}

\begin{figure}[t!]
  \hrule
  {
    \centering \(\begin{array}{lll} \sexp
    & ::= & \sint \mid
            \sbool \mid \sid \mid \sexp {\tt\ binop\ }
            \sexp \mid {\tt unop\ }\sexp \mid
            \fsh{call}\ \sfuncall \mid
            \laskip \mid
    \\ & &  \laseq{\sexp}{\sexp} \mid
           \laassign{\sid}{\sexp} \mid 
           \lawh{\sexp}{\sexp} \mid
           \laite{\sexp}{\sexp}{\sexp}
    \\
    \sfuncall
    & ::= & \sid \mid \sfuncall(\sexp)
    \qquad \hfill
    \sfunbody
     ::= \sexp \mid \lambda \sid . \sfunbody
    \\
    \sstack
    & ::= & \fsh{[]} \mid \sexp \lisymb \sstack
    \hfill
    \scfg
    ::=
            \cfgimp{\sstack}{\senv}{\sfuncs}
\end{array}\)
              }
  \caption{\label{fig:impsyntax}The syntax of \IMP{}.}
  \hrule
\end{figure}
The operational semantics of a programming language can be encoded as
a top-most LCTRS. As a running example, we feature equivalence proofs
for programs written in an imperative language that we call \IMP{}. In
Figure~\ref{fig:impsyntax}, we define the syntax of the \IMP{}
language.
The syntax is given in a BNF-like notation, and it should be
understood as defining an order-sorted term signature. The
configuration is a tuple \(\cfgimp{\mcs}{\menv}{\mfuncs}\), consisting
of a stack \(\mcs\) of program expressions and statements to be
evaluated in order (as exemplified in Section~\ref{sec:intro}), a map
\(\menv\) from identifiers (program variables) to integers that acts
as a global store, and a map \(\mfuncs\) from function identifiers to
function bodies. In \IMP{}, both expressions and statements are
grouped under the syntactic category \sexp, and the difference between
them is encoded in the semantics: expressions are replaced by their
values in the computation stack, and statements are erased after their
effect is performed.
The language \IMP{} comes in two variations, \IMP{1} and \IMP{2}, both
sharing the same syntax. The difference is in the semantics: \IMP{1}
has an unbounded stack, while \IMP{2} has a bounded stack (we use a
parametric bound of \(k\) elements). As explained in the introduction,
we write \IMP{} instead of \IMP{1} and \IMP{2} when the exact
variation does not matter. In Figure~\ref{fig:impsemantics}, we define
the operational semantics of \IMP{1} in frame stack
style~\cite{PittsAS2000} as an LCTRS.
\begin{figure}[t!]
\hrule
\begin{center}
  \(\begin{array}{l}
%    \hline
      \cfgimp{\laassign{\vx}{\ve} \lisymb \vcs}{\venv}{\vfuncs}

      \rewrite

      \cfgimp{\ve \lisymb \laassign{\vx}{\square{}} \lisymb \vcs}{\venv}{\vfuncs}

      \myif \lnot \bval(\ve) \qquad \hfill
      \textit{assignment} \\
      
      \cfgimp{\vi \lisymb \laassign{\vx}{\square{}} \lisymb \vcs}{\venv}{\vfuncs}
      
      \rewrite

      \cfgimp{\laassign{\vx}{\vi} \lisymb \vcs}{\venv}{\vfuncs} \\
      
      \cfgimp{\laassign{\vx}{\vi} \lisymb \vcs}{\venv}{\vfuncs}

      \rewrite

      \cfgimp{\vcs}{\bupdate(\venv, \vx, \vi)}{\vfuncs} \\
    \hline
      \cfgimp{\vx \lisymb \vcs}{\venv}{\vfuncs}

      \rewrite

      \cfgimp{\b(\vx, \venv) \lisymb \vcs}{\venv}{\vfuncs}

      \qquad \hfill \textit{identifier lookup} \\

      \hline
      %% BINARY OPERATORS 
      \cfgimp{\plus{\vesub{1}}{\vesub{2}} \lisymb \vcs}{\venv}{\vfuncs}

      \rewrite

      \cfgimp{\vesub{1} \lisymb \plus{\square{}}{\vesub{2}} \lisymb \vcs}{\venv}{\vfuncs}

      \myif  \lnot \bval(\vesub{1})
      \qquad \hfill \textit{binary} \\ 
      
      \cfgimp{\visub{1} \lisymb \plus{\square{}}{\vesub{2}} \lisymb \vcs}{\venv}{\vfuncs}

      \rewrite

      \cfgimp{\plus{\visub{1}}{\vesub{2}} \lisymb \vcs}{\venv}{\vfuncs} \hfill \textit{
      operations} \\
      
      \cfgimp{\plus{\visub{1}}{\vesub{2}} \lisymb \vcs}{\venv}{\vfuncs}

      \rewrite

      \cfgimp{\vesub{2} \lisymb \plus{\visub{1}}{\square{}} \lisymb \vcs}{\venv}{\vfuncs}

      \myif  \lnot \bval(\vesub{2}) \hfill \textit{} \\ 
      
      \cfgimp{\visub{2} \lisymb  \plus{\visub{1}}{\square{}} \lisymb \vcs}{\venv}{\vfuncs}

      \rewrite

      \cfgimp{\plus{\visub{1}}{\visub{2}} \lisymb \vcs}{\venv}{\vfuncs} \\ 
      
      \cfgimp{\plus{\visub{1}}{\visub{2}} \lisymb \vcs}{\venv}{\vfuncs}

      \rewrite

      \cfgimp{\visub{1} + \visub{2} \lisymb \vcs}{\venv}{\vfuncs} \\

      \hline
      % \hdots 

      % \qquad \hfill \textit{similar rules for all other binary operations} \\

      %% UNARY OPERATORS 
      \cfgimp{\lanot{\ve} \lisymb \vcs}{\venv}{\vfuncs} 

      \rewrite

      \cfgimp{\ve \lisymb \lanot{\square{}} \lisymb \vcs}{\venv}{\vfuncs}

      \myif  \lnot \bval(\ve) 

      \qquad \hfill \textit{unary} \\ 
      
      \cfgimp{\vb \lisymb \lanot{\square{}} \lisymb \vcs}{\venv}{\vfuncs}

      \rewrite

      \cfgimp{\lanot{\vb} \lisymb \vcs}{\venv}{\vfuncs} \hfill
      \textit{operations} \\
      
      \cfgimp{\lanot{\vb} \lisymb \vcs}{\venv}{\vfuncs} 

      \rewrite

      \cfgimp{\overline{\vb} \lisymb \vcs}{\venv}{\vfuncs} \hfill
      \textit{} \\
      \hline
      % \ldots 

      % \qquad \hfill \textit{similar rules for all other unary operations} \\

      %% ITE
      \cfgimp{\laite{\top}{\vesub{2}}{\vesub{3}} \lisymb \vcs}{\venv}{\vfuncs}

      \rewrite

      \cfgimp{\vesub{2} \lisymb \vcs}{\venv}{\vfuncs}

      \qquad \hfill \textit{if-then-else} \\ 

      \cfgimp{\laite{\bot}{\vesub{2}}{\vesub{3}} \lisymb \vcs}{\venv}{\vfuncs}
     
      \rewrite

      \cfgimp{\vesub{3} \lisymb \vcs}{\venv}{\vfuncs} \\ 

      \cfgimp{\laite{\vesub{1}}{\vesub{2}}{\vesub{3}} \lisymb \vcs}{\venv}{\vfuncs} 

      \rewrite

      \cfgimp{\vesub{1} \lisymb \laite{\square{}}{\vesub{2}}{\vesub{3}} \lisymb
   \vcs}{\venv}{\vfuncs} \\

   \hfill \myif  \lnot \bval(\vesub{1}) \\ 

      \cfgimp{\vb \lisymb \laite{\square{}}{\vesub{2}}{\vesub{3}} \lisymb \vcs}{\venv}{\vfuncs} 

      \rewrite

      \cfgimp{\laite{\vb}{\vesub{2}}{\vesub{3}} \lisymb \vcs}{\venv}{\vfuncs} \\

    \hline
    
      \cfgimp{\lawh{\vesub{1}}{\vesub{2}} \lisymb \vcs}{\venv}{\vfuncs} 

      \rewrite
       
      \qquad \hfill \textit{while loop} \\

    \qquad \cfgimp{\laite{\vesub{1}}{\mbrack{\laseq{\vesub{2}}{\lawh{\vesub{1}}{\vesub{2}}}}}{\laskip} \lisymb \vcs}{\venv}{\vfuncs} \\

    \hline
      \cfgimp{\laseq{\vesub{1}}{\vesub{2}} \lisymb \vcs}{\venv}{\vfuncs} 

      \rewrite

      \cfgimp{\vesub{1} \lisymb \vesub{2} \lisymb \vcs}{\venv}{\vfuncs} 

      \qquad \hfill \textit{sequence} \\ 
    \hline
      \cfgimp{\laskip \lisymb \vcs}{\venv}{\vfuncs} 

      \rewrite

      \cfgimp{\vcs}{\venv}{\vfuncs} 

      \qquad \hfill \textit{skip} \\ 
    \hline
      \cfgimp{\lacall{\vf} \lisymb \vcs}{\venv}{\vfuncs}

      \rewrite

      \cfgimp{\blookup(\vf, \vfuncs) \lisymb \vcs}{\venv}{\vfuncs}
    \hfill \textit{function calls}\\

      \cfgimp{\lacall{\vf(\ve)} \lisymb \vcs}{\venv}{\vfuncs}

      \rewrite

      \cfgimp{\lacall{\vf} \lisymb \lacall{\square{}(\ve)} \lisymb \vcs}{\venv}{\vfuncs} \\

      \cfgimp{\lambda \vx.\vfunbody \lisymb \lacall{\square{}(\ve)} \lisymb
      \vcs}{\venv}{\vfuncs}

      \rewrite

      \cfgimp{\lacall{\lambda \vx.\vfunbody (\ve)} \lisymb \vcs}{\venv}{\vfuncs}
      \\
      
      \cfgimp{\lacall{\lambda \vx.\vfunbody(\ve)} \lisymb \vcs}{\venv}{\vfuncs}

      \rewrite

      \cfgimp{\ve \lisymb \lacall{\lambda \vx.\vfunbody(\square{})} \lisymb \vcs}{\venv}{\vfuncs}
      
      \hfill \myif  \lnot \bval(\ve) \\
      
      \cfgimp{\vi \lisymb \lacall{\lambda \vx.\vfunbody(\square{})}
      \lisymb \vcs}{\venv}{\vfuncs}

      \rewrite

      \cfgimp{\lacall{\lambda \vx.\vfunbody(\vi)} \lisymb \vcs}{\venv}{\vfuncs} \\
      
      \cfgimp{\lacall{\lambda \vx.\vfunbody(\vi)} \lisymb \vcs}{\venv}{\vfuncs}

      \rewrite

    \cfgimp{\asubst(\vx, \vi, \vfunbody) \lisymb \vcs}{\venv}{\vfuncs}\\
%    \hline
    \end{array}\)
  \end{center}
  \caption{\label{fig:impsemantics}The small-step operational
    semantics of \IMP{1} encoded as an LCTRS in frame stack style. The
    language \IMP{1} has an unbounded stack. The variables have the
    following sorts: \( \vx : \sid, \ve : \sexp, \vcs : \sstack, \venv
    : \senv, \vfuncs : \sfuncs, \vi : \sint, \vb : \sbool, \vf : \sid,
    \vfunbody : \sfunbody \).}
\hrule
\end{figure}
Note the use of the axiomatized symbol \asubst in the last rule for
the function calls.

The axiomatized symbol \asubst is axiomatized by the following rules:
\begin{enumerate}
%% \[\begin{array}{l}
%%     \hline
\item \(\asubst(\vx, \ve, \laassign{\vy}{\vesub{1}}) \mathop{\rewrite}
  \laassign{\vy}{\asubst(\vx, \ve, \vesub{1})}\),
  
\item \(\asubst(\vx, \ve, \vy) \mathop{\rewrite} \vy \myif \vx \fsh{\not=}
  \vy\),
  
\item \(\asubst(\vx, \ve, \plus{\vesub{1}}{\vesub{2}}) \mathop{\rewrite} \plus{\asubst(\vx, \ve,
    \vesub{1})}{\asubst(\vx, \ve, \vesub{2})}\),
  
\item \(\asubst(\vx, \ve, \vx) \mathop{\rewrite} \ve\),
  
\item \(\allowbreak\asubst(\vx, \ve, \lanot{\ve'}) \mathop{\rewrite}
  \lanot{\asubst(\vx, \ve, \ve')}\),
  
\item \(\asubst(\vx, \ve, \laite{\ve'}{\vesub{1}}{\vesub{2}}) \mathop{\rewrite}
  \\ \phantom{\qquad} \laite{\allowbreak\asubst(\vx, \ve,
    \ve')}{\asubst(\vx, \ve, \vesub{1})}{\asubst(\vx, \ve,
    \vesub{2})}\),
  
\item \(\asubst(\vx, \ve, \laseq{\vesub{1}}{\vesub{2}}) \mathop{\rewrite} \laseq{\asubst\allowbreak (\vx, \ve,
    \vesub{1})}{\asubst(\vx, \ve, \vesub{2})}\).
  %%   \hline
  %% \end{array}\]
\end{enumerate}
Note that the symbols \blookup and \bupdate in the usual theory of
arrays are used for handling the environment in the imperative
part. The axiomatized symbol \bval returns true exactly for the values
of the language: integers and booleans:
\begin{enumerate*}
\item \( \fsh{\bval(\vi)} \mathop{\rewrite} \fsh{\top} \);
\item \( \fsh{\bval(\vb)} \mathop{\rewrite} \fsh{\top} \);
\item \( \fsh{\bval(\plus{\vesub{1}}{\vesub{2}})} \mathop{\rewrite}
  \fsh{\bot} \);
\item \( \fsh{\bval(\laseq{\vesub{1}}{\vesub{2}})} \mathop{\rewrite}
  \fsh{\bot} \); etc.
\end{enumerate*}

An example of evaluation using the operational semantics encoded as an
LCTRS is shown in Appendix~\ref{app:examplesemantics}.
The constrained rewrite rules defining \IMP{2} are presented in
Figure~\ref{fig:imp2semantics}. They are similar to the rules for
\IMP{1}, except that each rule that increases the stack size has an
additional constraint. The constraint prevents the rule from firing if
the stack would become too big.
\section{Simulation-based Equivalence Proofs}
\label{sec:equivalence}
In the following, we will assume that \(\R_L\) and \(\R_R\) are two LCTRSs
modeling the semantics of two programming languages: the \emph{left}
language and the \emph{right} language. We assume that the sorts
\(\CfgL\) and \(\CfgR\) are the sorts of configurations of the left
language and of the right language, respectively. The LCTRSs \(\R_L\)
and \(\R_R\) induce transition relations on the interpretations of
\(\CfgL\) and \(\CfgR\), transition relations that capture the operational
semantics of the two languages. The languages might be the same, or
they might be different; all our results are parametric in \(\R_L\) and
\(\R_R\). In the following examples we take \(\R_L = \R_R = \IMP{}\) and
\( \CfgL = \CfgR = \scfg \).

Formally, we define equivalence not between programs, but between
\emph{program configurations}. Program configurations typically
contain both a program and additional information such as program
counter, heap information, stack information or others, depending on
the particular programming language. In our examples, the program
configuration is a tuple
\(\cfgimp{\msh{[e_1, \ldots, e_n]}}{\menv}{\mfuncs}\) consisting of a
stack \(e_1, \ldots, e_n\) of expressions to be evaluated in this
order, an environment \menv mapping global program variables to
integers, and a map \mfuncs from identifiers to function bodies.

We sometimes distinguish between \emph{symbolic program
  configurations} (terms of sort \(\CfgL\) or \(\CfgR\), possibly with
variables) and \emph{ground program configurations} (elements of the
interpretation of the sorts \(\CfgL\) and \(\CfgR\)). Due to the
definition of the fixed model in which we work, any ground program
configuration is also a symbolic program configuration with \(0\)
variables. Our proof method shows equivalence between two symbolic
program configurations. The fact that the same variable occurs in both
symbolic configurations models that the two programs take the same
input.

\begin{figure}[t!]
\hrule
  \begin{center}
\(\begin{array}{l}
      \cfgimp{\laassign{\vx}{\ve} \lisymb \vcs}{\venv}{\vfuncs}

      \rewrite

      \cfgimp{\ve \lisymb \laassign{\vx}{\square{}} \lisymb \vcs}{\venv}{\vfuncs}

      \qquad \hfill \myif \lnot \bval(\ve) \land
      \alen(\vcs)  < k \\

      \cfgimp{\plus{\vesub{1}}{\vesub{2}} \lisymb \vcs}{\venv}{\vfuncs}

      \rewrite

      \cfgimp{\vesub{1} \lisymb \plus{\square{}}{\vesub{2}} \lisymb \vcs}{\venv}{\vfuncs}

      \qquad \hfill \myif  \lnot \bval(\vesub{1}) \land
      \alen(\vcs)  < k \\
      
      \cfgimp{\plus{i_1}{\vesub{2}} \lisymb \vcs}{\venv}{\vfuncs}

      \rewrite

      \cfgimp{\vesub{2} \lisymb \plus{i_1}{\square{}} \lisymb \vcs}{\venv}{\vfuncs}

      \qquad \hfill \myif  \lnot \bval(\vesub{2}) \land
      \alen(\vcs)  < k \\
      
      \cfgimp{\lanot{\vesub{1}} \lisymb \vcs}{\venv}{\vfuncs} 

      \rewrite

      \cfgimp{\vesub{1} \lisymb \lanot{\square{}} \lisymb \vcs}{\venv}{\vfuncs}

      \qquad \hfill \myif  \lnot \bval(\vesub{1}) \land
      \alen(\vcs)  < k \\
      
      \cfgimp{\laite{\vesub{1}}{\vesub{2}}{\ve_3} \lisymb \vcs}{\venv}{\vfuncs} 

      \rewrite

      \cfgimp{\vesub{1} \lisymb \laite{\square{}}{\vesub{2}}{\ve_3} \lisymb
      \vcs}{\venv}{\vfuncs} \\

      \qquad \hfill \myif  \lnot \bval(\vesub{1}) \land \alen(\vcs)  < k \\

      \cfgimp{\laseq{\vesub{1}}{\vesub{2}} \lisymb \vcs}{\venv}{\vfuncs} 

      \rewrite

      \cfgimp{\vesub{1} \lisymb \vesub{2} \lisymb \vcs}{\venv}{\vfuncs} 

      \qquad \hfill \myif  \alen(\vcs)  < k \\

      \cfgimp{\lacall{\vf(\ve)} \lisymb \vcs}{\venv}{\vfuncs}

      \rewrite

      \cfgimp{\lacall{\vf} \lisymb \lacall{\square{}(\ve)} \lisymb
      \vcs}{\venv}{\vfuncs}

      \qquad \hfill \myif  \alen(\vcs)  < k \\

      \cfgimp{\lacall{\lambda \vx.\vfunbody(\ve)} \lisymb \vcs}{\venv}{\vfuncs}

      \rewrite

      \cfgimp{\ve \lisymb \lacall{\lambda \vx.\vfunbody(\square{})} \lisymb \vcs}{\venv}{\vfuncs}
      
      \\ \qquad \hfill \myif  \lnot \bval(\ve) \land
      \alen(\vcs)  < k \\
  \end{array}\)
  \end{center}
  \caption{\label{fig:imp2semantics}The small-step operational
    semantics of \IMP{2}. The language \IMP{2} has a stack bounded to
    \(k + 1\) elements, where \(k\) is a parameter of the
    semantics. Only the rules different from the rules in \IMP{1} are
    presented. The constraint \(\alen(\vcs) < k\) ensures that the
    stack is of length at most \(k + 1\) in the rhs.}
  \hrule
\end{figure}

For our motivating example, we show the equivalence of the symbolic
program configurations
\[\cfgimp{[\lacall{f(\vN)}]}{\venv}{\funcsct}\mbox{ and
  }\cfgimp{[\lacall{F(\vN, 0, 0)}]}{\venv}{\funcsct}\mbox{ for
    \(\fsh{\vN \geq 0}\),}\] where \funcsct is the map containing the
function bodies corresponding to the function identifiers \fsh{f} and
\fsh{F}.  We use \(\fsh{F(\vN, 0, 0)}\) as an abbreviation for the
syntactic construct \(\fsh{F(\vN)(0)(0)}\) of sort \(\sfuncall\). Note
that \(\vN\) is a variable of sort \sint and \(\venv\) is a variable
of sort \senv (map from program identifiers to integers). % The map
% \(\funcsct\) defines the functions \(f\) and \(F\) (see
% Example~\ref{ex:running} on Page~\pageref{ex:running}).
The fact that both \(\vN\) and \(\venv\) occur in the two symbolic
configurations models that we want the two programs to take the same
input \(\vN\) and to start in the same environment \(\venv\). In the
initial configuration, the helper arguments of \(\fsh{F}\) are fixed
to \(\fsh{0}\).

Sometimes two programs perform the same computation, but record the
results in slightly different places. For example, an imperative
program might store its result in a global variable \pvar{result},
while a functional program would simply reduce to its final value. We
still want to consider these programs equivalent. Therefore, we
parameterize our definition for equivalence by a set of \emph{base
  cases}, which define the pairs of terminal ground configurations
that are known to be equivalent. We denote by
\(\GB\) (for base) the set of pairs of ground terminal program configurations
that are known to be equivalent. In our motivating example, \(\GB = \{
(\cfgimp{[\msh{i}]}{\menv}{\mfuncs},
\cfgimp{[\msh{i'}]}{\menv'}{\mfuncs'}) \mid i = i' \land i,i' \in
\mathbb{Z}
\}\) (both configurations have been reduced to the same integer \(i =
i'\), the environments are allowed to be different:
\(\textit{env}\) on the lhs and
\(\textit{env}'\) on the rhs, and the function maps are also allowed to
be different). See Example 4 in Appendix~\ref{app:examples} for a more
complex example.

We propose two definitions for the notion of functional equivalence of
programs, based on two-way simulations.
In the following, \(P\) denotes a symbolic configuration of sort
\(\CfgL\), \(Q\) denotes a symbolic configuration of sort \(\CfgR\) and
\(\phi\) denotes a first-order constraint.
In the following definitions, by a complete path
\(\rho(P) \gtran{\R_L}^* \GP'\), we mean that no further reduction
step is possible for \(\GP'\).
\begin{definition}[Full Simulation]\label{def:fullsimulation}
  We say that a symbolic program configuration \emph{\(P\) is fully
    simulated by a symbolic program configuration \(Q\) under constraint
    \(\phi\) with a set of base cases \(\GB\)} (denoted by
  \(\GB \models P \shprec Q\myif\phi\)) if, for
  any valuation \(\rho\) such that \(\rho(\phi) = \True\) and for any
  complete path \(\rho(P) \gtran{\R_L}^* \GP'\), there exists a complete
  path \(\rho(Q) \gtran{\R_R}^* \GQ'\) such that \((\GP', \GQ') \in
  \GB\);
\end{definition}
This intuitively states that for any terminating run of the left hand
side on some input, there is a terminating run of the right hand side
on the same input, such that the results are part of \(\GB\)
(e.g., the results are equal).
The notion of \emph{full simulation} is inspired from the usual notion
of \emph{full equivalence} in the relational program verification
literature. It can be seen as a lopsided version of full
equivalence. Full simulation is a transitive relation (assuming the
base cases are defined consistently).
\begin{definition}[Partial Simulation]
\label{def:psim}
  We say that a symbolic program configuration \emph{\(P\) is partially
    simulated by a symbolic program configuration \(Q\) under constraint
    \(\phi\) with a set of base cases \(\GB\)} (denoted by
  \(\GB \models P \shpreceq Q\myif \phi\)) if, for any
  valuation \(\rho\) such that \(\rho(\phi) = \True\) and for any
  complete path \(\rho(P) \gtran{\R_L}^* \GP'\), 
  one of the following holds: 
  \begin{itemize*}% [wide, labelwidth=!, labelindent=0pt]
  \item[\(\bullet\)]
  there exists a complete path \(\rho(Q)
    \gtran{\R_R}^* \GQ'\) s. t. \((\GP', \GQ') \in \GB\);
  \item[\(\bullet\)] there exists an infinite path \(\rho(Q) \gtran{\R_R} \ldots\).
  \end{itemize*}
\end{definition}
Intuitively, a terminating run of the left hand side on some input is
considered simulated by either (1) a terminating run of the right hand
side on the same input with the same output or (2) by an infinite run
of the right hand side on the same input.
The notion of \emph{partial simulation} is inspired from the usual
notion of \emph{partial equivalence} in the relation program
verification literature. It can be seen as a lopsided version of
partial equivalence. Partial simulation is not a transitive relation
(even for consistently defined sets of base cases).
Note that \({\shprec}\subseteq{\shpreceq}\) (for a fixed \GB), which
justifies the notation.
\begin{definition}[Full Equivalence]
  Two symbolic program configurations \(P\) and \(Q\) are fully equivalent
  under constraint \(\phi\) with a set of base cases \(\GB\),
  written \(\GB \models P \shsim Q\myif \phi,\) if
  \(\GB \models P \shprec Q\myif \phi\) and
  \(\GB^{-1} \models Q \shprec P\myif \phi\).
\end{definition}
Two-way full simulation gives the usual notion of full equivalence
(for determinate programs).
\begin{definition}[Partial Equivalence]
  Two symbolic program configurations \(P\) and \(Q\) are partially
  equivalent under constraint \(\phi\) with a set of base cases
  \(\GB\), written \(\GB \models P \shsimeq Q\myif\phi,\) if
  \(\GB \models P \shpreceq Q\myif\phi\) and
  \(\GB^{-1} \models Q \shpreceq P\myif\phi\).
\end{definition}
Two-way partial simulation is the usual notion of partial equivalence
(for determinate programs). Partial equivalence is not an equivalence
relation (hence the name \emph{partial}). The notation is justified by
the fact that \({\shsim}\subseteq{\shsimeq}\) (for a fixed \(\GB\)).
\section{Proof Systems}
\label{sec:proofsystem}
In this section, we give proof systems for partial and full
simulation. We work with simulation formulae of the form
\(P \shprec Q\myif \phi\) for full simulation and of the
form \(P \shpreceq Q\myif \phi\) for partial simulation,
where \(P\) is a symbolic configuration of sort \(\CfgL\), \(Q\) is a
symbolic configuration of sort \(\CfgR\) and \(\phi\) is a first-order
logical constraint.
We first present the proof system for full simulation. % The meaning of
% \(P \shprec Q\myif \phi\) (w.r.t. a set of base cases \GB)
% is given in Defintion~\ref{def:fullsimulation}: the symbolic
% configuration \(P\) is fully simulated by the symbolic configuration \(P\)
% in the cases where \(\phi\) holds. 
For a set \(R\) of simulation formulae
\(P \shprec Q\myif \phi\), its denotation is
\[\begin{array}{l}\llbracket R \rrbracket = \{ (\rho(P), \rho(Q)) \mid
   (P \shprec Q\myif\phi) \in R \mbox{ and } \rho \models \phi \}\end{array}\] i.e., 
   the pairs of instances of \(P\) and \(Q\) that satisfy \(\phi\).

We fix a set \(B\) of simulation formulae that under-approximates the
set \(\GB\) of base cases: \(\llbracket B \rrbracket \subseteq \GB\). We
also consider a set \(G\) (for \emph{goals}) consisting of simulation
formulae to be proven. The set \(G\) usually includes the actual goal,
but also a set of intermediate helper goals that are needed for the
proof that we call \emph{circularities}.
\begin{example}\label{ex:running}
  For our motivating example, we use the language \IMP{1} and the
  following set of base cases:
  \(B = \{ \cfgimp{[\vi]}{\venv}{\vfuncs} \shprec
  \cfgimp{[\vi']}{\venv'}{\vfuncs'} \myif \fsh{\vi=\vi'}\}\). 
  Thus, the set of
  base cases includes only pairs of identical terminal configurations
  where the stack of expressions contains the same integer. The set of goals includes the actual goal to be proven and
  two helper circularities:
  \[G = \left\{\begin{array}{l}

                 \cfgimp{[\lacall{f(\vN)}]}{\venv}{\funcsct}
                 \shprec % \\

                 % \qquad
                 \cfgimp{[\lacall{F(\vN, 0, 0)}]}{\venv}{\funcsct} % \\
                 % \qquad \qquad
                 \myif \mbox{\(\fsh{0 \leq \vN},\)} \\

                 \cfgimp{\lacall{f(\vI-1)} \lisymb \mathit{reduce}(\vI, \vN)}{\venv}{\funcsct}
                 \shprec % \\
                 
                 % \qquad
                 \cfgimp{[\lacall{F(\vN, 0, 0)}]}{\venv}{\funcsct}
                 \\
                 \qquad \qquad \myif \mbox{\(\fsh{0 \leq \vI
                 \leq \vN},\)} \\

                 \cfgimp{\vS \lisymb \mathit{reduce}(\vI,
                 \vN)}{\venv}{\funcsct} \shprec % \\ 

                 % \qquad 
                 \cfgimp{[\lacall{F(\vN, \vI, \vS)}]}{\venv}{\funcsct} % \\
                 % \qquad \qquad
                 \myif \mbox{\(\fsh{1 \leq \vI \leq \vN},\)}
\end{array}
\right\}\]
\noindent where 
\(\funcsct = \fsh{\{\ \pvar{f} \mapsto \lambda \pvar{n}.
  \laite{\pvar{n} = 0}{0}{\plus{\pvar{n}}{\lacall{\pvar{f}(\pvar{n} -
      1)}}}}, \ \fsh{\pvar{F} \mapsto \lambda \pvar{n} . \lambda
  \pvar{i} . \lambda \pvar{a} .}\allowbreak \fsh{\laite{\pvar{i} \leq
    \pvar{n}}{\lacall{F(n,\plus{i}{1}, \plus{a}{i})}}{\pvar{a}}}\}\)
is the function map,
\( \pvar{n}, \pvar{i}, \pvar{a}, \pvar{f}, \pvar{F} \) are constants
of sort \sid (program identifiers), and where
\( \vN, \vI, \vS : \sint \) are variables.
\end{example}
The proof system for full simulation, presented in
Figure~\ref{fig:fullsim}, manipulates sequents of the form
\(G, B \vdash^g P \shprec Q\myif \phi\), where
\(g \in \{ 0, 1 \}\), \(G\) is the set of goals and \(B\) is the set of base
cases. The superscript \(g\) to the turnstile is a boolean flag
(representing a \emph{guard}) that denotes whether circularities are
enabled or not, as formalized in the proof rules.

The \textsc{Axiom} rule states that any \(P\) is simulated by any \(Q\)
under the constraint \(\bot\) (false). The \textsc{Base} rule handles
the case when the right hand side \(Q\) can take a number of steps into
\(Q'\) so that the base cases are reached (the pair (\(P\), \(Q'\)) is part
of the base cases). The constraint \(\textit{sub}((P, Q), R)\) expresses
that \(P \shprec Q\) is an instance of \(R\). The \textsc{Circ} rule handles
the case when \(Q\) reaches in a number of steps a configuration \(Q'\)
such that the pair \(P\), \(Q'\) is subsumed by some goal in \(G\). In order
to ensure soundness, this rule can only be applied when the
superscript for the turnstile, \(g\), is \(1\). The superscript becomes
\(1\) only in rule \textsc{Step}, which intuitively takes a step in the
left-hand side. Therefore, when the rule \textsc{Circ} is actually
used, it means that there was progress on the lhs (by a
\emph{previous} \textsc{Step} on the lhs). This ensures soundness.

Finally, rule \textsc{Step} can be used to take a symbolic step in the
left-hand side. Recall that \(\Delta\) computes the symbolic
successors of a configuration. Note that all possible symbolic steps
from \(P\) are taken. This corresponds to the fact that in our notion
of simulation, \emph{every} run of \(P\) must have a corresponding run
in \(Q\). The constraint
\(\fsh{\lnot \msh{\phi^1} \land \msh{\ldots} \land \lnot
  \msh{\phi^n}}\) describes the instances of \(P\) where no step can
be taken, and therefore these configurations must be solved by some
other rule, hence the second line in the hypotheses of \textsc{Step}.
\begin{figure}[t]
\hrule
\begin{center}
  Notation: \(\textit{sub}\Big((P, Q), R\Big)\eqbydef\bigvee_{P' \shprec Q'\myif\phi' \in R} \fsh{\exists} \var(P', Q', \phi')\fsh{.(
    \msh{\phi'} \land \msh{P}{=}\msh{P'} \land \msh{Q}{=}\msh{Q'})}\)
\\
\(\inferrule*[left=\textsc{Axiom}]{\ }{G, B \vdash^g P \shprec Q\myif\bot}\) \\
  \(\inferrule*[left=\textsc{Base},right=\mbox{\(\models \phi_B
    \fsh{\limplies}\hspace{-0.75cm} \fsh{\bigvee\limits_{\msh{Q'}\myif\msh{\phi'} \msh{\in}
      \msh{\Delta_{\R_R}^{\leq k}(Q)}}} \hspace{-0.75cm}\phi' \fsh{\limplies} \textit{sub}\Big((P, Q'),
    B\Big)\)}]{G, B \vdash^g P \shprec Q\myif\phi \mathrel{\fsh{\land}}
  \fsh{\lnot}\phi_B}{G, B \vdash^g P \shprec Q\myif\phi}\)
 \(\inferrule*[left=\textsc{Circ},right=\mbox{\(\models \phi_G
    \fsh{\limplies}\hspace{-0.75cm} \fsh{\bigvee\limits_{\msh{Q'}\myif\msh{\phi'} \msh{\in}
      \msh{\Delta_{\R_R}^{\leq k}(Q)}}}\hspace{-0.75cm} \phi' \fsh{\limplies} \textit{sub}\Big((P,
    Q'), G\Big)\)}]{G, B \vdash^1 P \shprec Q\myif\phi
  \mathrel{\fsh{\land}} \fsh{\lnot}\phi_G}{G, B \vdash^1 P \shprec Q\myif\phi}\)
  \(\inferrule*[left=\textsc{Step},right=\mbox{\(\Delta_{\R_L}(P\myif\phi){=}\{ P^i\myif\phi^i{\mid}1 \leq i \leq n\}\)}]{\mbox{\(\begin{array}{l}G, B \vdash^1 P^i
     \shprec Q\myif\phi^i\mbox{ (for all \(1 \leq i \leq
       n\))}\\ G, B \vdash^g P \shprec Q\myif \fsh{\msh{\phi}
     \land \lnot \msh{\phi^1} \land \msh{\ldots} \land \lnot \msh{\phi^n}}
\end{array}\)}}{G, B \vdash^g P
  \shprec Q\myif\phi}\)
\end{center}
\caption{\label{fig:fullsim}The proof system for full simulation.}
\hrule
\end{figure}

We write \(G, B \vdash^g G'\) if
\(G, B \vdash^g P \shprec Q\myif \phi\) for any formula
\(P \shprec Q\myif \phi \in G'\) (that is, all formulae in
\(G'\) are provable from \(G, B\)). We are now ready to give the main
soundness theorem of our result for full simulation.
\begin{restatable}[Soundness for full simulation]{theorem}{soundenessfullsim}
\label{thm:sound-fsim}
If \(G, B \vdash^0 G\) and \(\llbracket B \rrbracket \subseteq \GB\), then
for any simulation formula \(P \shprec Q\myif \phi \in
G,\) we have that \(\GB \models P \shprec Q\myif \phi.\)
\end{restatable}
The theorem requires that \emph{all formulae} in \(G\) be proved in
order to trust \emph{any} of them. If a formula in \(G\) is not provable
then even if the others are provable, they cannot be trusted to hold
semantically. The starting superscript of the turnstile must be \(0\) in
order not to allow the \textsc{Circ} rule to fire immediately
(otherwise the \textsc{Circ} rule could be used to prove a circularity
by itself, leading to unsoundness).
% a terminating run of \(P\) to be simulated by a
% non-terminating run of \(Q\), which would not be sound for proving full
% simulation

% 
\begin{example}
  \label{ex:running1}
  Continuing Example~\ref{ex:running}, we have \(G, B \vdash^0 G\). To
  save space, we sometimes abbreviate
  \(\cfgimp{[\lacall{f(\vN)}]}{\venv}{\funcsct}\) by \(\fsh{f}\) and
  \(\cfgimp{[\lacall{F(\vN, 0, 0)}]}{\venv}{\funcsct}\) by
  \(\fsh{F}\). As \(G, B \vdash^0 G\), \(\fsh{f \shprec F}\) for
  \(\fsh{\vN \geq 0}\). As explained in the introduction, our proof
  system cannot establish the other direction,
  \(G^{-1}, B^{-1} \vdash^0 G^{-1}\), intuitively because it cannot
  prove that the termination of \(\fsh{F}\) (one phase) implies the
  termination of \(\fsh{f}\) (two phases). The full simulation
  relation would require an operationally-based termination
  argument~\cite{BuruianaEPTCS2019} for the second phase of
  \(\fsh{f}\), which we leave for future work.
\end{example}
\emph{Proving partial simulation}. We adopt the same notation as above
for a set of base cases \(B\) and a set of goals \(G\), but replacing
\(\shprec\) by \(\shpreceq\). The proof system for partial simulation,
presented in Figure~\ref{fig:partsim}, manipulates sequents of the
form \(G, B \vdash^g P \shpreceq Q\myif \phi\), where \(G\) is the set
of goals (with \(\shpreceq\) instead of \(\shprec\)), \(B\) is the set
of base cases (with \(\shpreceq\) instead of \(\shprec\)), and
\(g \in \{ 0, 1 \}\). The meaning of \(g\) is different: \(g = 1\)
enables the \textsc{Circ} rule to not make progress on the rhs.

The proof rules for partial simulation are similar to those for full
simulation, and therefore we only underline the main differences. The
rules \textsc{Axiom} and \textsc{Base} are identical to those in the
proof system for full simulation (except for \(\shpreceq\) instead of
\(\shprec\)). The rule \textsc{Circ} is also very similar, but progress is
required on the rhs unless \(g = 1\). Therefore, for partial simulation,
it is allowed to discharge a goal directly by \textsc{Circ}, without
taking any step in the left hand side, but with progress on the rhs
(note the superscript \(\geq 1-g\)). This corresponds to the case where
a terminal configuration is partially simulated by an infinite loop.
Finally, another difference is that once rule \textsc{Step} is applied
to make progress on the lhs, circularities can be applied even if
there no progress on the rhs. This corresponds potentially to the case
where the left-hand side loops forever and the right hand side
finishes.
\begin{figure}[t]
  \hrule
  \begin{center}
  Notation: \(\textit{sub}\Big((P, Q), R\Big)\eqbydef\bigvee_{P' \shpreceq Q'\myif\phi' \in R} \fsh{\exists} \var(P', Q', \phi')\fsh{.(\msh{\phi'} \land \msh{P}{=}\msh{P'} \land \msh{Q}{=}\msh{Q'})}\) \\

  \(\inferrule*[left=\textsc{Axiom}]{\ }{G, B \vdash^g P \shpreceq Q\myif\bot}\)   \\
  \(\inferrule*[left=\textsc{Base},right=\mbox{\(\models \phi_B
    \fsh{\limplies}\hspace{-0.75cm} \fsh{\bigvee\limits_{\msh{Q'}\myif\msh{\phi'} \msh{\in}
      \msh{\Delta_{\R_R}^{\leq k}(Q)}}}\hspace{-0.75cm} \phi' \fsh{\limplies}\textit{sub}\Big((P, Q'),
    B\Big)\)}]{G, B \vdash^g P \shpreceq Q\myif\phi \mathrel{\fsh{\land}}
  \fsh{\lnot}\phi_B}{G, B \vdash^g P \shpreceq Q\myif\phi}\) % \and
  \(\inferrule*[left=\textsc{Circ},right=\mbox{\(\models \phi_G \fsh{\limplies}\hspace{-0.75cm}
    \fsh{\bigvee\limits_{\msh{Q'}\myif\msh{\phi'} \msh{\in} \msh{\Delta_{\R_R}^{\geq 1-g,
        \leq k}(Q)}}}\hspace{-0.75cm} \phi' \fsh{\limplies} \textit{sub}\Big((P, Q'),
    G\Big)\)}]{G, B \vdash^g P \shpreceq Q\myif\phi \mathrel{\fsh{\land}}
  \fsh{\lnot}\phi_G}{G, B \vdash^g P \shpreceq Q\myif\phi}\) % \and
  \(\inferrule*[left=\textsc{Step},right=\mbox{\(\Delta_{\R_L}(P\myif\phi) {=} \{ P^i\myif\phi^i{\mid}1 \leq i \leq n\}\)}]{\mbox{\(\begin{array}{l}G, B \vdash^1 P^i
  \shpreceq Q\myif\phi^i\mbox{ (for all \(1 \leq i \leq
       n\))}\\ G, B \vdash^g P \shpreceq Q\myif \fsh{\msh{\phi}
     \land \lnot \msh{\phi^1} \land \ldots \land \lnot \msh{\phi^n}}
\end{array}\)}}{G, B \vdash^g P
  \shpreceq Q\myif\phi}\)
\caption{\label{fig:partsim}The proof system for partial simulation.}
\end{center}
\hrule
\end{figure}
As for full simulation, we write \(G, B \vdash G'\) if
\(G, B \vdash P \shpreceq Q\myif \phi\) for any formula
\( P \shpreceq Q\myif \phi \in G'\) (that is, all formulae in
\(G'\) are provable from \(G, B\)). We now give the main soundness theorem
of our result for partial simulation.
\begin{restatable}[Soundness for partial simulation]{theorem}{soundenesspartsim}
\label{thm:sound-psim}
%\begin{theorem}[Soundness for partial simulation]
%
  If \(G, B \vdash^0 G\) and \(\llbracket B \rrbracket \subseteq \GB\), then
  for any formula \(P \shpreceq Q\myif \phi \in G,\) we
  have that \(\GB \models P \shpreceq Q\myif \phi.\)
%
%\end{theorem}
\end{restatable}
Note that still all goals in \(G\) must be proven to trust any of them.
\begin{example}
  \label{ex:running2}
  Continuing Example~\ref{ex:running}, we have shown that
  \(G, B \vdash^0 G\) and that \(G^{-1}, B^{-1} \vdash^0 G^{-1}\) in
  the sense of partial simulation (with \(\shprec\) changed to
  \(\shpreceq\) in \(G\) and \(B\)) and therefore
  \(\fsh{f} \shpreceq \fsh{F}\) and \(\fsh{F} \shpreceq \fsh{f}\) for
  \(\fsh{\vN \geq 0}\) (i.e., \( \fsh{f} \) and \( \fsh{F} \) are
  partially equivalent).
\end{example}
\emph{Implementation.} We have a prototype implementation of the two
proof systems in the \RMT{} tool
(\url{http://profs.info.uaic.ro/~stefan.ciobaca/rmteq}). \RMT{}
implements order-sorted logically constrained term rewriting, that is
rewriting of mixed terms, which contains both free symbols and symbols
in some theory solvable by an SMT solver. \RMT{} relies on \ZTREI, and
therefore any combination of theories solvable by \ZTREI can be used
(e.g., bitvectors, arrays, LIA, etc.). \RMT{} already supported
reachability proofs~\cite{CiobacaIJCAR2018}. In this paper, we extend
it with axiomatized symbols and we implement the algorithms for checking
partial and full simulation in Section~\ref{sec:proofsystem}. To prove
equivalence, we check two-way simulation. The implementation for
checking partial/full simulation works in two phases:
\begin{itemize}[wide, labelwidth=!, labelindent=0pt]
\item[\(\bullet\)] the \textit{left} phase implements the
  \textsc{Step} rule: to prove \(L \shprec R\myif \phi\)
  (\(L \shpreceq R\myif \phi\)), \RMT{} finds all symbolic successors
  of \(L\) and, for each successor \(L'\myif \phi'\), attempts to
  prove \(L' \shprec R\myif \phi \fsh{\land} \phi'\)
  (\(L' \shpreceq R\myif \phi \fsh{\land} \phi'\)). If \(L'\) unifies with
  the lhs of either a base case or a circularity goal (e.g., there is
  a chance to reach the base case or a circularity), the algorithm
  moves to the \textit{right} phase.
\item[\(\bullet\)] the \textit{right} phase implements the \textsc{Base}
  and \textsc{Circ} proof rules: to prove \(L \shprec R\mbox{
    \(\textit{if}\) }\phi\) (\(L \shpreceq R\myif \phi\)),
  the symbolic successors of \(R\) are searched, in order to find a
  constraint \(\phi^1\) s.t. \(L \shprec R'\myif \fsh{\msh{\phi} \land
  \msh{\phi^1}}\) (\(L \shpreceq R'\myif \phi \fsh{\land} \phi^1\)) are
  either base cases or circularities for some symbolic successor \(R'\)
  of \(R\myif \phi\). If \(\phi^1\) is valid, the proof is
  done. Otherwise, the \textit{left} phase resumes, limiting the
  search space to \(\fsh{\lnot} \phi^1\).
\end{itemize}
In both phases, we fix a user-settable bound on the number of symbolic
steps (by default, \(100\)). If the bound is reached, then the current
branch of the proof fails. Unification modulo
axiomatized symbols is not fully implemented. Instead, we use a couple
of heuristics to handle it the cases of interest:
\begin{enumerate*}
\item in order to compute symbolic successors of a term (possibly
  containing axiomatized symbols), we first unroll the definition of
  the axiomatized symbols and
\item in order to check whether the current goal is an instance of the
  base cases or of a circularity (rules \textsc{Base} and
  \textsc{Circ}), we perform a bounded search with the equations of
  the symbol.
\end{enumerate*}

\emph{Examples}. We use the \IMP{} language defined in
Section~\ref{sec:semantics}. We have worked out the following
equivalence examples using our method:
\begin{enumerate*}%[wide, labelwidth=!, labelindent=0pt]
\item[\(\bullet\)] We show that \(\fsh{f(\vN) \shprec F(\vN, 0, 0)}\),
  \(\fsh{f(\vN) \shpreceq F(\vN, 0, 0)}\) and that
  \(\fsh{F(\vN, 0, 0) \shpreceq f(\vN)}\) under the constraint
  \(\fsh{\vN \geq 0}\) for our running example in the language
  \IMP{1}. Our method is not sufficiently powerful to show
  \(\fsh{F(\vN, 0, 0) \shprec f(N)\myif\vN \geq 0}\).
\item[\(\bullet\)] In \IMP{2}, none of
  \(\fsh{f(\vN) \shprec F(\vN, 0, 0)}\),
  \(\fsh{f(\vN) \shpreceq F(\vN, 0, 0)}\),
  \(\fsh{F(\vN, 0, 0) \shpreceq f(\vN)}\),
  \(\fsh{F(\vN, 0, 0) \shpreceq f(\vN)}\) hold under the constraint
  \(\fsh{\vN \geq 0}\), and therefore the proofs of these goals
  (correctly) fail.
\item[\(\bullet\)] Our method can prove programs in two different
  languages as well. We show that \(\fsh{f(\vN)}\), interpreted in \IMP{1}, is
  partially equivalent to \(\fsh{F(\vN, 0, 0)}\), interpreted in \IMP{2}, when
  \(\fsh{\vN \geq 0}\). For one direction (\(\fsh{f \shprec F}\)), we establish full
  simulation; for the other direction, just partial simulation.
\item[\(\bullet\)] We show that a {\tt while} loop is partially equivalent
  to a recursive function, when both compute the sum of the first \(\vN\)
  naturals.
\item[\(\bullet\)] We prove full equivalence for an instance of loop
  unswitching, showing that our method can handle programs that are
  structurally unrelated.
\end{enumerate*}
More details on these examples can be found in
Appendix~\ref{app:examples}.
\section{Proving Equivalence of Program Schemas}
\label{sec:programschemas}
The method that we have introduced in Section~\ref{sec:proofsystem}
can be used to show full/partial simulations between \emph{symbolic
  program configurations}. Symbolic program configurations can contain
variables. In our running example, we prove two-way simulations
between the symbolic configurations
\(\cfgimp{[\lacall{f(\vN)}]}{\venv}{\funcsct}\) and
\(\cfgimp{[\lacall{F(\vN, 0, 0)}]}{\venv}{\funcsct}\), under the
constraint \(\fsh{\vN \geq 0}\). The variables \(\venv\) and \(\vN\)
occur in both configurations, denoting the fact that their value is
shared in both configurations, and \(\funcsct\) is the map defined in
Example~\ref{ex:running}.

It is also possible to use variables of sort \(\sexp\) (variables
standing for program expressions or statements) in a symbolic
configuration. For example, we might want to prove that
\(\cfgimp{[\ve]}{\venv}{\funcsct}\) and
\(\cfgimp{[\ve]}{\venv}{\funcsct}\) (both configurations are the same)
are equivalent. The variable \(\ve\) of sort \(\sexp\) denotes a
program expression. We call such variables, which stand for parts of
the program (such as \(\ve\)), \emph{structural variables}, in
constrast to other variables (such as \(\vN\) or \(\venv\)).

Our proof method fails when trying to prove equivalence between
symbolic configuration containing structural variables\footnote{It may
  seem surprizing that our system cannot prove an expression
  equivalent to itself, but our method is \emph{operational}, not
  axiomatic -- there is no rule for reflexivity.}. The issue is in the
\textsc{Step} rule, which tries to compute the possible symbolic
successors of \(\cfgimp{[\ve]}{\venv}{\funcsct}\). When computing all
successors, a case analysis on the structural variable \(\ve\) is
performed: \(\ve\) might be an addition
\(\plus{\vesub{1}}{\vesub{2}}\), an {\tt ite} statement
\(\laite{\vesub{0}}{\vesub{1}}{\vesub{2}}\), etc. This case analysis occurs
ad-infinitum, and no real progress is made in the proof.

We call such symbolic program configurations, with variables denoting
program parts, \emph{program configuration schemas} or simply
\emph{configuration schemas}. We also use \emph{program schema} when
we refer just to the part of the configuration holding the
program. The naming is because such a program schema denotes several
programs, depending on how the structural variables are instantiated.

Therefore, even if our proof system can technically handle program
configuration schemas, it cannot be used directly to show interesting
properties of such schemas. In particular, it is not possible to
\emph{directly use} our proof system to show the correctness of
program optimizations such as the constant propagation optimization
presented in Figure~\ref{fig:exschemas}.
\begin{figure}[t]
  \hrule
  {
    \centering
    \(\begin{array}{ccc}
        \textit{Initial Program} & \phantom{space} & \textit{Optimized Program} \\
        \fsh{\laseq{\laseq{\laassign{\cxsub{1}}{\vesub{1}\;}}{\;\vstsub{1}\;}}{\laassign{\;\cxsub{2}}{\vesub{1}}}}
                                 & &
                                     \fsh{\laseq{\laseq{\laassign{\cxsub{1}}{\vesub{1}\;}}{\;\vstsub{1}\;}}{\laassign{\;\cxsub{2}}{\cxsub{1}}}}
      \end{array}\)
      \caption{\label{fig:exschemas}Two program schemas, where:
        \(\cxsub{1}, \cxsub{2} : \sid\) are two identifiers,
        \(\vesub{1} : \sexp\) is an expression, and
        \(\vstsub{1} : \sexp\) is a statement. The optimization is
        valid assuming that \(\vstsub{1}\) does not change
        \(\cxsub{1}\) or the program variables in \(\vesub{1}\).  }
   }
   \hrule
 \end{figure}
However, our proof system \emph{can directly prove} instances of this
optimization (i.e., for particular instantiations of
\(\vesub{1}, \vstsub{1}\)).

We show how our method can be easily extended to prove simulations
between program schemas. This extension crucially relies on the fact
that our proof method is parametric in the operational semantics.  We
consider two configuration schemas. In order to prove their
equivalence, we transform the structural variables into fresh
constants. We give semantics to the new constants by \emph{adding new
  rules} to the operational semantics of the language. These rules
capture the read-set of expressions and the read-set and write-set of
statements.

We explain this encoding on the example in
Figure~\ref{fig:exschemas}. We create fresh constants \(\cesub{1}\)
and \(\cstsub{1}\) of sort \(\sexp\). The constants are not considered
values of the language:
\(\bval(\cesub{1}) \rewrite \fsh{\bot}, \bval(\cstsub{1}) \rewrite
\fsh{\bot}\). In the two programs schemas above, we replace the
structural variables \(\vesub{1}\) and \(\vstsub{1}\) by the new
constants \(\cesub{1}\) and \(\cstsub{1}\), respectively. We say that
the constant \(\cesub{1}\) (\(\cstsub{1}\)) abstracts the variable
\(\vesub{1}\) (\(\vstsub{1}\)). After this abstraction, we obtain the following
configurations that we would like to prove equivalent:
\[\cfgimp{[\laseq{\laassign{\cxsub{1}}{\cesub{1}}}{\laseq{\cstsub{1}}{\laassign{\cxsub{2}}{\cesub{1}}}}]}{\venv}{\funcsct}\mbox{
    and }
  \cfgimp{\laseq{\laassign{\cxsub{1}}{\cesub{1}}}{\laseq{\cstsub{1}}{\laassign{\cxsub{2}}{\cxsub{1}}}}}{\venv}{\funcsct}.\]

A new problem when proving such abstracted configurations is that they
block whenever \(\cesub{1}\) or \(\cstsub{1}\) reach the top of the
evaluation stack, as there is no operational rule that describes their
semantics. Thus, a configuration like
\(\cfgimp{\cesub{1} \lisymb
  \msh{\ldots}}{\msh{\ldots}}{\msh{\ldots}}\) is stuck. We take
advantage of the fact that our proof systems for showing simulation
and equivalence are \emph{parametric} in the operational semantics and
we add semantic rules that specify the behaviors of \(\cesub{1}\) and
\(\cstsub{1}\).

The new rules formalize in a rigorous manner the notions of
\emph{read-set} and \emph{write-set}. The read-set of an expression is
the set of program variables that the expression is allowed to depend
on. In addition to a \emph{read-set}, a statement, such as \(\cstsub{1}\), also
has a \emph{write-set}, that is a set of program variables that the
statement is allowed to write to. Assume that the read-set of \(\cesub{1}\) is
\(\{ \cysub{2}, \cxsub{2} \}\). We formalize this read-set by adding to the
semantics the rule
\[\cfgimp{\cesub{1} \lisymb \vcs}{\venv}{\vfuncs} \rewrite
  \cfgimp{\ciesub{1}(\blookup(\venv, \cysub{2}), \blookup(\venv,
    \cxsub{2})) \lisymb \vcs}{\venv}{\vfuncs},\] where
\(\ciesub{1} : \sint \times \sint \to \sint\) is a fresh builtin
uninterpreted function symbol. This rule models the fact that
\(\cesub{1}\) terminates, is deterministic, and evaluates to a value
\(\ciesub{1}(\blookup(\venv, \cysub{2}), \blookup(\venv, \cxsub{2}))\)
that only depends on the program variables \(\cysub{2}\) and
\(\cxsub{2}\). In general, for an expression having read-set
\(x_1, \ldots, x_n\), we add a fresh \(n\)-ary builtin symbol and we
use it in a rule such as the one above.

\begin{figure}[t]
  \hrule
  {
    \centering
    \(
    \begin{array}{ll}
      \mbox{New constructors:} & \cesub{1} : \sexp, \cstsub{1} : \sexp \\
      \hline
      \mbox{New builtins:} & \ciesub{1} : \sint^2 \to \sint, \cistsub{1} : \sint^4 \to \sint \\
      \hline
      \mbox{New rules:} & \cfgimp{
                          \cesub{1} \lisymb \vcs}{\venv}{\vfuncs} \rewrite
                          \cfgimp{\ciesub{1}(\venv[\cysub{2}], \venv[\cxsub{2}]) \lisymb \vcs}{\venv}{\vfuncs}\\
      \cfgimp{\cstsub{1} \lisymb \vcs}{\venv}{\vfuncs} \rewrite &
                                                             \cfgimp{s}{\venv[\cysub{1} \mapsto \cistsub{1}(\venv[\cysub{1}],
                                                             \venv[\cysub{2}], \venv[\cxsub{1}], \venv[\cxsub{2}])]}{\vfuncs}
    \end{array}
    \)
  }
  \caption{\label{fig:exabstraction}Abstraction process required to
    prove the optimization described in
    Figure~\ref{fig:exschemas}. The notation \(\menv[x]\) and
    \(\menv[x \mapsto w]\) are short for \(\blookup(\menv, x)\) and
    \(\bupdate(\menv, x, w)\), respectively.}
  \hrule
\end{figure}
To model write-sets of statements, we add rules that modify in the
environment only the program variables that are written to. For
example, assume that the read-set of \(\cstsub{1}\) is \(\cysub{1}, \cysub{2}, \cxsub{1}, \cxsub{2}\) and
that the write-set of \(\cstsub{1}\) is \(\cysub{1}\). We add the following rule:
% \[\begin{array}{l}
\(\cfgimp{\cstsub{1} \lisymb \vcs}{\venv}{\vfuncs} \rewrite
    % \\ \qquad 
\cfgimp{s}{\bupdate(\venv, \cysub{1}, \msh{nv})}{\vfuncs},\)
      %     \end{array}\]
where
\(\msh{nv} = \cistsub{1}(\venv[\cysub{1}],
\venv[\cysub{2}],\allowbreak \venv[\cxsub{1}], \venv[\cxsub{2}]) \),
the symbol \(\cistsub{1} : \sint^4 \to \sint\) is a fresh builtin and
\(\menv[x]\) is short for \(\blookup(\menv, x)\). This rule models
that \(\cstsub{1}\) terminates, is deterministic, writes to
\(\cysub{1}\) only and the value computed and written to \(\cysub{1}\)
only depends on \(\cysub{1}, \cysub{2}, \cxsub{1}, \cxsub{2}\). We
summarize the abstraction process in Figure~\ref{fig:exabstraction}.

With the encoding in Figure~\ref{fig:exabstraction}, our proof system
and \RMT{} can show the equivalence of the two programs in
Figure~\ref{fig:exschemas} and therefore the correctness of this
optimization in the context of \IMP{}. \emph{Comparison with CORK and
  PEC.}  We show that our approach also generalizes to a number of
compiler optimizations previously discussed in the context of the
\texttt{CORK}~\cite{LopesMonteiro} and PEC~\cite{KunduPLDI2009}
optimization correctness verification tools. The comparison is shown
Figure~\ref{fig:pec-cork-rmt}. We use two annotations for special
cases:
\begin{enumerate*}
\item the mark \(\ocircle\) denotes that, even if the two programs
  schemas are functionally equivalent, there is no simulation of one
  by the other -- instead, we prove two different simulations, one for
  each of the two output variables;
\item \(\square\) denotes that we have used an upper bound on one of the
  program variables -- the bound is not a weakness of \RMT{} or of our
  proof method, but of the fact that the SMT solver that we use, \ZTREI,
  does not handle non-linear arithmetic well enough.
\end{enumerate*}
Using another SMT solver for non-linear integer arithmetic, like CVC4,
could potentially allow us to prove these examples, marked with
\(\square\), in the unbounded case as well. In comparison with existing
approaches, we can prove the correctness of an optimization (Loop
flattening) that the two previous
approaches~\cite{KunduPLDI2009,LopesMonteiro} cannot. However, our
tool is not automated when loops are involved and must be guided by
\emph{helper circularities}, as explained in Example~\ref{ex:running}
on Page~\pageref{ex:running}. We describe the proof of each
optimization in turn in Appendix~\ref{app:optimization}. In this
appendix, we also develop a methodology that help us find such helper
circularities, giving evidence that our equivalence/simulation checker
could be automated for optimization correctness verification purposes.
\begin{figure}[t]
\hrule
  \centering
  \begin{tabular}{|l | c | r | c r |}
    \hline
    Optimization & \texttt{PEC} & \texttt{CORK} &  & \RMT{} \\
    \hline
    Code hoisting & \(\checkmark\)  & 0.32s&  & 0.41s  \\
    Constant propagation & \(\checkmark\) & 0.33s & & 0.31s  \\
    Copy propagation & \(\checkmark\) & 0.33s & & 0.26s  \\
    If-conversion & \(\checkmark\) & 0.34s & & 0.48s  \\
    Partial redundancy & \(\checkmark\) & 0.34s & & 0.75s  \\
     elimination & & & & \\
    Loop invariant & \(\checkmark\) & 3.48s & & 3.79s  \\
     code motion & & & & \\
    Loop peeling & \(\checkmark\) & 3.26s & & 0.97s  \\
    Loop unrolling & \(\checkmark\) & 12.17s & & 7.09s  \\
    Loop unswitching & \(\checkmark\) & 8.19s & & 4.71s  \\
    \hline
  \end{tabular}
  \begin{tabular}{|l | c | r | c r |}
    \hline
    Optimization & \texttt{PEC} & \texttt{CORK} &  & \RMT{} \\
    \hline
    Software pipelining & \(\checkmark\) & 8.02s & & 3.56s  \\
    Loop fission & \(\checkmark_p\) & 23.45s & \(\ocircle\) & 10.40s  \\
    Loop fusion & \(\checkmark_p\) & 23.34s & \(\ocircle\) & 9.67s  \\
    Loop interchange & \(\checkmark_p\) & 29.30s & \(\square\) & 108.63s  \\
    Loop reversal & \(\checkmark_p\) & 8.41s & & 2.70s  \\
    Loop skewing & \(\checkmark_p\) & 8.50s & & 7.68s  \\
    Loop flattening & \(\times\) & \(\times\) & \(\square\) & 8.14s  \\
    Loop strength & \(\times\) & 5.63s & & 5.26s  \\
    reduction & & & &  \\
    \hline
    Loop tiling 01 & \multirow{2}{*}{\(\times\)}  & \multirow{2}{*}{10.94s} & & 25.41s  \\
    Loop tiling 02 &  &  & \(\square\) & 21.58s  \\
    \hline
  \end{tabular}
  \caption{\label{fig:pec-cork-rmt} Optimizations on which we compare
    the tools PEC, CORK, and \RMT{} (our prototype). Columns one-three are
    due to Lopes and Monteiro~\cite{LopesMonteiro} (\(\checkmark_p\) means
    PEC needs a heuristic called \emph{permute}). The third column is
    based on our own benchmark. The annotations \(\ocircle\) and
    \(\square\) are described in the main text.}
  \hrule
\end{figure}
% citation for nonlinear integer arithmetic \cite{liberti2019undecidability}
%
\section{Related Work}
\label{sec:related}
\(\bullet\)~In the series of
papers~\cite{PittsLICS1996,PittsICALP1998,PittsMSCS2000,PittsAS2000},
Pitts was one of the first to propose the use of operationally-based
notions of contextual equivalence%%  for various languages with several
%% features
. The differences to our work is that we only consider functional
equivalence, and not contextual equivalence, but in our approach the
operational semantics can be varied. We also explicitly allow for
nondeterminism and there is no need to define explicitly a logical
relation for the entire language: instead, the user defines a
simulation relation that depends only on the particular programs to be
shown equivalent. In~\cite{PittsAS2000}, the \emph{frame stack}
approach for small-step semantics that we use is introduced. The same
style of using a frame stack was popularized by the K
framework~\cite{StefanescuOOPSLA2016} in the rewriting based semantics
of several large languages~\cite{EllisonPOPL2012,BogdanasPOPL2015}. We
make extensive use of this frame-stack technique, which enables
simpler equivalence proofs%% than the usual inductive definitions of
%% small-step of big-step SOS
.
\(\bullet\)~Logical relations and bisimulation can be used to prove
contextual equivalence.  Bisimulation techniques such
as~\cite{SangiorgiTOPLAS2011} are usually language dependent and
proofs of congruence and other properties need to be established
independently. Language features such as higher-order functions are
handled by enhancing the bisimulation with an \emph{environment}
holding the current knowledge of the observer. Instead, by reducing
the scope to functional equivalence instead of contextual equivalence,
we allow to use simpler simulation relations that depend only on the
particular pair of programs to be proven equivalent (there is no need
to prove congruence in our case). Logical relations techniques such
as~\cite{DreyerLMCS2011} can be used to prove contextual equivalences
for various languages. Logical relations can also used in mechanized
frameworks for separation logic such as Iris~\cite{JungJFP2018} in
order to handle contextual equivalence in the presence of state
(see~\cite{TimanyPOPL2017}) or continuations
(see~\cite{TimanyPHD2018}). However, logical relations may be
difficult to adapt to different languages and may require additional
indexing to account for language features. Mechanized proofs may be
quite long and tedious. Game semantics can be used to reason
denotationally (see, e.g.,~\cite{MurawskiFTPL2016}) about contextual
equivalence, but it does not enjoy good algorithmic
properties~\cite{MurawskiFMSD2018}; however, proof search can be
implemented for languages with higher-order functions and effects, as
shown by Jaber~\cite{JaberPOPL2020}.
\(\bullet\)~Several relational Hoare logics were proposed
(e.g.,~\cite{BentonPOPL2004,BentonPPDP2018,AguirreICFP2017}) for
reasoning about pairs of programs. Typically, such logics are
developed for a particular language and can usually be used to prove
equivalence of syntactically similar programs. For example, they
usually assume that two matching while loops will both take the same
number of steps. In contract, the logic that we propose can also be
used to reason about structurally dissimilar programs%% , as illustrated
%% in the fourth example in Section~\ref{sec:implementation}
. Relational
higher-order logic, introduced in~\cite{AguirreICFP2017}, allows both
synchronous and asynchronous reasoning about pairs of programs in a
higher-order lambda calculus.
It can be used to show functional equivalence, but also for other
properties such as relational cost analysis. In contrast to RHOL, the
logic that we propose here is formally less expressive (RHOL is as
expressive as HOL). However, unlike RHOL, it is simpler to use and mechanize.%% it can be used to reason
%% about impure higher-order programs, as we demonstrated in
%% Section~\ref{sec:implementation}.
Relational separation logic~\cite{YangTCS2007} enhances relational
Hoare logic with the ability to reason about the
heap. In~\cite{BanerjeeFSTTCS2016}, the authors propose a relational
logic with a framing rule that enables a SMT-friendly encoding of the
heap, but also enhances the ability to reason about less structurally
related programs. Unlike these logics, we do not currently handle the
heap, but our proof system is much simpler, because most of the
complexity of reasoning about the language features goes to the LCTRS
encoding the language semantics. Also, in our case, it is much simpler
to experiment with variations of the language semantics, as explained
in Sections~\ref{sec:proofsystem} and~\ref{sec:programschemas}. A
concept close to relational Hoare logic is that of
product-program~\cite{BartheJLAMP2016}, which are programs that mimic
the behavior of two programs; they allow to reduce relational
reasoning to reasoning about a single program. In our work, there is
no need to construct such product programs. Such a product
construction is possible in a rewriting-based scenario as
well~\cite{CiobacaSYNASC2014,CiobacaWADT2014}. Compared to all
approaches above, the logic that we introduce in this paper has the
advantage that the underlying operational semantics of the language
can be easily changed. This makes it easy to experiment with various
settings. In our examples, we show how we go from a semantics with an
unbounded stack to a semantics with a bounded stack, but other
variations of interest could be using fixed-size integers (bitvectors)
instead of unbounded integers, enabling or disabling language features
such as exceptions, introspection, etc. in order to check how each
affects functional equivalence. In~\cite{KunduPLDI2009}, an
implementation of a parametrized equivalence prover is presented and
we compare against the tool in Section~\ref{sec:programschemas}. Grimm
et al.~\cite{GrimmCPP2018} propose a general method for relational
proofs based on encoding the state transformation as a monad in the {\tt F*}
proof assistant.  After encoding, relational proofs then require user
interaction, although significant parts are solved directly by an SMT
solver. Maillard et al.~\cite{MaillardPOPL2020} show how to generalize
this to arbitrary monadic effects.
\(\bullet\)~In~\cite{StrichmanFMSD2015RV}, Chaki et al. propose a new
definition of equivalence suitable for nondeterministic programs,
extending the usual definition of partial equivalence for
deterministic programs, and introduce sound proof rules for regression
verification of multithreaded programs. Our definition of equivalence,
defined as two-way simulation, is implied by the definition of partial
equivalence proposed here -- the difference is that our notion of
equivalence allows a terminating execution on some input to be
simulated by an infinite execution of the other program on the same
input. However, we also additionally propose a definition for full
equivalence suitable for a non-deterministic setting; this definition
is more involved than the usual definition of full equivalence as
partial equivalence plus mutual termination as outlined
in~\cite{ElenbogenHVC2012}, since a non-deterministic program could
have both terminating and non-terminating runs starting with the same
input. Felsing et al.~\cite{FelsingASE2014} propose an automated
method for regression verification. Lahiri et al.~\cite{LahiriCAV2012}
present a method based on translation into the intermediate
verification language Boogie for checking \emph{semantical
  differences} between programs. A technique for automated discovery
of simulation relations is proposed
in~\cite{FedyukovichLPAR2015}. Their technique is automated using
\ZTREI as a solver. Techniques based on an efficient encodings of the
relational property as a set of constrained Horn clause are described
in~\cite{DeAngelisSAS2016}. Another technique for automatic proving of
equivalences for procedural programs that is also based on LCTRSs is
proposed in~\cite{FuhsTOCL2017}. Unlike our approach,
in~\cite{FuhsTOCL2017} the two C-like programs are translated by a
tool called C2LCTRS into LCTRSs. An advantage of their approach is
automation by using a constrained version of the well-known technique
of rewriting induction. However, the C2LCTRS tool contains an implicit
semantics of the C-like language and therefore, unlike in our work,
variations of the semantics that change various language features
(like stack size, integer semantics, etc.) require changing the
tool. Moreover, in~\cite{FuhsTOCL2017}, the two programs are also
assumed to be deterministic. Even if we changed the C2LCTRS program to
explicitly model a stack, constrained rewriting induction would fail
in general to find an equivalence proof between two programs such as
example in Section~\ref{sec:intro}, as the simulation relation
requires \emph{axiomatized symbols} to state.
\(\bullet\)~Early ideas on adding logical constraints to deduction rules in
general date back to the 1990s, in work like~\cite{Kirchner1990}
and~\cite{DarlingtonCTRS1990}. Logically constrained term rewriting
systems, which combine term rewriting and SMT constraints are
introduced in~\cite{KopFROCOS2013}. LCTRSs generalize previous
formalisms like TRSs enriched with numbers and Presburger constraints
(e.g., as in~\cite{FalkeRTA2008}) by allowing arbitrary theories that
can be handled by SMT solvers. Rewriting modulo SMT is introduced
in~\cite{RochaJLAMP2017} for analyzing open
systems. In~\cite{BaeFACS2017}, the authors introduce guarded terms,
which generalize logically constrained terms. A narrowing calculus for
constrained rewriting is introduced
in~\cite{AguirrePPDP2017}. In~\cite{NagaoPPDP2016}, an approach to
proving \emph{inequalities} based on constrained rewriting induction
is proposed. Finally, logically constrained rewriting enjoys
completion procedures, as shown in~\cite{WinklerFSCD2018}.
\(\bullet\)~\emph{Our own related work.} We first considered
semantics-based equivalence in~\cite{LucanuFAOC2015} for symbolic
programs in the context of the K framework
in~\cite{StefanescuOOPSLA2016}, but for a notion of behavioural
equivalence of deterministic programs. In~\cite{CiobacaFAOC2016}, we
give a semantics-based proof system for full equivalence. Our present
work improves on this by adding axiomatized symbols, using different
notions of equivalence that handle non-determinism and are more
modular (we now also test for one-way simulation) and providing a
working implementation based on LCTRSs with several novel
examples. Most of the infrastructure required for LCTRSs is based on
our earlier work on proving reachability in
LCTRSs~\cite{CiobacaIJCAR2018} and solving unification modulo
builtins~\cite{CiobacaWOLLIC2018}. However, the present work includes
\emph{axiomatized symbols}, which pose new technical
challenges.
%\stefan{TODO:newpapers DONE:notime}

% \paragraph{Other}

% TODO Lahiri et al., Dijkjstra monads/F*.
%% Guilhem Jaber

\section{Conclusion and Future Work}
\label{sec:conclusion}

We have introduced and implemented in \RMT{} a new method for proving
simulation and equivalence in languages whose semantics are defined by
LCTRSs in frame stack style. Our method allows to easily check program
equivalence in various settings, such as unbounded versus bounded
stack, arbitrary precision versus fixed size integers, etc. To express
simulation relations, we enrich standard LCTRSs with \emph{axiomatized
  symbols}, which raise new research questions such as
\emph{unification modulo axiomatized symbols}. We also generalize existing
definitions for full/partial equivalence. Our approach allows for
nondeterminism in the definitions and in the proofs, but we currently
do not exploit this, as we only have simple examples. We also show an
advantage of an operational semantics-based approach: we can easily
model read-sets and write-set and prove simulation/equivalence of
program schemas.

As future work, we would like to apply our methods to more challenging
concurrent programs and to realistic language definitions, available
as part of the K framework~\cite{EllisonPOPL2012,BogdanasPOPL2015}. We
would also like to integrate an external termination checker to handle
full equivalence better. Other directions for future work include
relational cost analysis, as in~\cite{RadicekPOPL2017}, possibly by
simply using an appropriate set \(\GB\) of base cases, and generalizing
to contextual equivalence, possibly by extending the techniques in
Section~\ref{sec:programschemas}.

\newpage

\bibliographystyle{splncs04}
\bibliography{refsacm}

\clearpage
\appendix

\clearpage
\section{Example of Program Execution}
\label{app:examplesemantics}

Here is an example of how evaluation proceeds for the program
\[ \fsh{\laassign{\cx}{\lacall{\pvar{f}(0)}}} \]
in an initial environment mapping the program identifier \(\cx\) to
\( \fsh{12}\), and a function map \( \mfuncs \), mapping the program
identifier \( \fsh{f} \) to
\[\fsh{(\lambda {\cy}.\laite{\cy > 5}{\plus{\cy}{\cx}}{0})}\] according to the semantics of \IMP{}
introduced in Section~\ref{sec:semantics}: \begin{enumerate}

\item \cfgimp{[\laassign{\cx}{\lacall{f(10)}]}
  }{\cx \mapsto 12
}{\mfuncs} \rewrite

\item \cfgimp{\lacall{f(10)} \lisymb \laassign{\cx}{\square{}} \lisymb
  []}{\cx \mapsto 12}{\mfuncs} \rewrite

\item \cfgimp{\lacall{f} \lisymb \lacall{\square{} (10)} \lisymb
  \laassign{x}{\square{}} \lisymb []
}{x \mapsto 12}{\mfuncs} \rewrite

\item \cfgimp{\blookup(f, {\mfuncs}) \lisymb
  \lacall{\square{} (10)} \lisymb \laassign{\cx}{\square{}} \lisymb []}{\cx
  \mapsto 12}{\mfuncs} \fsh{=}

\item \cfgimp{(\lambda \cy.\laite{\cy > 5}{\plus{\cy}{\cx}}{0}) \lisymb \lacall{\square{}
    (10)} \lisymb \laassign{\cx}{\square{}} \lisymb []}{\cx \mapsto 12}{\mfuncs} \rewrite

\item \cfgimp{\lacall{(\lambda \cy. \laite{\cy > 5}{\plus{\cy}{\cx}}{0})} (10) \lisymb
  \laassign{\cx}{\square{}} \lisymb []}{\cx \mapsto 12}{\mfuncs} \rewrite

\item \cfgimp{\asubst(\cy, 10, ( \laite{\cy > 5}{\plus{\cy}{\cx}}{0})) \lisymb \laassign{\cx}{\square{}} \lisymb []}{\cx \mapsto 12}{\mfuncs} \rewrite

\item \cfgimp{\laite{10 > 5}{\plus{10}{\cx}}{0} \lisymb \laassign{\cx}{\square{}} \lisymb []}{\cx
  \mapsto 12}{\mfuncs} \rewrite

\item \cfgimp{10 > 5 \lisymb \laite{\square{}}{\plus{10}{\cx}}{0} \lisymb \laassign{\cx}{\square{}} \lisymb []}{\cx \mapsto  12}{\mfuncs} \rewrite
    
\item \cfgimp{\top \lisymb \laite{\square{}}{\plus{10}{\cx}}{0} \lisymb \laassign{\cx}{\square{}} \lisymb []}{\cx \mapsto 12}{\mfuncs} \rewrite

\item \cfgimp{\laite{\top}{\plus{10}{\cx}}{0} \lisymb \laassign{\cx}{\square{}} \lisymb []}{\cx
  \mapsto 12}{\mfuncs} \rewrite

\item \cfgimp{\plus{10}{\cx} \lisymb \laassign{\cx}{\square{}} \lisymb []}{\cx \mapsto 12}{\mfuncs} \rewrite

\item \cfgimp{\cx \lisymb \plus{10}{\square{}} \lisymb \laassign{\cx}{\square{}} \lisymb []}{\cx \mapsto 12}{\mfuncs} \rewrite

\item \cfgimp{\blookup(\cx,{\cx \mapsto 12})
  \lisymb \plus{10}{\square{}} \lisymb \laassign{\cx}{\square{}} \lisymb []}{\cx \mapsto 12}{\mfuncs} =

\item \cfgimp{12 \lisymb \plus{10}{\square{}} \lisymb \plus{\cx}{\square{}} \lisymb []}{\cx \mapsto 12}{\mfuncs} \rewrite

\item \cfgimp{\plus{10}{12} \lisymb \laassign{\cx}{\square{}} \lisymb []}{\cx \mapsto 12}{\mfuncs} \rewrite

\item \cfgimp{22 \lisymb \laassign{\cx}{\square{}} \lisymb []}{\cx \mapsto
  12}{\mfuncs} \rewrite

\item \cfgimp{[\laassign{\cx}{22}]}{\cx \mapsto 12}{\mfuncs} \rewrite

\item \cfgimp{[]}{\cx \mapsto 22}{\mfuncs} \notrewrite.

\end{enumerate}

%%% Local Variables:
%%% mode: latex
%%% TeX-master: "submission.tex"
%%% End:

\clearpage
\section{Proofs}
\label{app:proofs}

This section includes the proof for the soundness theorems. The proof
principle used is a kind of parametric coinduction.  We start with two
lemmas that state coinductive characterizations for full simulation
and partial simulation, respectively.

\begin{lemma}
\label{lem:fsim-coind}
 \(\GB \models P \fsh{\prec} Q\myif \phi\) iff
\(\llbracket P \fsh{\prec} Q\myif \phi \rrbracket\subseteq\nu\,S\,.\,\mathit{fsim}(S)\), where
\begin{align*}
\mathit{fsim}(S)= \{(\GP,\GQ)\mid\,& \GP{\Downarrow}\fsh{\limplies} \exists\,\GQ'.\,\GQ\rewrite_{\R_L}^* \GQ'\fsh{\land}(\GP,\GQ')\in \GB\,\fsh{\land}\\
&\fsh{\lnot} \GP{\Downarrow}\fsh{\limplies}\forall\,\GP'.\,\GP\rewrite_{\R_L} \GP'\fsh{\limplies}\\
&\phantom{\fsh{\lnot} \GP{\Downarrow}\fsh{\limplies}\forall\,\GP'\,.}\exists\,\GQ'.\, \GQ\rewrite_{\R_R}^* \GQ'\fsh{\land}(\GP',\GQ')\in S\}
\end{align*}
and \(\GP{\Downarrow}\) means that \(\GP\) is a \({\rewrite_{\R_R}{}}\)\!\!-irreducible configuration, i.e., \(\forall \GP'\,.\,\GP \rewrite_{\R_L}^* \GP'\fsh{\limplies} \GP=\GP'\).
\end{lemma}
\begin{proof}%[Sketch]
We first notice that we have \(\GB \models P \fsh{\prec} Q\myif \phi\) iff \(\llbracket P \fsh{\prec} Q\myif \phi \rrbracket\subseteq {S^\fsh{\prec}}\), where
\begin{align*}
{S^\fsh{\prec}}=\{(\GP,\GQ)\mid \forall\,\GP'.\,&\GP\rewrite_{\R_L}^* \GP'\fsh{\land} \GP'{\Downarrow}\fsh{\limplies}\\
&\exists\,\GQ'.\, \GQ\rewrite_{\R_R}^* \GQ'\fsh{\land}(\GP',\GQ')\in \GB\}
\end{align*}
\\%
\(\Rightarrow\). We show that \({S^\fsh{\prec}}\)
is post-fixed point for \(\mathit{fsim}\), i.e., \({S^\fsh{\prec}}\subseteq \mathit{fsim}({S^\fsh{\prec}})\). Let \((\GP,\GQ)\in {S^\fsh{\prec}}\). There are two cases:
\begin{enumerate}
\item \(\GP{\Downarrow}\).  Then \(\exists\,\GQ'\,.\,\GQ\rewrite_{\R_L}^* \GQ'\fsh{\land}(\GP,\GQ')\in \GB\) by the definition of \({S^\fsh{\prec}}\).
\item \(\fsh{\lnot} \GP{\Downarrow}\). Let \(\GP'\) be such that \(\GP\rewrite_{\R_L} \GP'\). We show that \((\GP',\GQ)\in {S^\fsh{\prec}}\). Let \(\GP''\) be such that  \(\GP'\rewrite_{\R_L}^* \GP''\fsh{\land} \GP''{\Downarrow}\).  Then \(\GP\rewrite_{\R_L}^* \GP''\fsh{\land} \GP''{\Downarrow}\), which implies \(\exists\,\GQ'\,.\,\GQ\rewrite_{\R_L}^* \GQ'\fsh{\land}(\GP'',\GQ')\in \GB\) (recall that \((\GP,\GQ)\in {S^\fsh{\prec}}\)). Hence \((\GP',\GQ)\in {S^\fsh{\prec}}\).
Since \(\GP'\) is arbitrary, it follows that 
\(\forall\,\GP'.\,\GP\rewrite_{\R_L} \GP'\fsh{\limplies}(\GP',\GQ)\in {S^\fsh{\prec}}\).
\end{enumerate}
From the above case analysis we may conclude \((\GP,\GQ)\in \mathit{fsim}({S^\fsh{\prec}})\).
\\
\(\Leftarrow\). We show that if \(S\) is a post-fixed point for \(\mathit{fsim}\) then \(S\subseteq{S^\fsh{\prec}}\).
 Let \((\GP,\GQ)\in S\) and assume that \(\GP\rewrite_{\R_L}^* \GP'\fsh{\land} \GP'{\Downarrow}\).
In order to show that \((\GP,\GQ)\in {S^\fsh{\prec}}\) we have to prove that \(\exists\,\GQ'\,.\,\GQ\rewrite_{\R_L}^* \GQ'\fsh{\land}(\GP',\GQ')\in \GB\).
We have two cases:
\begin{enumerate}
\item \(\GP=\GP'\). Then \(\exists\,\GQ'\,.\,\GQ\rewrite_{\R_L}^* \GQ'\fsh{\land}(\GP,\GQ')\in \GB\) by the first part of the definition of \(\mathit{fsim}\), which implies \((\GP,\GQ)\in{S^\fsh{\prec}}\).
\item \(\GP\not=\GP'\) (which implies \(\fsh{\lnot} \GP{\Downarrow}\)). 
Let \(\GP'_1\) be s.t.  \(\GP\rewrite_{\R_L} \GP'_1\rewrite_{\R_L}^* \GP'\). We obtain \((\GP'_1,\GQ'_1)\in S\) by the second part of the definition of \(\mathit{fsim}\). We repeat the same reasoning until we obtain \(\GP'_n=\GP'\) and \((\GP'_n,\GQ'_n)\in S\). Then \(\exists\,\GQ'\,.\,\GQ'_n\rewrite_{\R_L}^* \GQ'\fsh{\land}(\GP'_n,\GQ')\in \GB\) is proved in a similar way to the first case. Since \(\GP'_1\) is arbitrarily chosen, it follows that \[\forall\,\GP'.\,\GP\rewrite_{\R_L}^* \GP'\fsh{\land} \GP'{\Downarrow}\fsh{\limplies}\exists\,\GQ'.\, \GQ\rewrite_{\R_R}^* \GQ'\fsh{\land}(\GP',\GQ')\in \GB,\] which implies \((\GP,\GQ)\in{S^\fsh{\prec}}.\)
\end{enumerate}
\qed 
\end{proof}

\begin{corollary}
1. \({S^\fsh{\prec}}=\nu\,S\,.\,\mathit{fsim}(S)\).\\
2. \(\GB\) is a post-fixed point of \(\mathit{fsim}\), i.e., \(\GB\subseteq \mathit{fsim}(\GB)\).
\end{corollary}
\begin{lemma}
\label{lem:psim-coind}
 \(\GB \models P \fsh{\preceq} Q\myif \phi\) iff
\(\llbracket P \fsh{\preceq} Q\myif \phi \rrbracket\subseteq\nu\,S\,.\,\mathit{psim}(S)\), where
\begin{align*}
\mathit{psim}(S)= \{(\GP,\GQ)\mid\,& (\GP{\Downarrow}\fsh{\limplies} \exists\,\GQ'.\,\GQ\rewrite_{\R_R}^* \GQ'\fsh{\land}(\GP,\GQ')\in \GB\fsh{\land}{}\\
&~~\fsh{\lnot} \GP{\Downarrow}\fsh{\limplies}\forall\,\GP'.\,\GP\rewrite_{\R_L} \GP'\fsh{\limplies}\\
&~~\phantom{\fsh{\lnot} \GP{\Downarrow}\fsh{\limplies}\forall\,\GP'.\,}\exists\,\GQ'\,.\, \GQ\rewrite_{\R_R}^* \GQ'\fsh{\land}(\GP',\GQ')\in S\\
&)\lor{}\\
&\exists\,\GQ'.\,\GQ\rewrite_{\R_R} \GQ'\fsh{\land}(\GP,\GQ')\in S\\
\}
\end{align*}
\end{lemma}
\begin{proof}%[Sketch]
We have \(\GB \models \GP \fsh{\preceq} \GQ\myif \phi\) iff \(\llbracket \GP \fsh{\preceq} \GQ\myif \phi \rrbracket\subseteq {S^\fsh{\preceq}}\), where
\begin{align*}
{S^\fsh{\preceq}}=\{(\GP,\GQ)\mid {}&\forall\,\GP'.\,\GP\rewrite_{\R_L}^* \GP'\fsh{\land} \GP'{\Downarrow}\fsh{\limplies}\\
&\phantom{\forall\,\GP'\,.}\exists\,\GQ'.\, \GQ\rewrite_{\R_R}^* \GQ'\fsh{\land}(\GP',\GQ')\in \GB\\
&\lor \GQ{\Uparrow}
\}
\end{align*}
and \(\GQ{\Uparrow}\) means that there is an infinite execution starting from \(\GQ\), i.e., \(\exists\, \GQ_1,\GQ_2,\ldots\) such that \(\GQ\rewrite_{\R_R}{}\GQ_1\rewrite_{\R_R}{}\GQ_2\rewrite_{\R_R}{}\cdots\).
\\
\(\Rightarrow\). We show that \({S^\fsh{\preceq}}\)
is post-fixed point for \(\mathit{psim}\), i.e., \({S^\fsh{\preceq}}\subseteq \mathit{psim}({S^\fsh{\preceq}})\). Let \((\GP,\GQ)\in {S^\fsh{\preceq}}\). There are three cases:
\begin{enumerate}
\item \(\GP{\Downarrow}\) and \(\exists\,\GQ'\,.\,\GQ\rewrite_{\R_R}^* \GQ'\fsh{\land}(\GP,\GQ')\in \GB\). 
\item \(\fsh{\lnot} \GP{\Downarrow}\) and \(\forall\,\GP''.\,\GP\rewrite_{\R_L}^* \GP''\fsh{\land} \GP''{\Downarrow}\fsh{\limplies}\exists\,\GQ'.\, \GQ\rewrite_{\R_R}^* \GQ'\fsh{\land}(\GP'',\GQ')\in \GB\).
Let \(\GP\rewrite_{\R_L} \GP'\) arbitrary. Then \(\forall\,\GP''.\,\GP'\rewrite_{\R_L}^* \GP''\fsh{\land} \GP''{\Downarrow}\fsh{\limplies}\exists\,\GQ'.\, \GQ\rewrite_{\R_R}^* \GQ'\fsh{\land}(\GP'',\GQ')\in \GB\). Hence \((\GP',\GQ')\in{S^\fsh{\preceq}}\).
\item \(\GQ{\Uparrow}\), i.e., \(\exists\, \GQ_1,\GQ_2,\ldots\) such that \(\GQ\rewrite_{\R_L}{}\GQ_1\rewrite_{\R_R}{}\GQ_2\rewrite_{\R_R}{}\cdots\). Then \((\GP,\GQ'=Q_1)\in{S^\fsh{\preceq}}\).
\end{enumerate}
The above case-analysis shows that \((\GP,\GQ)\in\mathit{psim}({S^\fsh{\preceq}})\) in all the cases.
\\
\(\Leftarrow\). We show that if \(S\) is a post-fixed point for \(\mathit{psim}\) then \(S\subseteq{S^\fsh{\preceq}}\).
Let \((\GP,\GQ)\in S\subseteq \mathit{psim}(S)\). We have the following cases:
\begin{enumerate}
\item \(\GP{\Downarrow}\) and \( \exists\,\GQ'.\,\GQ\rewrite_{\R_R}^* \GQ'\fsh{\land}(\GP,\GQ')\in \GB\). We obviously have \((\GP,\GQ)\in{S^\fsh{\preceq}}\).
\item\label{it:doi} \(\exists\,\GQ_1.\,\GQ\rewrite_{\R_R} \GQ_1\fsh{\land} \GP_1=\GP\fsh{\land}\fsh{\land}(\GP_1,\GQ_1)\in S\). 
\item\label{it:trei} \(\forall\,\GP_1\,.\,\exists\,\GQ_1.\,\GP\rewrite_{\R_L}^* \GP_1\fsh{\land}\GQ\rewrite_{\R_R}^* \GQ_1\fsh{\land}(\GP_1,\GQ_1)\in S\).    
\end{enumerate}
The steps~\ref{it:doi} and~\ref{it:trei} are repeated until  
for each \(\GP_n\) with \(\GP\rewrite_{\R_L}^* \GP_n\fsh{\land} \GP_n{\Downarrow}\) either exists \(\GQ_{n+1}\) such that \(\GQ\rewrite_{\R_R}^* \GQ_n\fsh{\land}(\GP_n,\GQ_n)\in \GB\) and \((\GP_i,\GQ_i)\in S\) for \(i\in\{1,\ldots, n\}\), or we obtain an infinite sequence \(\GQ\rewrite_{\R_R}\GQ_1\rewrite_{\R_R}\GQ_2,\ldots\) with \((\GP,\GQ_i)\in S\). In both cases we obtain \((\GP,\GQ)\in{S^\fsh{\preceq}}\).
\qed 
\end{proof}
\begin{corollary}
1. \(S^\fsh{\preceq} =\nu\,S\,.\,\mathit{psim}(S)\).\\
2. \(\GB\) is a post-fixed point of \(\mathit{psim}\), i.e., \(\GB\subseteq \mathit{psim}(\GB)\).
\end{corollary}
%\begin{restatable}[Soundness for full simulation]{theorem}{soundenessfullsim}
%\label{thm:sound-fsim}
%%
%  If \(G, B \vdash^0 G\) and \(\llbracket
%  B \rrbracket \subseteq \GB\), then for any sequent \(\vdash
%  P \fsh{\prec} Q\myif \phi \in G,\) we have that \(\GB \models
%  P \fsh{\prec} Q\myif \phi.\)
%%
%%\end{theorem}
%\end{restatable}
%
\soundenessfullsim*
Before proving Theorem~\ref{thm:sound-fsim}, we introduce the following notations, where \(\GP\) and \(\GQ\) denote ground configurations , \(\GC\) and \(S\) a set of pairs of ground configurations:
\begin{description}
\item[]\(\mathit{CoReach}((\GP,\GQ),\GC)\equiv \exists\,\GQ'\,.\,\GQ\rewrite_{\R_R}^*\GQ'\fsh{\land} (\GP,\GQ')\in \GC\),
\item[]\(\mathit{Reach}^{\!+}((\GP,\GQ),S)\equiv \fsh{\lnot}\GP{\Downarrow}\fsh{\land} \forall\,\GP'\,.\,\GP\rewrite_{\R_L}\GP'\fsh{\limplies} \exists\,\GQ'\,.\,\GQ\rewrite_{\R_R}^*\GQ'\fsh{\land} (\GP',\GQ'){\in} S\),
\item[]\(f_\GC(S)=\{(\GP,\GQ)\mid \mathit{CoReach}((\GP,\GQ),\GC)\lor \mathit{Reach}^{\!+}((\GP,\GQ),S)\}\).
\end{description}
\begin{lemma}
\label{lem:fprop}
1. \(f_\GC\) is monotonic.\\
2. If \(\GC \subseteq \GD\) then \(f_\GC(S)\subseteq f_\GD(S)\) for any \(S\).\\
3. \(f_{\GC\cup\GD}(S)=f_\GC(S)\cup f_\GD(S)\).
\end{lemma}
\begin{proof}
The conclusions of the lemma are direct consequences of the definition.
\qed
\end{proof}
\begin{lemma}
\label{lem:coreach}
Let $\GC$ and $S_\GC$ be such that \(\GC=\{(P',Q')\mid\mathit{Reach}^{\!+}((P',Q'),S_\GC)\}\). Then 
\[\models \mathit{CoReach}((\GP,\GQ),\GC)\fsh{\limplies} 
\mathit{Reach}^{\!+}((\GP,\GQ),S_\GC)\]
\end{lemma}
\begin{proof}
\begin{align*}
\mathit{CoReach}((\GP,\GQ),\GC)&\iff
\exists\,\GQ''\,.\,\GQ\rewrite_{\R_R}\GQ''\fsh{\land} (\GP,\GQ'')\in\GC\\
&\iff
\exists\,\GQ''\,.\,\GQ\rewrite_{\R_R}\GQ''\fsh{\land} \fsh{\lnot}\GP{\Downarrow}\fsh{\land}{}\\
&\phantom{\iff\exists\,\GQ''\,.\,~} \forall\,\GP'\,.\,\GP\rewrite_{\R_L}\GP'\fsh{\limplies}\\
&\phantom{\iff\exists\,\GQ''\,.\, \forall\,\GP'\,.\,~} \exists\,\GQ'\,.\,\GQ''\rewrite_{\R_R}^*\GQ'\fsh{\land} (\GP',\GQ')\in S_\GC\\
&\fsh{\limplies}
\fsh{\lnot}\GP{\Downarrow}\fsh{\land} \forall\,\GP'\,.\,\GP\rewrite_{\R_L}\GP'\fsh{\limplies}\\
&\phantom{\iff \fsh{\lnot}\GP{\Downarrow}\fsh{\land} \forall\,\GP'\,.\,~} \exists\,\GQ'\,.\,\GQ\rewrite_{\R_R}^*\GQ'\fsh{\land} (\GP',\GQ')\in S_\GC\\
&\iff
\mathit{Reach}^{\!+}((P,Q),S_\GC)
\end{align*}
\qed
\end{proof}
\begin{corollary}
\label{cor:coreach}
If \(\GC=\{(P',Q')\mid\mathit{Reach}^{\!+}((P',Q'),S_\GC)\}\)
then \(f_\GC(S)\subseteq f_\emptyset(S\cup S_\GC)\).
\end{corollary}
\begin{proof}
\begin{align*}
f_\GC(S)&=\{(P,Q)\mid \mathit{CoReach}((P,Q),\GC)\lor \mathit{Reach}^{\!+}((P,Q),S)\}\\
&\subseteq\{(P,Q)\mid \mathit{Reach}^{\!+}((P,Q),S_\GC)\lor \mathit{Reach}^{\!+}((P,Q),S)\}\\
&=\{(P,Q)\mid \mathit{Reach}^{\!+}((P,Q),S\cup S_\GC)\}\\
&=f_\emptyset(S\cup S_\GC)
\end{align*}
\qed
\end{proof}
\begin{lemma}
\label{lem:feqfsim}
\(f_\GB(S) = \mathit{fsim}(S)\).
\end{lemma}
\begin{proof}
%\(\mathit{CoReach}((P,Q),\GC)\) implies \(\GP{\Downarrow}\).
We first notice that the following fact holds:
\[
\models (V\fsh{\limplies} W)\fsh{\limplies} ((U\fsh{\limplies} V) \fsh{\land} (\fsh{\lnot} U\fsh{\limplies} W)\iff V\lor (\fsh{\lnot} U\fsh{\land} W))
\]
We have 
\begin{align*}
\mathit{fsim}(S) &{}=\{(\GP,\GQ)\mid \GP{\Downarrow}\fsh{\limplies} \mathit{CoReach}((\GP,\GQ),\GB)\fsh{\land} \fsh{\lnot}\GP{\Downarrow}\fsh{\limplies} \mathit{Reach}^((\GP,\GQ),S)\}\\
&{}=\{(\GP,\GQ)\mid \mathit{CoReach}((\GP,\GQ),\GB)\lor \mathit{Reach}^((\GP,\GQ),S)\}\\
&{}=f_\GB(S)
\end{align*}
by applying the above fact, where
\(U\equiv \GP{\Downarrow}\), \(V\equiv \mathit{CoReach}((P,Q),\GB)\), and \(W\equiv  \forall\,\GP'\,.\,(\GP\rewrite_{\R_L}\GP')\fsh{\limplies} (\exists\,\GQ'\,.\,\GQ\rewrite_{\R_R}^*\GQ'\fsh{\land} (\GP',\GQ')\in S_\GB)\).
%(declare-fun u () Bool)
%(declare-fun v () Bool)
%(declare-fun w () Bool)
%(assert (implies v w))
%; (assert (implies v u))
%(assert (not (implies (and (implies u v) (implies (not u) w)) (or v (and (not u) w)))))
%(check-sat)
\qed
\end{proof}
\begin{lemma}
\label{lem:rulestabf}
Let \(\mathit{PT}\) the set of the proof trees of \(G, B \vdash^0 G\) and let \(S\) be the union of all the sets  \(\llbracket \varphi\rrbracket\) with \(\varphi\) occurring in \(\mathit{PT}\). 
Then \(\llbracket \varphi\rrbracket\subseteq f_{\GB\cup\GG}(S)\) for each \( \varphi\) occurring in \(\mathit{PT}\).
\end{lemma}
\begin{proof}[Sketch]
We proceed by induction on the height of \(\mathit{PT}\) and case analysis on the rule applied in the root.
\\
\textsc{Axiom}. We have \(\llbracket \varphi\rrbracket = \emptyset\).
\\
\textsc{Base}.
\( \phi' \fsh{\limplies} \textit{sub}\Big((P, Q'), B\Big)\) means \( \phi' \fsh{\limplies} \llbracket (P, Q')\rrbracket\subseteq \llbracket B\rrbracket \subseteq \GB\), which implies \(\phi'\fsh{\limplies} \forall\,\GP\in\llbracket P\myif\phi'\rrbracket\,.\,\GP{\Downarrow}\fsh{\land} \forall\,\GQ'\in\llbracket Q'\myif\phi'\rrbracket\,.\,\GQ'{\Downarrow}\) (i.e, \(\phi'\) implies that \(\GP\) and \(\GQ'\) are terminal). Since \(\phi'\) is a path condition derived from \(Q'\), it follows that it does not affect the termination of \(P\), i.e. we have \(\models \forall\,\GP\in\llbracket P\myif\top\rrbracket\,.\,\GP{\Downarrow}\). It follows that \(\models \forall\,(\GP,\GQ)\in\llbracket P \fsh{\prec} Q\myif\phi\rrbracket\,.\,\mathit{CoReach}((\GP,\GQ),\GB)\), which implies
\(\llbracket P \fsh{\prec} Q\myif\phi\rrbracket\subseteq f_\GB(\emptyset)\cup f_{\GB\cup\GG}(\llbracket P \fsh{\prec} Q\mbox{ \textit{if} } \phi\fsh{\land} \fsh{\lnot} \phi_B\rrbracket)\subseteq f_{\GB\cup\GG}(S)\) by the definition of \(f_{\GB}\), Lemma~\ref{lem:fprop}, and the inductive hypothesis.
\\
\textsc{Circ}.
\( \phi' \fsh{\limplies} \textit{sub}\Big((P, Q'), G\Big)\) means \( \phi' \fsh{\limplies} \llbracket (P, Q')\rrbracket\subseteq \GG\), which implies \(\models\forall\,(\GP,\GQ)\in \llbracket (P, Q)\rrbracket\,.\,\mathit{CoReach}((P,Q),\GG)\). We obtain
\[\llbracket P \fsh{\prec} Q\myif\phi\rrbracket\subseteq f_\GG(\llbracket P \fsh{\prec} Q\mbox{ \textit{if} } \phi\rrbracket)\cup f_{\GB\cup\GG}(\llbracket P \fsh{\prec} Q\mbox{ \textit{if} } \phi\fsh{\land} \fsh{\lnot} \phi_B\rrbracket)\subseteq f_{\GB\cup\GG}(S) \] by the definition of \(f_{\GG}\), Lemma~\ref{lem:fprop}, and the inductive hypothesis.
\\
\textsc{Step}. 
We have 
\(\llbracket P \fsh{\prec} Q\myif\phi\rrbracket =\llbracket  P \fsh{\prec} Q\myif\phi \fsh{\land} (\phi_1\lor\cdots\lor\phi_n)\rrbracket\cup \llbracket  P \fsh{\prec} Q\myif\phi \fsh{\land}\fsh{\lnot}\phi_1\fsh{\land}\cdots\fsh{\land}\fsh{\lnot}\phi_n\rrbracket\).
We obtain \(\llbracket  P \fsh{\prec} Q\myif\phi \fsh{\land}\fsh{\lnot}\phi_1\fsh{\land}\cdots\fsh{\land}\fsh{\lnot}\phi_n\rrbracket \subseteq f_{\GB\cup\GG}(S) \) by the inductive hypothesis.
We have \(\{\GP'\mid \GP\rewrite_{\R_L}\GP',\GP\in\llbracket P\myif\phi\rrbracket\}=\llbracket \Delta_{\R_L}(P\myif\phi)\rrbracket\)  (by Theorem~\ref{th:der2}) and \(\{(\GP',\GQ)\mid \GP\rewrite_{\R_L}\GP',\GP\in\llbracket P\myif\phi\rrbracket,\GQ\in\llbracket Q\myif\phi\rrbracket \}=\{(\GP',\GQ)\mid \GP'\in\llbracket \Delta_{\R_L}(P\myif\phi)\rrbracket,\GQ\in\llbracket Q\myif\phi\rrbracket\}={}
\llbracket \{P^i \fsh{\prec} Q\myif\phi^i\mid 1 \leq i \leq n\}\rrbracket\), which implies\\
\begin{align*}
\models\forall\,(\GP,\GQ)\,.\,&(\GP,\GQ)\in \llbracket P \fsh{\prec} Q\myif\phi \fsh{\land} (\phi_1\lor\cdots\lor\phi_n)\rrbracket\fsh{\limplies}\\
&\mathit{Reach}^{\!+}((\GP,\GQ), \llbracket \{P^i \fsh{\prec} Q\myif\phi^i\mid 1 \leq i \leq n\}\rrbracket)
\end{align*}
We obtain
\[\begin{array}{l}\llbracket P \fsh{\prec} Q\myif\phi\rrbracket\subseteq\\ f_{\GB\cup\GG}(\llbracket \{P^i \fsh{\prec} Q\myif\phi^i\mid 1 \leq i \leq n\}\rrbracket)\cup \llbracket P \fsh{\prec} Q\myif\phi\fsh{\land}\fsh{\lnot}\phi_1\fsh{\land}\ldots\fsh{\land}\fsh{\lnot}\phi_n\rrbracket \subseteq \\f_{\GB\cup\GG}(S)\end{array} \] by the definition of \(f_{\GG}\), Lemma~\ref{lem:fprop}, and the inductive hypothesis.
\qed
\end{proof}

\begin{proof}[Theorem~\ref{thm:sound-fsim}]
Let \(\mathit{PT}\) the set of the proof trees of \(G, B \vdash^0 G\) and let \(S\) be the union of all the sets  \(\llbracket P \fsh{\prec} Q\myif\phi\rrbracket\) with \( P \fsh{\preceq} Q\myif\phi\) occurring in \(\mathit{PT}\). We obtain \(S\subseteq f_{\GB\cup\GG}(S)\) by Lemma~\ref{lem:rulestabf}  and Lemma~\ref{lem:fprop}.
%We assume w.l.o.g. that for any \(P \fsh{\prec} Q\myif \phi \in G,\) we have \(\forall\,\GP\in\llbracket P\myif \phi\rrbracket\,.\,\fsh{\lnot}\GP{\Downarrow}\).
%Otherwise, \(\GB\models  P \fsh{\prec} Q\myif\phi\) iff  \(\GB\models  P \fsh{\prec} Q\myif\phi \fsh{\land}\fsh{\lnot}\phi_1\fsh{\land}\cdots\fsh{\land}\fsh{\lnot}\phi_n\) and \(\GB\models P \fsh{\prec} Q\myif\phi \fsh{\land} (\phi_1\lor\cdots\lor\phi_n)\), where \(\Delta_{\R_L}(P\mbox{
%      \(\textit{if}\) }\phi) = \{ P^i\myif\phi^i\mid i
%    \in \{ 1, \ldots, n \} \}\). 
%    The former can be proved using \textsc{Base} and/or \textsc{Axiom}, and the later satisfies our assumption.
%First only the rule \textsc{Step} can be applied for all the formulae in \(G\). It follows that there is \(S_\GG\subseteq S\) such that \(\models\forall\,(\GP,\GQ)\in\GG\,.\,\mathit{Reach}^{\!+}((\GP,\GQ), S_\GG)\).
Let \((G_0,G_1)\) denote the partition of \(G\) such  that the proof trees corresponding to \(G_i\) uses only instances of inference rules  \(G,B\vdash^g\varphi\) with \(g\le i\). Let \(\GG_i\) denote \(\llbracket G_i\rrbracket\). For \(G_0\) we have \(f_{\GG_0}(S) \subseteq f_\GB(S)\) since only \textsc{Base} \textsc{Axiom} rules are applied in the proof of \(G_0\). 

We obtain \(S\subseteq f_{\GB\cup\GG}(S) =f_{\GB}(S)\cup f_{\GG_0}(S)\cup f_{\GG_1}(S)\subseteq f_{\GB}(S)\cup f_{\emptyset}(S_{\GG_1})=f_{\GB}(S\cup S_{\GG_1})=f_{\GB}(S)=\mathit{fsim}(S)\) by Lemma~\ref{lem:fprop}, Corollary~\ref{cor:coreach}, and Lemma~\ref{lem:feqfsim}.
Hence \(S\subseteq \nu\,Y\,.\,\mathit{fsim}(Y)\), which implies the conclusion of the theorem by Lemma~\ref{lem:fsim-coind}.
\qed
\end{proof}

\soundenesspartsim*

We first introduce the following additional notations:
\begin{description}
\item[]\(\mathit{CoReach}^{\!+}((\GP,\GQ),\GC)\equiv \exists\,\GQ'\,.\,\GQ\rewrite_{\R_R}\GQ'\fsh{\land} (\GP,\GQ')\in \GC\),
\item[]\(p_\GC(S)=\{(P,Q)\mid \mathit{CoReach}((P,Q),\GC)\lor \mathit{CoReach}^{\!+}((P,Q),S)\lor\mathit{Reach}^{\!+}((P,Q),S)\}\).
\end{description}

The function \(p_\GC\) has properties similar to those of \(f_\GC\):
\begin{lemma}
\label{lem:pprop}
1. \(p_\GC\) is monotonic.\\
2. If \(\GC \subseteq \GD\) then \(p_\GC(S)\subseteq p_\GD(S)\) for any \(S\).\\
3. \(p_{\GC\cup\GD}(S)=p_\GC(S)\cup p_\GD(S)\).\\
4. If \(\GC=\{(P',Q')\mid\mathit{Reach}^{\!+}((P',Q'),S_\GC)\}\)
then \(p_\GC(S)\subseteq p_\emptyset(S\cup S_\GC)\).
\end{lemma}
\begin{proof}
It follows directly from the definition.
\end{proof}

\begin{lemma}
\label{lem:peqpsim}
\(p_\GB(S) = \mathit{psim}(S)\).
\end{lemma}
\begin{proof}
The first member of the disjunction from the definition of \(\mathit{psim}\) is equivalent to \( \mathit{CoReach}((P,Q),\GC)\lor\mathit{Reach}^{\!+}((P,Q),S)\) by Lemma~\ref{lem:feqfsim} and the second one is equivalent to \(\mathit{CoReach}^{\!+}((P,Q),S)\).
\end{proof}

\begin{lemma}
\label{lem:rulestabp}
Let \(\mathit{PT}\) the set of the proof trees of \(G, B \vdash^0 G\) and let \(S\) be the union of all the sets  \(\llbracket \varphi\rrbracket\) with \( \varphi\) occurring in \(\mathit{PT}\). Then \(\llbracket \varphi\rrbracket\subseteq p_{\GB\cup\GG}(S)\) for each \( \varphi\) occurring in \(\mathit{PT}\).
\end{lemma}
\begin{proof}[Sketch]
We proceed by induction on the height of \(\mathit{PT}\) and case analysis on the rule applied in the root.
For \textsc{Axiom} and \textsc{Base} the proofs are similar to those of Lemma~\ref{lem:rulestabf}.
For the rest of rules we let \((G_0,G_1)\) denote the partition of \(G\) such  that the proof trees corresponding to \(G_i\) uses only instances of inference rules  \(G,B\vdash^g\varphi\) with \(g\le i\). We also use  \(\GG_i\) to denote \(\llbracket G_i\rrbracket\).
\\
\textsc{Circ}.
Let \(\varphi\) denote \(P \fsh{\preceq} Q\myif\phi\).
\\
Subcase \(g=1\).
We have \(\{(\GP,\GQ)\in \llbracket\varphi\rrbracket\mid \mathit{CoReach}^+((\GP,\GQ),\GG_1)\}=\{(\GP,\GQ)\in \llbracket\varphi\rrbracket\mid \mathit{Reach}^{\!+}((\GP,\GQ),S_{\GG_1})\}\subseteq  p_\emptyset(S_{\GG_1})\subseteq p_\GB(S)\) by Lemma~\ref{lem:coreach} and the definition of \(p\);  the inclusion \(S_{\GG_1}\subseteq S\) follows by the fact there is an instance of \textsc{Step} in \(\mathit{PT}\) for each formula in \(G_1\).
Since there is no an instance of \textsc{Step} in \(\mathit{PT}\) for any formula in \(G_0\), \((\GP,\GQ)\in\GG_0\) implies \(\mathit{CoReach}((\GP,\GQ),\GB )\lor \mathit{CoReach}^+((\GP,\GQ),S)\) (corresponding to \textsc{Base} and \textsc{Circ}, respectively).
We obtain 
\begin{align*}
&\{(\GP,\GQ)\in \llbracket\varphi\rrbracket\mid \mathit{CoReach}((\GP,\GQ),\GG_0)\}\subseteq{}\\
& \{(\GP,\GQ)\in \llbracket\varphi\rrbracket\mid \mathit{CoReach}((\GP,\GQ),\GB )\lor \mathit{CoReach}^+((\GP,\GQ),S)\}\subseteq {}\\
&p_\GB(S)
\end{align*} 
\(\GG=\GG_0\cup\GG_1\) implies \(\{(\GP,\GQ)\in \llbracket\varphi\rrbracket\mid \mathit{CoReach}((\GP,\GQ),\GG)\}\subseteq p_\GB(S)\).
\\
Subcase \(g=0\). We have 
\begin{align*}
&\{(\GP,\GQ)\in \llbracket\varphi\rrbracket\mid \mathit{CoReach}^+((\GP,\GQ),\GG)\}\subseteq{}\\
&\{(\GP,\GQ)\in \llbracket\varphi\rrbracket\mid \mathit{CoReach}^+((\GP,\GQ),S)\}\subseteq{}\\
&p_\GB(S)
\end{align*} 
Now the two cases are finished.
\\
From the definition of the rule we obtain
\begin{align*}
&\llbracket \varphi\rrbracket\subseteq {}\\
&\{(\GP,\GQ)\mid \mathit{CoReach}^+((\GP,\GQ),\GG)\}\cup\llbracket P \fsh{\preceq} Q\myif\phi\fsh{\land}\fsh{\lnot}\phi_G\rrbracket\subseteq{}\\
&  p_\GB(S)\cup p_{\GB\cup\GG}(S) \subseteq {}\\
& p_{\GB\cup\GG}(S)
\end{align*} 
by the definition, the properties of \(p\) and the inductive hypothesis.
\\ 
\textsc{Step}.
Let \(\varphi\) denote \(P \fsh{\preceq} Q\myif\phi\).
We have 
\begin{align*}
&\llbracket \varphi\rrbracket\subseteq{}\\
& \{(\GP,\GQ)\mid \mathit{Reach}^+((\GP,\GQ),\Delta_{\R_l}(P\myif\phi))\}\cup \llbracket P \fsh{\preceq} Q\myif\phi\fsh{\land}\fsh{\lnot}\phi_1\fsh{\land}\cdots\fsh{\land}\fsh{\lnot}\phi_n\rrbracket \subseteq{}\\
&  p_{\GB\cup\GG}(S)
\end{align*} 
by the inductive hypothesis and the definition of \(P\).
\end{proof}

\begin{proof}[Theorem~\ref{thm:sound-psim}]
Let \(\mathit{PT}\) the set of the proof trees of \(G, B \vdash^0 G\) and let \(S\) be the union of all the sets  \(\llbracket\varphi\rrbracket\) with \(\varphi\) occurring in \(\mathit{PT}\). We obtain \(S\subseteq p_{\GB\cup\GG}(S)\) by Lemma~\ref{lem:rulestabp}.
We have \(S\subseteq p_{\GB\cup\GG}(S) =p_{\GB}(S)\cup p_{\GG_0}(S)\cup p_{\GG_1}(S)\subseteq p_{\GB}(S)\cup p_{\emptyset}(S_{\GG_1})=p_{\GB}(S\cup S_{\GG_1})=p_{\GB}(S)=\mathit{psim}(S)\) by Lemma~\ref{lem:pprop}, Lemma~\ref{lem:rulestabp}, and Lemma~\ref{lem:peqpsim}.
Hence \(S\subseteq \nu\,Y\,.\,\mathit{psim}(Y)\), which implies the conclusion of the theorem by Lemma~\ref{lem:psim-coind}.

\end{proof}

%%% Local Variables:
%%% mode: latex
%%% TeX-master: "submission.tex"
%%% End:

\clearpage
\section{Examples of Optimization Correctness Proofs}
\label{app:optimization}

% TODO: should this be in the paper?
% Each optimization is formalized in
% 3 files:
% \begin{enumerate*}
% \setlength{\itemsep}{0pt}
% \item a \texttt{.orig} file, containing the original program;
% \item a \texttt{.opt} file, containing the optimized program;
% \item a \texttt{.wp} file, which contains the read- and write-sets of
%   symbolic expressions and statements.
% \end{enumerate*}
% Table~\ref{table:cork-example} contains an example optimization
% from~\cite{LopesMonteiro}.

% \begin{table}\footnotesize
% \centering
% \begin{tabular}{|l | l | l|}
% \hline
% \texttt{constant-propagation01.orig} & \texttt{constant-propagation01.opt} & \texttt{constant-propagation01.wp} \\
% \hline
% \texttt{\laassign{V1}{E1}} & \texttt{\laassign{V1}{E1}} & \texttt{R[S1]=C1,C2,V1,V2} \\
% \texttt{S1} & \texttt{S1} & \texttt{W[S1]=C1} \\
% \texttt{\laassign{V2}{E1}} & \texttt{\laassign{V2}{V1}} & \texttt{R[E1]=C2,V2} \\
% \hline
% \end{tabular}
% \caption{An example optimization, as specified in the \texttt{CORK} tool \cite{LopesMonteiro}}
% \label{table:cork-example}
% \vspace*{-2ex}
% \end{table}

We present these examples in a \texttt{C}-like language, denoting
symbolic expressions and sequences by suggestive identifiers (e.g.,
\texttt{E1}, \texttt{S1}). We go through all examples worked out in
CORK~\cite{LopesMonteiro}. The read-sets and write-sets of these
expressions/sequences are given in a \texttt{.wp} (\textit{weakest
  precondition}) file. We encode these read/write-sets as explained in
Section~\ref{sec:programschemas}. In this section, we will use the
term equivalence and we will mean two-way simulation. When proving
equivalence, we prove the simulation of the rhs by the lhs and
vice-versa. The base cases and the goals used for the reverse
direction are always symmetric. These examples are all implemented in the \RMT{} tool (\url{http://profs.info.uaic.ro/~stefan.ciobaca/rmteq}).

The base equivalence \(\GB\) we consider has two terminal programs, under
the constraint that output variables have equal values in the two
resulting environments.

In this section we list each of the optimizations individually,
analyze them, and present how our method can be applied to prove the
equivalences. For consistency, we will always present the original
program on the left and the optimized one on the right.

In order for \RMT{} to successfully prove equivalences involving
loops, it requires helper equivalences (circularities). We describe a
methodology, which we call \emph{snapshotting} the two programs at
certain points, that allows us to easily find these helper
circularities for proving optimizations.

By taking a \emph{snapshot} of a program at a certain point, we mean
running it until it reaches that point and saving its form once that
point is reached. \RMT{} provides a \texttt{run} query, which can be
used to make the process of taking these snapshots easier.

In general, in order to obtain a helpful circularity, the snapshot
needs to happen at a point in which the structure of the program
remains similar after some program steps are executed (e.g., inside
loops). In addition, the two snapshots of the programs still need to
be equivalent. Usually, this only happens under some
\emph{constraint}.

Consider a simple example in which we want to prove the equivalence of
a program with itself:

\begin{minipage}[t]{4cm}
	\begin{lstlisting}
	V1 = 0;
	while (V1 < V2) {
	  S1;
	  V1 = V1 + 1;
	}
	\end{lstlisting}
\end{minipage}
\hfill
\begin{minipage}[t]{4cm}
	\begin{lstlisting}
	V1 = 0;
	while (V1 < V2) {
	  S1;
	  V1 = V1 + 1;
	}
	\end{lstlisting}
\end{minipage}

For this example, no matter how many program steps we execute, the
program will never be structurally similar to the initial one, due to
the initial assignment is \emph{erased} after being executed. However,
we snapshot the program at the beginning of the loop (just after the
initial assignment was executed). In order to prove the initial
equivalence, we use the following helper equivalence found by
snapshotting:

\begin{minipage}[t]{4cm}
	\begin{lstlisting}
	while (V1 < V2) {
	  S1;
	  V1 = V1 + 1;
	}
	\end{lstlisting}
\end{minipage}
\hfill
\begin{minipage}[t]{4cm}
	\begin{lstlisting}
	while (V1 < V2) {
	  S1;
	  V1 = V1 + 1;
	}
	\end{lstlisting}
\end{minipage}

Intuitively, because the snapshot reaches a configuration having a
form similar to itself, this new equivalence can be used to prove
itself (hence why we call these equivalences \emph{circularities}). In
addition, the original equivalence can be easily reduced to this
second one.

Since the first assignment is missing from these new programs, we must
specify that this equivalence only holds true if the value of
\texttt{V1} in the left-hand-side program is equal to the value of
\texttt{V1} in the right-hand-side program. We add this as a logical
constraint of the circularity. In the next examples, we will always
assume equality of variables with the same name, unless otherwise
specified.

Note that this is not the only way to specify a helper
circularity. For example, we could snapshot the programs at
\texttt{S1}, at \texttt{V1 = V1 + 1}, or even at different points
inside the loop, under proper constraints.

\paragraph{Code hoisting}

Program sequences which appear on both branches of a conditional
branch can be hoisted out of the if-else statement. On modern CPUs,
this could potentially help with pipelining and branch prediction,
improving run-time performance.

\begin{minipage}[t]{4cm}
	\begin{lstlisting}
	if (B1) {
	  S1;
	  S2;
	} else {
	  S1;
	  S3;
	}
	\end{lstlisting}
\end{minipage}
\hfill
\begin{minipage}[t]{4cm}
	\begin{lstlisting}
	
	S1;
	if (B1) {
	  S2;
	} else {
	  S3;
	}
	\end{lstlisting}
\end{minipage}

In this example, \texttt{B1} is a symbolic boolean expression and
\texttt{S1} is a symbolic sequence which does not write to the read
set of \texttt{B1}. There are no restrictions on the read/write-sets
of symbolic sequences \texttt{S2} and \texttt{S3}. \RMT{} is able to
prove the equivalence of the two programs with no additional
circularities (as expected, since helper circularities are generally
only needed when programs contain loops). The name of the file
corresponding to this examples is \texttt{imp-hoisting.rmt}.

\paragraph{Constant propagation}

The goal of constant propagation is to eliminate the need to evaluate
certain expressions multiple times, if these expressions remain
constant thorough the program's execution. We have already discussed
this example in Section~\ref{sec:programschemas}.

\begin{minipage}[t]{4cm}
	\begin{lstlisting}
	V1 = E1;
	S1;
	V2 = E1;
	\end{lstlisting}
\end{minipage}
\hfill
\begin{minipage}[t]{4cm}
	\begin{lstlisting}
	V1 = E1;
	S1;
	V2 = V1;
	\end{lstlisting}
\end{minipage}

In the original program, we can see that the result of \texttt{E1} is
stored in variable \texttt{V1}. On line 3, expression \texttt{E1} is
evaluated again. Provided that the evaluation of \texttt{E1} does not
need the value of \texttt{V1} or any variables which \texttt{S1}
modifies, we could avoid this re-evaluation and use the memorized
value, stored in \texttt{V1}. Depending on the complexity of
\texttt{E1}, this could greatly improve run-time performance.

\begin{minipage}[t]{4cm}
	\begin{lstlisting}
	V1 = E1;
	S1;
	V2 = E1;
	\end{lstlisting}
\end{minipage}
\hfill
\begin{minipage}[t]{4cm}
	\begin{lstlisting}
	V1 = E1;
	V2 = V1;
	S1;
	\end{lstlisting}
\end{minipage}

If additionally \texttt{S1} does not use the value of \texttt{V2}, the
order in which the last two lines are executed becomes irrelevant, as
illustrated in the example above. Modern CPUs could pick up on this
and execute the lines in parallel, further improving performance.

\RMT{} is able to prove both of these equivalences, with no helper
circularities. The names of the two files
corresponding to this examples are of the form \texttt{imp-constant-propagation*.rmt}.

\paragraph{Copy propagation}

In compiler theory, copy propagation is the process of replacing the
occurrences of targets of direct assignments with their values. Copy
propagation is a useful \emph{clean up} optimization frequently used
after other optimizations have already been run. Some optimizations,
such as elimination of common sub expressions, require that copy
propagation be run afterwards in order to achieve an increase in
efficiency.

\begin{minipage}[t]{4cm}
	\begin{lstlisting}
	V1 = V2;
	V3 = V1;
	\end{lstlisting}
\end{minipage}
\hfill
\begin{minipage}[t]{4cm}
	\begin{lstlisting}
	V1 = V2;
	V3 = V2;
	\end{lstlisting}
\end{minipage}

The two programs above illustrate an example of a copy propagation
optimization. Proving the equivalence of the two is similar to proving
equivalence in the case of constant propagation. As expected, \RMT{}
was able to prove this equivalence as well, without the need for
additional circularities. The name of the file
corresponding to this examples is \texttt{imp-copy-propagation.rmt}.

\paragraph{If-conversion}

If-conversion is an optimization which deletes a branch around an
instruction and replaces it with a predicate on the instruction. This
optimization can be described as a transformation which converts
control dependencies into data dependencies, and it may be required
for software pipelining.

\begin{minipage}[t]{4cm}
	\begin{lstlisting}
	
	if (B1) {
	  V1 = E1;
	}
	
	\end{lstlisting}
\end{minipage}
\hfill
\begin{minipage}[t]{4cm}
	\begin{lstlisting}
	if (B1) {
	  V1 = E1;
	} else {
	  V1 = V1;
	}
	\end{lstlisting}
\end{minipage}

The programs above illustrate an example of if-conversion. Using the
ternary operator from the C language, the optimized program could also
be expressed in a single line as \texttt{V1 = B1 ? E1 : V1;}. It is
assumed that neither \texttt{B1} nor \texttt{E1} have any side
effects.

\RMT{} is able to prove this equivalence as well, with no helper
circularities. The name of the file
corresponding to this examples is \texttt{imp-if-conversion.rmt}.

\paragraph{Partial redundancy elimination}

Partial redundancy elimination (PRE) is a compiler optimization that
eliminates expressions that are redundant on some but not necessarily
all paths through a program. PRE is a form of common subexpression
elimination.

An expression is called partially redundant if the value computed by
the expression is already available on some but not all paths through
a program to that expression. An expression is fully redundant if the
value computed by the expression is available on all paths through the
program to that expression. PRE can eliminate partially redundant
expressions by inserting the partially redundant expression on the
paths that do not already compute it, thereby making the partially
redundant expression fully redundant.

\begin{minipage}[t]{4cm}
	\begin{lstlisting}
	if (B1) {
	  S1;
	  V1 = E1;
	  S2;
	  
	} else {
	  S3;
	}
	V2 = E1;
	\end{lstlisting}
\end{minipage}
\hfill
\begin{minipage}[t]{4cm}
	\begin{lstlisting}
	if (B1) {
	  S1;
	  V1 = E1;
	  S2;
	  V2 = V1;
	} else {
	  S3;
	  V2 = E1;
	}
	\end{lstlisting}
\end{minipage}

Here, it is assumed that \texttt{B1} and \texttt{E1} have no side
effects, \texttt{E1} does not read from \texttt{V1} and \texttt{S1}
does not write to \texttt{V1} or \texttt{V2} and, in addition,
\texttt{S1} does not write to any variables read by \texttt{E1}.

Under these assumptions, it can be observed that \texttt{E1} is
partially redundant: its value is already available at the end of the
\texttt{if} branch, but not at the end of the \texttt{else}
branch. PRE inserts this expression on the \texttt{else} branch,
whereas on the \texttt{if} branch it uses the value stored in
\texttt{V1}, avoiding the re-evaluation of \texttt{E1}.

As before, \RMT{} is able to prove the equivalence of the two programs
without the need for helper circularities. The following examples all
have loops and therefore helper circularities are required. The name of the file
corresponding to this examples is \texttt{imp-pre.rmt}.

\paragraph{Loop invariant code motion}

Loop invariant code consists of statements and/or expressions inside a
loop body that do not depend on the contents of the loop itself, and
as such could be moved outside the loop without affecting the results
of the program. Loop invariant code motion (LICM) is the compiler
optimization which identifies such statements and moves them outside
the loop automatically. This results in a single evaluation of the
loop invariant code, as opposed to multiple ones, which could
significantly improve performance.

\begin{minipage}[t]{4cm}
	\begin{lstlisting}
	
	while (V1 < V2) {
	  S1;
	  S2;
	  V1 = V1 + 1;
	}
	
	\end{lstlisting}
\end{minipage}
\hfill
\begin{minipage}[t]{4cm}
	\begin{lstlisting}
	if (V1 < V2) {
	  S2;
	  while (V1 < V2) {
	    S1;
	    V1 = V1 + 1;
	  }
	}
	\end{lstlisting}
\end{minipage} 

In this example, \texttt{S2} is a symbolic statement, which does not
read from and does not write to any variables modified inside the loop
(i.e. \texttt{V1} and the write-set of \texttt{S1}). The \texttt{if}
instruction added in the optimized program ensures that \texttt{S2} is
only evaluated if the initial loop was going to be entered into at
least once, thus preserving program semantics.

As explained, since the right-hand-side program does not preserve its
structure (the conditional statement \emph{disappears} after a few
steps), we need a helper circularity. We created such a circularity by
snapshotting the optimized program at the start of the loop. With this
new circularity, \RMT{} is indeed able to successfully prove the
desired equivalence. The name of the file
corresponding to this examples is \texttt{imp-licm.rmt}.

\paragraph{Loop peeling}

Loop splitting is a compiler optimization technique that attempts to
simplify a loop or eliminate dependencies by breaking it into multiple
loops which have the same bodies but iterate over different contiguous
portions of the index range. Loop peeling is a special case of loop
splitting which splits any potentially problematic first (or last) few
iterations from the loop and performs them outside of the loop body.

\begin{minipage}[t]{4cm}
	\begin{lstlisting}
	
	
	
	while (V1 < V2) {
	  S1;
	  V1 = V1 + 1;
	}
	
	\end{lstlisting}
\end{minipage}
\hfill
\begin{minipage}[t]{4cm}
	\begin{lstlisting}
	if (V1 < V2) {
	  S1;
	  V1 = V1 + 1;
	  while (V1 < V2) {
	    S1;
	    V1 = V1 + 1;
	  }
	}
	\end{lstlisting}
\end{minipage} 

In this example, one loop step from the initial program is
\emph{peeled} outside the loop in the optimized one.

Interestingly, even though the second program begins with an
\texttt{if} instruction, which will be eliminated after some steps,
\RMT{} does not require an additional circularity in order to prove
this example. This is because of the definition of the language
semantics we used. In the semantics, \texttt{while(B) S;} is rewritten
to \texttt{if(B) \{S; while(B) S;\}}. It can be observed that, by
applying this transformation, the inner loop of the optimized program
translates into a program which structurally matches the initial
one. If the semantics were defined differently, we might have had to
build an auxiliary circularity in order for \RMT{} to successfully
prove the equivalence. The name of the file
corresponding to this examples is \texttt{imp-loop-peeling.rmt}.

\paragraph{Loop unrolling}

Loop unrolling is an optimization that attempts to improve the
execution speed of a program at the expense of code size. It involves
repeating the loop body multiple times inside a single iteration,
eliminating some of the loop overhead, such as unnecessary termination
condition checks.

\begin{minipage}[t]{4cm}
	\begin{lstlisting}
	while (V1 < V2) {
	  S1;
	  V1 = V1 + 1;
	}
	
	
	
	
	
	
	\end{lstlisting}
\end{minipage}
\hfill
\begin{minipage}[t]{4cm}
	\begin{lstlisting}
	while (V1+1 < V2) {
	  S1;
	  V1 = V1 + 1;
	  S1;
	  V1 = V1 + 1;
	}
	if (V1 < V2) {
	  S1;
	  V1 = V1 + 1;
	}
	\end{lstlisting}
\end{minipage} 

In this example, the loop body of the optimized program corresponds to
two iterations of the initial loop. This means that the termination
condition will be checked in the optimized program roughly half the
number of times compared to the original one. This assumes that
\texttt{S1} does not write to \texttt{V1} and \texttt{V2}. The final
\texttt{if} statement from the optimized program is needed for when
the original loop would execute \texttt{S1} an odd number of times.

As before, since the optimized program does not preserve structure
(the final conditional statement is pushed on the computation stack
before the loop is executed), we need an additional circularity. We
created a new circularity by snapshotting the second program before
the loop. \RMT{} is able to use this circularity and prove the
equivalence of the two programs.

\begin{minipage}[t]{4cm}
	\begin{lstlisting}
	V1 = 0;
	while (V1 < V2) {
	  S1;
	  V1 = V1 + 1;
	}
	
	
	\end{lstlisting}
\end{minipage}
\hfill
\begin{minipage}[t]{4cm}
	\begin{lstlisting}
	V1 = 0;
	while (V1 < V2) {
	  S1;
	  V1 = V1 + 1;
	  S1;
	  V1 = V1 + 1;
	}
	\end{lstlisting}
\end{minipage}

In this second example of unrolling, we illustrate that if we know a
priori that the loop executes an even number of times (i.e. that the
value of \texttt{V2} is an even number), we can omit the final
\texttt{if} statement, simplifying the optimized program.

As with previous examples, we use a helper circularity, which consists
of the two programs snapshotted at the start of the loops. With this
circularity in place, \RMT{} is able to successfully prove this
equivalence as well. The names of the two files
corresponding to this examples are of the form \texttt{imp-loop-unrolling*.rmt}.

\paragraph{Loop unswitching}

Loop unswitching is a compiler optimization that moves a conditional
inside a loop outside of it, by duplicating the body of loop and
placing a version of the body in each of the two branches of the
conditional statement. Despite roughly doubling the code size, this
optimization not only allows the conditional expression to be
evaluated only once (as opposed to on each iteration of the loop), but
also allows each conditional branch to be further optimized
separately.

\begin{minipage}[t]{4cm}
	\begin{lstlisting}
	
	while (V1 < V2) {
	  if (B1) {
	    S1;
	  } else {
	    S2;
	  }
	  V1 = V1 + 1;
	}
	
	
	\end{lstlisting}
\end{minipage}
\hfill
\begin{minipage}[t]{4cm}
	\begin{lstlisting}
	if (B1) {
	  while (V1 < V2) {
	    S1;
	    V1 = V1 + 1;
	  }
	} else {
	  while (V1 < V2) {
	    S2;
	    V1 = V1 + 1;
	  }
	}
	\end{lstlisting}
\end{minipage}

In this example, we assume that \texttt{B1} does not depend on
\texttt{V1}, \texttt{V2}, or on any variable in the write-sets of
\texttt{S1} and \texttt{S2}. In other words, it does not change its
value thorough the execution of the loop.

As with previous examples, since the optimized program starts with an
\texttt{if} instruction, which disappears after some program steps, we
need additional circularities. Interestingly, since the \texttt{if}
statement has two branches, we need two new circularities (one for
each branch). In other words, in the two new circularities, the left
program will remain unchanged, whereas the right program will
respectively turn into the two programs below:

\begin{minipage}[t]{4cm}
	\begin{lstlisting}
	while (V1 < V2) {
	  S1;
	  V1 = V1 + 1;
	}
	\end{lstlisting}
\end{minipage}
\hfill
\begin{minipage}[t]{4cm}
	\begin{lstlisting}
	while (V1 < V2) {
	  S2;
	  V1 = V1 + 1;
	}	
	\end{lstlisting}
\end{minipage}

Of course, we can only prove the first additional equivalence under
the constraint that \texttt{B1} evaluates to true and the second one
under the constraint that \texttt{B1} evaluates to false. \RMT{} is
able to successfully prove these equivalences. The name of the file
corresponding to this examples is \texttt{imp-loop-unswitch.rmt}.

\paragraph{Software pipelining}

Software pipelining is a technique used to optimize loops, in a manner
that enables better parallelization via hardware pipelining. This
optimization is a type of out-of-order execution, which is done by the
compiler (or by the programmer).

\begin{minipage}[t]{4cm}
	\begin{lstlisting}
	
	
	while (V1 < V2) {
	  S1;
	  S2;
	  V1 = V1 + 1;
	}
	
	
	
	\end{lstlisting}
\end{minipage}
\hfill
\begin{minipage}[t]{4cm}
	\begin{lstlisting}
	if (V1 < V2) {
	  S1;
	  while (V1 < V2-1) {
	    S2;
	    V1 = V1 + 1;
	    S1;
	  }
	  S2;
	  V1 = V1 + 1;
	}
	\end{lstlisting}
\end{minipage}

In this example, if we compare the loop bodies of the two programs, we
can see that statements \texttt{S1} and \texttt{S2} are executed in a
different order. This could lead to performance improvements if the
processor considers it easier to parallelize the second loop compared
to the first one.

As with previous examples, since the second program does not preserve
structure, we use a helper circularity, in which the second program is
snapshotted just before the \texttt{S1} instruction. With this helper
circularity, \RMT{} is able to successfully prove the equivalence of
the two programs. The name of the file
corresponding to this examples is \texttt{imp-software-pipelining.rmt}.

\paragraph{Loop fission and fusion}

Loop fission (or loop distribution) is a compiler optimization in
which a loop is broken into multiple loops over the same index range
with each taking only a part of the original body of loop. The goal is
to break down a large loop body into smaller ones for better locality.

Conversely, loop fusion (or loop jamming) is the loop transformation
that replaces multiple loops with a single one.

\begin{minipage}[t]{4cm}
	\begin{lstlisting}
	
	V1 = E1;
	
	while (V1 < V2) {
	  S1;
	  S2;
	  V1 = V1 + 1;
	}
	
	
	\end{lstlisting}
\end{minipage}
\hfill
\begin{minipage}[t]{4cm}
	\begin{lstlisting}
	V1 = E1;
	while (V1 < V2) {
	  S1;
	  V1 := V1 + 1;
	}	
	V1 = E1;
	while (V1 < V2) {
	  S2;
	  V1 = V1 + 1;
	}	
	\end{lstlisting}
\end{minipage}

The programs above represent an example of loop fission. If the order
of the programs were reversed, it would be an example of loop
fusion. It is assumed that \texttt{S1} and \texttt{S2} write to
disjoint sets of variables (let us denote these sets by \texttt{C1}
and \texttt{C2} respectively), and none of the two sequences read from
variables written to by the other. In addition, \texttt{E1} does not
read from \texttt{C1}, \texttt{C2}, or \texttt{V1}.

In order to solve this example, we need two separate simulation
proofs, with two different base equivalences. Usually, we consider the
base equivalence to be two terminal programs in which \textit{all}
relevant variables have equal values. For this example, we first
consider them equivalent (1) if \textit{only} variables in \texttt{C1}
have equal values, and then (2) if \textit{only} variables in
\texttt{C2} have equal values. In other words, we track the results of
\texttt{S1} and \texttt{S2} separately.

For (1), we construct a helper circularity by snapshotting the
left-hand-side program just before the \texttt{while} loop and the
right-hand-side program just before the \textit{first} \texttt{while}
loop. We construct another helper circularity by snapshotting the
first program at its termination point and the second program just
before the \textit{second} \texttt{while} loop. Intuitively, this last
circularity has the role of ensuring that the final loop of the second
program does not modify variables written to by \texttt{S1}. Because
this circularity contains a terminal program configuration, we can
only prove the \textit{partial} equivalence of the two programs.

For (2), we construct a helper circularity by snapshotting the
left-hand-side program just before the \texttt{while} loop and the
right-hand-side program just before the \textit{second} \texttt{while}
loop. We construct another helper circularity consisting of the first
program (unchanged) and the second program snapshotted just before the
\textit{first} \texttt{while} loop. When proving this final
circularity, the left-hand-side program does not advance; only the
right-hand side program advances, ensuring that the first loop does
not modify the variables written to by \texttt{S2}. Because the first
program must not make progress, we can only prove \textit{partial}
simulation for this case as well.

The names of the four files
corresponding to this examples are of the form \texttt{imp-loop-fission*.rmt}
and \texttt{imp-loop-fusion*.rmt}.

\paragraph{Loop interchange}

Loop interchange is the process of exchanging the order of two
iteration variables used by a nested loop. The variable used in the
inner loop switches to the outer loop, and vice versa. It is often
done to ensure that the elements of a multi-dimensional array are
accessed in the order in which they are present in memory, improving
locality of reference.

\begin{minipage}[t]{4cm}
	\begin{lstlisting}
	V1 = 0;
	V3 = 0;
	if (V3 < V4) {
	  while (V1 < V2) {
	    V3 = 0;
	    while (V3 < V4) {
	      S1;
	      V3 = V3 + 1;
	    }
	    V1 = V1 + 1;
	  }
	}	
	\end{lstlisting}
\end{minipage}
\hfill
\begin{minipage}[t]{4cm}
	\begin{lstlisting}
	V1 = 0;
	V3 = 0;
	if (V1 < V2) {
	  while (V3 < V4) {
	    V1 = 0;
	    while (V1 < V2) {
	      S1;
	      V1 = V1 + 1;
	    }
	    V3 = V3 + 1;
	  }
	}
	\end{lstlisting}
\end{minipage}

In this example, we assume that \texttt{S1} does not modify the values
of variables \texttt{V1} through \texttt{V4} and does not read the
values of \texttt{V1} and \texttt{V3}.

Similarly to previous examples, we take a snapshot of each of these
programs just before the execution of \texttt{S1}.

In order for the programs in the new circularity to truly be
equivalent, we need to add the constraint that \texttt{S1} was
executed the same number of times on both sides. We note that, in the
left-hand-side program, when the flow reaches the inner loop,
\texttt{S1} was executed \texttt{V1} \(\times\) \texttt{V4} \(+\)
\texttt{V3} times. Similarly, in the right-hand-side program,
\texttt{S1} was executed \texttt{V3} \(\times\) \texttt{V2} \(+\)
\texttt{V1} times. Therefore, the equality of these two quantities is
the required constraint.

In addition, since our circularity snapshots the program at
\texttt{S1}, we need to add to the constraint the conditions necessary
for the programs to actually reach \texttt{S1} (i.e., that all loop
conditions evaluate to true).

With this helper circularity under the discussed constraints, we can
prove the equivalence of the two programs. However, the constraint
\texttt{V1} \(\times\) \texttt{V4} \(+\) \texttt{V3} = \texttt{V3}
\(\times\) \texttt{V2} \(+\) \texttt{V1} introduces a component of
non-linear integer algebra into the proof. As discussed in
Section~\ref{sec:programschemas}, the SMT solver that we use (Z3) does
not handle non-linear integer algebra well. Because of this, \RMT{}
can only successfully prove the equivalence of the two programs if the
loop limits (i.e., the values of \texttt{V2} and \texttt{V4}) are
bounded. Our prover does not use this bound explicitly, but it is
required for the SMT solver to solve the non-linear integer algebra
problems. The name of the file
corresponding to this examples is \texttt{imp-loop-interchange.rmt}.

\paragraph{Loop reversal}

Loop reversal is an optimization that reverses the order in which
values are assigned to the loop variable, essentially changing the
direction in which the loop is iterated. In some cases, this might
improve cache efficiency and enable other optimizations.

\begin{minipage}[t]{4cm}
	\begin{lstlisting}
	V1 = E1;
	
	while (V1 < V2) {
	  S1;
	  V1 = V1 + 1;
	}
	
	
	
	
	\end{lstlisting}
\end{minipage}
\hfill
\begin{minipage}[t]{4cm}
	\begin{lstlisting}
	if (E1 < V2) {
	  V1 = V2 - 1;
	  while (V1 >= E1) {
	    S1;
	    V1 = V1 - 1;
	  }
	  V1 = V2;
	} else {
	  V1 = E1;
	}
	\end{lstlisting}
\end{minipage}

Here, we assume that \texttt{S1} does not write to \texttt{V1}, and
\texttt{E1} does not depend on \texttt{V1}, or to any variables that
\texttt{S1} writes to. The \texttt{else} branch of the second program
ensures equivalence if \texttt{E1} is greater than \texttt{V2}. The
final assignment inside the \texttt{if} branch of the second program
ensures that the value of \texttt{V1} is the same at the end of the
executions of the two programs.

In a similar manner to previous examples, we create a new circularity
by snapshotting each program at the start of the loop execution.

In this new circularity, we note that the programs are equivalent only
if the value of \texttt{V1} in the left program is equal to the value
of \texttt{V2} - 1 - \texttt{V1} in the right one. With this helper
circularity, \RMT{} is able to successfully prove the equivalence of
the two programs. The name of the file
corresponding to this examples is \texttt{imp-loop-reversal.rmt}.

\paragraph{Loop skewing}

This optimization skews the execution of an inner loop relative to an outer
one, which could be useful if the inner loop has a dependence on the
outer loop which prevents it from running in parallel. This
optimization is often combined with loop interchanging in order to
improve parallelization.

\begin{minipage}[t]{4cm}
	\begin{lstlisting}
	while (V1 < V2) {
	
	  V3 = E1;
	  
	  while (V3 < V4) {
	  
	    S1;
	    V3 = V3 + 1;
	  }
	  
	  
	  V1 = V1 + 1;
	}
	\end{lstlisting}
\end{minipage}
\hfill
\begin{minipage}[t]{5cm}
	\begin{lstlisting}
	while (V1 < V2) {
	  V5 = E1 + V6;
	  V3 = V5 - V6;
	  if (V5 < V4 + V6) {
	    while (V5 < V4 + V6) {
	      V3 = V5 - V6;
	      S1;
	      V5 = V5 + 1;
	    }
	    V3 = V4;
	  }
	  V1 = V1 + 1;
	}
	\end{lstlisting}
\end{minipage}

As with previous examples, we use an additional circularity which
consists of both programs snapshotted once they have reached
\texttt{S1}.

We need the constraint under which the two programs in the new
circularity truly are equivalent (i.e., some relationship between the
variables in the two programs). By analyzing the programs, we see that
when both programs have reached \texttt{S1}, the values of \texttt{V3}
are equal and the value of \texttt{V5} (in the right-hand-side
program) is the sum of the values of \texttt{V3} and \texttt{V6}. We
also need to add the condition that the program flows truly reach
\texttt{S1}, which is that all loop conditions evaluate to true.

With this new circularity, under the described constraints, \RMT{} is
able to successfully prove the equivalence of the two programs. The name of the file
corresponding to this examples is \texttt{imp-loop-skewing.rmt}.

\paragraph{Loop strength reduction}

Strength reduction is a compiler optimization that replaces expensive
operations by equivalent but less expensive ones. For example,
converting multiplications inside a loop into repeated additions,
which can often be used, for example, to improve the performance of
array addressing.

\begin{minipage}[t]{4cm}
	\begin{lstlisting}
	
	while (V1 < V2) {
	  V3 = V1 * V4;
	  
	  S1;
	  V1 = V1 + 1;
	}
	\end{lstlisting}
\end{minipage}
\hfill
\begin{minipage}[t]{4cm}
	\begin{lstlisting}
	V5 = V1 * V4;
	while (V1 < V2) {
	  V3 = V5;
	  V5 = V5 + V4;
	  S1;
	  V1 = V1 + 1;
	}
	\end{lstlisting}
\end{minipage}

In these examples, \texttt{S1} is a symbolic sequence that does not
write to variables \texttt{V1} through \texttt{V4}. We can see that
the original program executes a multiplication at every loop
iteration, whereas the optimized program executes the multiplication
only once, and instead replaces the original multiplications by
additions.

As in previous examples, we need to use an additional
circularity. Though multiple options are possible, we chose to
snapshot the first program just before \texttt{S1}, and the second
program just before the line \texttt{V5 = V5 + V4}. At these points,
we notice that the value of \texttt{V3} is the same on both sides. In
addition, \texttt{V3 = V1 * V4} on the left side, and \texttt{V3 = V5}
on the right side. Using these constraints, \RMT{} is able to prove
this additional circularity and, therefore, the original equivalence
as well.

Interestingly, even though we are dealing with multiplications and
therefore non-linear integer algebra, unlike Loop Interchange, the
required formula is properly solved by \ZTREI{}. The name of the file
corresponding to this examples is \texttt{imp-strength-reduction.rmt}.

\paragraph{Loop tiling}

Loop tiling is a technique that partitions the iteration space of a
loop into smaller chunks or blocks, often with the purpose of locality
optimization or parallelization.

\begin{minipage}[t]{4cm}
	\begin{lstlisting}
	
	
	
	while (V1 < V2) {
	  S1;
	  V1 = V1 + 1;
	}
	
	
	\end{lstlisting}
\end{minipage}
\hfill
\begin{minipage}[t]{6cm}
	\begin{lstlisting}
	V3 = V1;
	while (V3 < V2) {
	  V1 = V3;
	  while (V1 < min(V2, V3 + V4)) {
	    S1
	    V1 = V1 + 1;
	  }
	  V3 = V3 + V4;
	}
	\end{lstlisting}
\end{minipage}

In this example, \texttt{S1} is a symbolic statement that cannot write
to the variables \texttt{V1} through \texttt{V4}. The two programs are
similar, except that in the second one the outer loop is broken down
into smaller chunks of size \texttt{V4}. The call to \texttt{min}
inside the loop condition ensures that the programs are equivalent
even if the outer loop cannot be broken down into an exact number of
full chunks, by potentially cutting the final chunk short.

As in previous examples, we use a helper circularity, which consists
of the two programs snapshotted at \texttt{S1}. The only constraints
that we need are the conditions needed for both programs to reach
\texttt{S1} (i.e., that all loop conditions evaluate to true).

\RMT{} is able to successfully prove the equivalence of the two
programs.

\begin{minipage}[t]{4cm}
	\begin{lstlisting}
	V1 = 0;
	
	
	while (V1 < V2) {
	  S1;
	  V1 = V1 + 1;
	}
	
	
	\end{lstlisting}
\end{minipage}
\hfill
\begin{minipage}[t]{6cm}
	\begin{lstlisting}
	V1 = 0;
	while (V1 < V2) {
	  V3 = 0;
	  while (V3 < V4) {
	    S1;
	    V3 = V3 + 1;
	  }
	  V1 = V1 + V4;
	}
	\end{lstlisting}
\end{minipage}

This second example of loop tiling is similar to the first one, except
we assume that \texttt{V2} is a multiple of \texttt{V4}. This allows
avoiding the overhead of treating the case in which the outer loop
cannot be broken into an exact number of full chunks.

Even though the code in this second example is simpler than the first
one, the fact that we have to consider that \texttt{V2} is a multiple
of \texttt{V4} introduces a component of non-linear integer algebra,
which \ZTREI{} cannot properly handle. Because of this, as explained
in Section~\ref{sec:programschemas}, this example can currently be
proven by \RMT{} only when the loop limit (the value of \texttt{V2})
is bounded. Again, the bound is not a limitation of \RMT{} itself. It
might be possible to use another SMT solver, such as CVC4, as an
oracle that can handle this case of non-linear integer algebra. We
leave this for future work. The names of the two files
corresponding to this examples are of the form \texttt{imp-loop-tiling*.rmt}.

\clearpage
\section{Examples}
\label{app:examples}

In this section, we describe in greater detail all functional equivalence examples that we
prove in \IMP{}.

\paragraph{Example 1: recursive functions with and without an
  accumulator} This example corresponds to the motivating example in
Section~\ref{sec:intro}, of showing the equivalence of function with
and without accumulators. We use the language \IMP{1} (unbounded
stack). We prove the following goal:
\[
\mbox{\(\begin{array}{l}
    \cfgimp{\lacall{\pvar{f}(\vN)} \lisymb []}{\venv}{\funcsct} \rangle
  \end{array}\)}
\fsh{\prec} \mbox{\(\begin{array}{l}
    \cfgimp{\lacall{\pvar{F}(\vN, 0, 0)} \lisymb []}{\venv}{\funcsct} \end{array}\)} \mbox{ \(\myif \)
}\mbox{\( \fsh{0 \leq \vN}, \)}
\]
\noindent where \(\funcsct = \fsh{\{} \fsh{f \mapsto \lambda
\pvar{x}.\laite{0 \leq \cx}{\plus{\cx}{\lacall{\pvar{f}(\pvar{x}
    - 1)}}}{0})},\ \fsh{\pvar{F} \mapsto \lambda \pvar{n} . \lambda \pvar{i}
. \lambda \pvar{a}.}\allowbreak\fsh{\laite{\pvar{i} \leq \pvar{n}}{\lacall{\pvar{F}(n,
    \plus{i}{1}, \plus{a}{i})}}{\pvar{a}}} \fsh{\}}\).
    
Note that \(\pvar{F}, \pvar{f}, \pvar{n}, \pvar{i}, \pvar{a}\) are
identifiers (program variables), while \(\vN\) and \(\venv\) are
variables of type \sint{} and \sarray{} (from identifiers to
integers). The fact that the programs run with the same input is
implemented by the fact that the same variable \(\vN\) appears in both
the lhs (\(\ldots \lacall{\pvar{f}(\vN)} \ldots\)) and the rhs
(\(\ldots \lacall{\pvar{F}(\vN, 0, 0)} \ldots\)). The two programs
configurations have the same environment \(\menv\), although, as there
are no global variables in the examples, the environment does not
matter.

For the set \(B\) of base cases, we use \(B = \{ \cfgimp{[\vs]}{\venv}{\funcsct}
\fsh{\prec} \cfgimp{[\vs]}{\venv}{\funcsct}\},\) where \(\vs\) is a variable of sort
\sint{} (recall that \(\fsh{[\vs]}\) is a notation for the one-element
cons-list \(\vs \lisymb \fsh{[]}\)). That is, two terminal configurations are
considered equivalent if the programs are reduced to the same integer
\(\vs\) and the environments are the same.

As explained in the introduction, in order to express the
circularities, we create a defined function \(\areduce\),
axiomatized by the following constrained rules:
\begin{enumerate}
\item \(\fsh{\areduce(\vI, \vN) \mathrel{\rewrite} []\myif \vI > \vN}\);
\item
  \(\fsh{\areduce(\vI, \vN) \mathrel{\rewrite} \mbrack{\plus{\vI}{\square}}
    \lisymb \areduce(\vI + 1, \vN)\myif \vI \leq \vN}\).
\end{enumerate}

We use two helper circularities to prove the goal:
\[
  \mbox{\(\begin{array}{l}\fsh{\langle \lacall{\pvar{f}(\vI - 1)}} \\
      \qquad \lisymb \fsh{\areduce(\vI, \vN), \venv, \funcsct
           \rangle}\end{array}\)} \fsh{\prec}
  \mbox{\(\begin{array}{l}\fsh{\langle \lacall{\pvar{F}(\vN, 0, 0)}} \\ \qquad \fsh{\lisymb [], \venv, \funcsct \rangle} \end{array}\)}\qquad \myif \fsh{0 \leq \vI \leq \vN},
\]

\[\mbox{\(\begin{array}{l}\fsh{\langle S \lisymb \areduce (\vI, \vN),} \\
           \qquad \fsh{\venv, \funcsct
  \rangle} \end{array}\)}
\fsh{\prec}
\mbox{\(\begin{array}{l}\fsh{\langle [\lacall{\pvar{F}(\vN, \vI, \vS)}],} \\ \qquad
         \fsh{\venv, \funcsct \rangle}\end{array}\)}
\qquad \myif \fsh{1 \leq \vI \leq \vN}.
\]

The first circularity represents the expansion phase of the left-hand
side program, while the second circularity corresponds to the
contraction phase, as explained in the introduction. Our prover can
establish using the circularities above that \(\fsh{f \prec F}\) and that \(\fsh{f
\preceq F}\) under the constraint \(\fsh{\vN \geq 0}\). By \(\fsh{f}\) we formally mean
\(\cfgimp{[\lacall{\pvar{f}(\vN)}]}{\venv}{\funcsct}\) and by \(\fsh{F}\)
we formally mean \(\cfgimp{[\lacall{\pvar{F}(\vN, 0,
      0)}]}{\venv}{\funcsct}\), as introduced above (we use this
shorthand in the following two examples as well). By reversing the lhs
and rhs of the circularities, our algorithm can also show \(\fsh{F \preceq
f}\) under the constraint \(\fsh{\vN \geq 0}\). However, it cannot show \(\fsh{F \prec
f}\) under the same constraint, because the circularity for the second
phase of \(\fsh{f}\) cannot be established due to lack of progress on the
left-hand side. {\bf Therefore, our tool establishes partial
  equivalence of \(\fsh{f}\) and \(\fsh{F}\) and half of what is necessary for full
  equivalence.} As future work, in order to enable the complete proof
of full equivalence, we will add termination measures to the proof
system as in~\cite{BuruianaEPTCS2019} -- the termination measure for
the second phase of \(\fsh{f}\) will enable \(\fsh{F \prec f}\) to be proven. The
names of the four files corresponding to this example are of the form
\texttt{example1*.rmt}.

\paragraph{Example 2: recursive functions in the presence of a bounded
  stack.} In this example, we work in the language \IMP{2}, which has
a bounded stack of length \(10\).. The equivalence between \(\fsh{f}\) and \(\fsh{F}\)
does not hold in \IMP{2} (because for a sufficiently high input, \(\fsh{F}\)
will work as expected, while \(\fsh{f}\) will crash with a stack
overflow). Our tool correctly fails to prove, in the operational
semantics \IMP{2}, any of the cases \(\fsh{f \prec F}\), \(\fsh{F \prec f}\), \(\fsh{f \prec
F}\), \(\fsh{F \prec f}\) (under the constraint \(\fsh{\vN \geq 0)}\). We used the same
base cases and circularities as above. The names of the four files
corresponding to this examples are of the form \texttt{example2*.rmt}.

\paragraph{Example 3: two different semantics.} As explained in
Section~\ref{sec:proofsystem}, our algorithm for functional
equivalence works even for two programs written in different
languages. We exploit this to prove that \(\fsh{f}\), interpreted in \IMP{1},
is partially equivalent to \(\fsh{F}\), interpreted in \IMP{2}. The
equivalence works because \(\fsh{F}\) uses constant stack space, so it works
properly in \IMP{2}. The only simulation which the tools fails to
prove in this setting is \(\fsh{F \prec f}\), for the same reasons as
above. The names of the four files corresponding to this example are
of the form \texttt{example3*.rmt}.

\paragraph{Example 4: imperative and functional style.} This example
shows that our proof method allows proofs of structurally different
programs. We show that a recursive function is equivalent to a while
loop (both computing the sum of the first \(\vN\) numbers). We prove:
\[
  \mbox{\(\begin{array}{l}\fsh{\langle \lacall{\pvar{f}(\vN)}} \\ \qquad \fsh{\lisymb [] , \venv, \funcsct
           \rangle} \end{array}\)}
\fsh{\preceq}
\mbox{\(\begin{array}{l} \fsh{\langle \pvar{i} = 0;} \\ \fsh{\pvar{s} = 0;} \\
         \fsh{\mathit{while} (\pvar{i} \leq \vN)} \\
         \qquad \fsh{\laassign{\pvar{s}}{\plus{\pvar{s}}{\pvar{i}}};} \\
         \qquad \fsh{\laassign{\pvar{i}}{\plus{\pvar{i}}{1}}} \\ \;\;\; \fsh{\lisymb [],
         \venv, \funcsct \rangle} \end{array}\)} \myif \fsh{0 \leq \vN}
   \]
\noindent
and vice-versa (rhs partially simulated by lhs and lhs partially
simulated by rhs), where \(\funcsct = \{\ \fsh{\pvar{f} \mapsto \lambda
\pvar{x} . \laite{0 \leq \pvar{x}}{x + \lacall{\pvar{f}(\pvar{x} -
    1)}}{0}\ } \}\). We can also show full simulation for one of the
direction (the other direction fails for the same reason as in the
first example). The fact that both programs take the same input is
represented by the integer variable \(\vN\) appearing in both sides. The
same variable \(\venv\) is also used on both sides, meaning that
the two programs start with the same set of values associated to the
global variables (however, as the programs do not depend on the value
of the globals, the equivalence proof would also work when starting
with two different environments).

An interesting observation in this example is that the first program
is written in a functional style and therefore it will reduce to a
value without modifying the environment. The imperative program will
hold the result in the environment, associated to the program
identifier \(\pvar{s}\). Additionally, the second program will modify in
the environment the variable \(\pvar{i}\), whose value should not
considered to be part of the result of the program. Therefore, the set
of base cases we use is \[B = \{
\cfgimp{[\vx]}{\venvsub{1}}{\funcsct} \fsh{\prec}
\cfgimp{[]}{\venvsub{2}}{\funcsct}\myif \fsh{\vx =
\blookup(\pvar{s}, \venvsub{2})} \}.\] The files
corresponding to this example are of the form \texttt{example4*.rmt}.

\paragraph{Example 5: loop unswitching}

We prove:
\[
  \mbox{\(\begin{array}{l} \fsh{\langle \laassign{\pvar{a}}{\vA};} \\
           \fsh{\laassign{\pvar{y}}{\vY};} \\
           \fsh{if \; even(\pvar{a}) \; then} \\
           \;\;\;\; \fsh{while (\pvar{y} \leq \vN)} \\
           \;\;\;\; \;\;\;\; \fsh{\laassign{\pvar{y}}{\plus{\pvar{y}}{1}}} \\
           \fsh{else} \\
           \;\;\;\; \fsh{while (\pvar{y} \leq \vN)} \\
           \;\;\;\; \;\;\;\; \fsh{ \laassign{\pvar{y}}{\plus{\pvar{y}}{2}}} \\
           \fsh{\lisymb [], \venv, \funcsct \rangle}\end{array}\)} \fsh{\prec}
       \mbox{\(\begin{array}{l} \fsh{\langle \laassign{\pvar{a}}{\vA};} \\
                \fsh{\laassign{\pvar{y}}{\vY};} \\
                \fsh{while (\pvar{y} \leq \vN)} \\
                \;\;\;\; \fsh{if \; even(\pvar{a})\ then} \; \\
                \;\;\;\; \;\;\;\; \fsh{\laassign{\pvar{y}}{\plus{\pvar{y}}{1}}} \\
                \;\;\;\; \fsh{else} \\
                \;\;\;\; \;\;\;\; \fsh{\laassign{\pvar{y}}{\plus{\pvar{y}}{2}}} \\
                \fsh{\lisymb [], \venv, \funcsct \rangle}\end{array}\)}
          \]
\noindent
and vice-versa, where \(\funcsct\) is arbitrary. For this, we require two circularities:
\[
  \mbox{\(\begin{array}{l} \fsh{\langle while (\pvar{y} \leq \vN)} \\
           \qquad \fsh{\laassign{\pvar{y}}{\plus{\pvar{y}}{1}}} \\
           \fsh{\lisymb [], \venv, \funcsct \rangle}\end{array}\)}
  \fsh{\prec}
  \mbox{\(\begin{array}{l}\fsh{\langle while (\pvar{y} \leq \vN)} \\
           \;\;\;\; \fsh{if \; even(\pvar{a})\ then} \; \\
           \;\;\;\; \;\;\;\; \fsh{\laassign{\pvar{y}}{\plus{\pvar{y}}{1}}} \\
           \;\;\;\; \fsh{else} \\
           \;\;\;\; \;\;\;\; \fsh{\laassign{\pvar{y}}{\plus{\pvar{y}}{2}}} \\
           \fsh{\lisymb [], \venv, \funcsct \rangle}\end{array}\)}\hspace{-0.2cm} \myif \mbox{\( \fsh{even(\blookup(\venv, \pvar{a}))} \)}
\]
and
\[
  \mbox{\(\begin{array}{l} \fsh{\langle while (\pvar{y} \leq \vN)} \\
           \qquad \fsh{\laassign{\pvar{y}}{\plus{\pvar{y}}{2}}} \\
           \fsh{\lisymb [], \venv, \funcsct \rangle}\end{array}\)}
  \fsh{\prec}\hspace{-0.2cm}
  \mbox{\(\begin{array}{l} \fsh{\langle while (\pvar{y} \leq \vN)} \\
           \;\;\;\; \fsh{if \; even(\pvar{a})\ then} \; \\
           \;\;\;\; \;\;\;\; \fsh{\laassign{\pvar{y}}{\plus{\pvar{y}}{1}}} \\
           \;\;\;\; \fsh{else} \\
           \;\;\;\; \;\;\;\; \fsh{\laassign{\pvar{y}}{\plus{\pvar{y}}{2}}} \\
           \fsh{\lisymb [], \venv, \funcsct \rangle}\end{array}\)}\hspace{-0.2cm} \myif \mbox{\(\fsh{\lnot even(\blookup(\venv, \pvar{a}))}.\)}
     \]
\noindent
We use the base cases
\[B = \{ \cfgimp{[]}{\venvsub{1}}{\funcsct} \fsh{\prec} \cfgimp{[]}{\venvsub{2}}{\funcsct} \myif \fsh{\blookup(\venvsub{1}, \pvar{y}) = \blookup(\venvsub{2}, \pvar{y})} \}.\] That is, we consider two
terminal configurations to be equivalent when the corresponding
environments map the program variable \(\pvar{y}\) to the same
value. The names of the files corresponding to this example are of the
form \texttt{example5*.rmt}.

\end{document}